\documentclass[a4paper,fleqn]{cas-sc}

\usepackage[authoryear,longnamesfirst]{natbib} 

\def\tsc#1{\csdef{#1}{\textsc{\lowercase{#1}}\xspace}}
\tsc{WGM}
\tsc{QE}
\tsc{EP}
\tsc{PMS}
\tsc{BEC}
\tsc{DE}

\usepackage[displaymath, mathlines]{lineno}

\usepackage{amsmath,amssymb,amsfonts,amsthm}
\usepackage{bm}
\usepackage{booktabs}
\usepackage{graphicx}
\usepackage{mathtools}
\usepackage{microtype}
\usepackage{hyperref}
\usepackage{fullpage}
\usepackage{tikz}
\usepackage[nameinlink,capitalise]{cleveref}

\usepackage{tabularx}
\usepackage{array}
\usepackage{svg}
\usepackage{wrapfig}

\definecolor{python1}{HTML}{1f77b4}
\definecolor{python2}{HTML}{ff7f0e}
\definecolor{python3}{HTML}{2ca02c}
\definecolor{python4}{HTML}{d62728}
\definecolor{python5}{HTML}{9467bd}
\definecolor{python6}{HTML}{8c564b}
\definecolor{python7}{HTML}{e377c2}
\definecolor{python8}{HTML}{7f7f7f}
\definecolor{python9}{HTML}{bcbd22}
\definecolor{python10}{HTML}{17becf}
\definecolor{designcolor}{HTML}{1f77b4}

\colorlet{tablehead}{python1!12}
\colorlet{tablerelation}{python1!5}
\colorlet{tableline}{python8!35}
\colorlet{tabletext}{black!88}

\newtheorem{theorem}{Theorem}[section]
\newtheorem{proposition}[theorem]{Proposition}
\newtheorem{lemma}[theorem]{Lemma}
\newtheorem{corollary}[theorem]{Corollary}

\theoremstyle{definition}
\newtheorem{definition}[theorem]{Definition}
\newtheorem{assumption}[theorem]{Assumption}

\theoremstyle{remark}
\newtheorem{remark}[theorem]{Remark}

\crefname{theorem}{theorem}{theorems}
\Crefname{theorem}{Theorem}{Theorems}

\crefname{proposition}{proposition}{propositions}
\Crefname{proposition}{Proposition}{Propositions}

\crefname{lemma}{lemma}{lemmas}
\Crefname{lemma}{Lemma}{Lemmas}

\crefname{corollary}{corollary}{corollaries}
\Crefname{corollary}{Corollary}{Corollaries}

\crefname{definition}{definition}{definitions}
\Crefname{definition}{Definition}{Definitions}

\crefname{assumption}{assumption}{assumptions}
\Crefname{assumption}{Assumption}{Assumptions}

\crefname{remark}{remark}{remarks}
\Crefname{remark}{Remark}{Remarks}

\definecolor{michaelgreen}{rgb}{0,0.5,0}

\renewcommand{\d}{\mathop{}\!\mathrm{d}}

\begin{document}
\let\WriteBookmarks\relax
\def\floatpagepagefraction{1}
\def\textpagefraction{.001}
\shorttitle{Stochastic Latent Twin}
\shortauthors{M. Chung et~al.}

\title [mode = title]{Stochastic Latent Twins: Learning Stochastic Solution Operators}                      



\author[1]{Matthias Chung}[orcid=0000-0001-7822-4539]
\cormark[1]
\ead{matthias.chung@emory.edu}
\ead[url]{https://matthiaschung.github.io/}

\credit{Conceptualization of this study, Methodology, Software}

\affiliation[1]{organization={Department of Mathematics, Emory University},
                addressline={400 Dowman Drive}, 
                city={Atlanta},
                postcode={30322}, 
                state={Georgia},
                country={USA}}

\author[1]{Jack Michael Solomon}[orcid=0009-0001-0822-9438]
\ead{jack.michael.solomon@emory.edu}

\credit{Data curation, Writing - Original draft preparation}

\author[2]{Deepanshu Verma}[orcid=0000-0002-0904-5470]

\ead{dverma@clemson.edu}
\ead[URL]{https://dpnshvrm.github.io/}

\affiliation[3]{organization={School of Mathematical and Statistical Sciences, Clemson University},
                addressline={220 Parkway Drive}, 
                city={Clemson},
                postcode={29634}, 
                state={South Carolina}, 
                country={USA}}

\begin{abstract}
Scientific prediction under uncertainty requires more than a single future state: it requires the distribution of possible outcomes. We show that this distribution-valued evolution can be learned directly as a stochastic solution operator between arbitrary source and target times. We introduce \emph{Stochastic Latent Twins}, which represent the physical state through a deterministic encoder--decoder while placing stochasticity entirely in a probability-valued latent evolution map. This separation makes the stochastic transition law independent of the state representation and enables direct probabilistic prediction without recursive time stepping.

We establish an end-to-end Wasserstein approximation theory that separates representation and latent-transition errors, derive stability guarantees, and develop a distributional theory of physics conformity with latent-to-physical transfer bounds and conformity by construction. The resulting framework can be
learned directly from observed transitions and supports forward prediction, backward conditional evolution, one-shot generation, and physics-conforming stochastic dynamics within the same formulation. Numerical experiments ranging from linear and non-Gaussian SDEs to stochastic Lorenz dynamics, generative modeling, a stochastic Hamiltonian system, and a two-dimensional stochastic Burgers equation verify the theory, expose the distinct roles of state representation and transition-family expressivity, compare favorably with recent stochastic and one-step generative methods, and demonstrate exact physical conservation by construction. More broadly, Stochastic Latent Twins provide a step toward scientific surrogates that represent not only the evolution of a system, but also its uncertainty and physical structure, opening new opportunities across uncertainty quantification, data assimilation, inverse problems, control, and digital twins.
\end{abstract}



\maketitle

\section{Introduction}\label{sec:introduction}

Stochastic differential equations (SDEs) provide a fundamental model for dynamical systems subject to random forcing and unresolved variability, with applications ranging from physics and engineering to biology, climate, and finance; see, e.g., \citep{oksendal2003stochastic,kloeden1992numerical}. Consider the classical It\^o SDE
\begin{equation}
    \mathrm{d}X_t = a(X_t,t)\,\mathrm{d}t + G(X_t,t)\,\mathrm{d}W_t,
    \label{eq:intro-sde}
\end{equation}
where $X_t\in\mathcal X$ denotes the random state and $W_t$ is a Wiener process. A realization of \Cref{eq:intro-sde} defines a sample path $t\mapsto x_t$. In contrast to a deterministic dynamical system, prescribing the state $X_s=x_s$ at time $s$ does not in general determine a unique state at a time $t\neq s$. Instead, it determines a conditional distribution, which we write as
\begin{equation}\label{eq:intro-stochastic-solution-operator}
    \Phi_{s,t}(x_s) := \operatorname{Dist}\left(X_t\mid X_s=x_s\right),
    \qquad
    \Phi_{s,t}:\mathcal X\longrightarrow\mathcal P(\mathcal X),
\end{equation}
where $\mathcal P(\mathcal X)$ denotes the space of Borel probability measures on $\mathcal X$. We call $\Phi_{s,t}$ the stochastic solution operator.

Numerical simulation of an SDE typically requires a time discretization together with the generation of random increments, and accurate approximation of sample paths or their statistics may require many time steps and realizations \citep{kloeden1992numerical}. This burden motivates learning $\Phi_{s,t}$ directly between arbitrary time points $s$ and $t$, avoiding repeated small-step simulation at inference. Data-driven surrogate models provide one such approach \citep{hasan2022identifying,chen2024learning,kiyohara2026neural}.

We introduce Stochastic Latent Twins (SLTs), a latent neural-network surrogate framework for $\Phi_{s,t}$, based on an encode--evolve--decode representation in which the encoder and decoder are deterministic and the latent evolution is probability-valued. Building on the Latent Twin framework \citep{chung2026latent,chung2027physics,verma2026latent}, we replace latent evolution from point-valued latent states to conditional latent distributions between arbitrary time points, while retaining the separation between state representation and temporal evolution.

With encoder $e:\mathcal X\rightarrow\mathcal Z$ and decoder $d:\mathcal Z\rightarrow\mathcal X$, a state $x_s$ is encoded as $z_s=e(x_s)$ and evolved according to
\begin{equation}\label{eq:intro-stochastic-latent-map}
    m(z_s,s,t) = \operatorname{Dist}\left(Z_t\mid Z_s=z_s\right) \in\mathcal P(\mathcal Z).
\end{equation}
The predicted latent distribution is then mapped back to physical space through $d$. A precise definition is given in \Cref{def:stochastic-latent-twin} and illustrated in \Cref{fig:stochasticLatentTwinDS}.

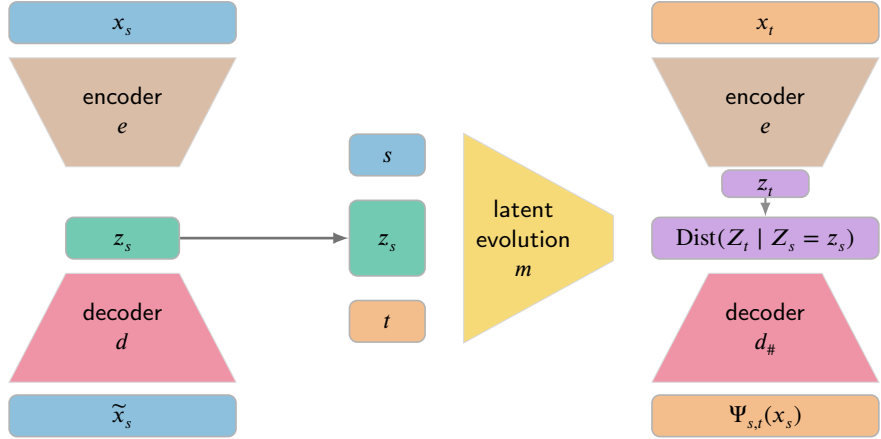
\begin{figure}
\begin{center}
    \begin{tikzpicture}[
box/.style={
    rounded corners=3pt,
    draw=black!30,
    line width=0.6pt,
    font=\small,
    align=center,
    inner sep=2pt
},
statebox/.style={
    box,
    minimum height=0.55cm,
    text height=1.7ex,
    text depth=0.45ex
},
latentbox/.style={
    box,
    minimum height=0.55cm,
    text height=1.7ex,
    text depth=0.45ex
},
smallbox/.style={
    box,
    minimum height=0.36cm,
    text height=1.5ex,
    text depth=0.35ex,
    inner sep=1pt
},
trapez/.style={
    draw=black!15,
    line width=0.4pt
},
flow/.style={
    -latex,
    line width=1pt,
    black!55
},
sampleflow/.style={
    -latex,
    dashed,
    line width=0.8pt,
    black!45
}
]

\definecolor{stateColS}{RGB}{0,114,178}    
\definecolor{stateColT}{RGB}{230,111,0}    
\definecolor{latentColS}{RGB}{0,158,115}   
\definecolor{latentColT}{RGB}{153,79,201}  
\definecolor{encoderCol}{RGB}{196,154,120} 
\definecolor{decoderCol}{RGB}{230,90,120}  
\definecolor{evolCol}{RGB}{240,193,30}     

\pgfmathsetmacro{\leftX}{0.0}
\pgfmathsetmacro{\midX}{3.5}
\pgfmathsetmacro{\rightX}{8.5}

\pgfmathsetmacro{\stateY}{0.0}
\pgfmathsetmacro{\latentY}{-2.85}
\pgfmathsetmacro{\outputY}{-5.20}

\pgfmathsetmacro{\stateHalf}{1.50}
\pgfmathsetmacro{\latentHalf}{0.75}

\pgfmathsetmacro{\funnelTopGap}{0.18}
\pgfmathsetmacro{\funnelDepth}{1.62}

\newcommand{\funnel}[6]{%
    \draw[fill=#4!65, trapez]
        ([xshift=-#2 cm,yshift=-\funnelTopGap cm]#1.south) --
        ([xshift=-#3 cm,yshift=-\funnelDepth cm]#1.south) --
        ([xshift=#3 cm,yshift=-\funnelDepth cm]#1.south) --
        ([xshift=#2 cm,yshift=-\funnelTopGap cm]#1.south) -- cycle
        node (#5) at ([yshift=-0.91cm]#1.south) {#6};
}

\node[
    statebox,
    fill=stateColS!45,
    minimum width=3.0cm
]
    (stateS) at (\leftX,\stateY)
    {$x_s$};

\funnel{stateS}{\stateHalf}{\latentHalf}{encoderCol}{encoderS}
    {$\begin{matrix}\mbox{encoder}\\[-0.0ex]e\end{matrix}$}

\node[
    latentbox,
    fill=latentColS!55,
    minimum width=1.5cm
]
    (latentS) at (\leftX,\latentY)
    {$z_s$};

\funnel{latentS}{\latentHalf}{\stateHalf}{decoderCol}{decoderS}
    {$\begin{matrix}\mbox{decoder}\\[-0.0ex]d\end{matrix}$}

\node[
    statebox,
    fill=stateColS!45,
    minimum width=3.0cm
]
    (reconstructionS) at (\leftX,\outputY)
    {$\widetilde x_s$};

\node[
    statebox,
    fill=stateColT!45,
    minimum width=3.0cm
]
    (stateT) at (\rightX,\stateY)
    {$x_t$};

\funnel{stateT}{\stateHalf}{\latentHalf}{encoderCol}{encoderT}
    {$\begin{matrix}\mbox{encoder}\\[-0.0ex]e\end{matrix}$}

\node[
    smallbox,
    fill=latentColT!50,
    minimum width=1.15cm
]
    (latentTObserved) at (\rightX,-2.13)
    {$z_t$};

\node[
    latentbox,
    fill=latentColT!50,
    minimum width=3.0cm
]
    (latentDistribution) at (\rightX,\latentY)
    {$\operatorname{Dist}(Z_t\mid Z_s=z_s)$};

\draw[sampleflow]
    (latentTObserved.south) --
    (latentDistribution.north);

\funnel{latentDistribution}{\latentHalf}{\stateHalf}{decoderCol}{decoderT}
    {$\begin{matrix}\mbox{decoder}\\[-0.0ex]d_\#\end{matrix}$}

\node[
    statebox,
    fill=stateColT!45,
    minimum width=3.0cm
]
    (predictedLaw) at (\rightX,\outputY)
    {$\Psi_{s,t}(x_s)$};

\node[
    statebox,
    fill=stateColS!45,
    minimum width=1.0cm
]
    (timeS) at (\midX,-1.75)
    {$s$};

\node[
    latentbox,
    fill=latentColS!55,
    minimum width=1.0cm,
    minimum height=1.0cm
]
    (latentInput) at (\midX,\latentY)
    {$z_s$};

\node[
    statebox,
    fill=stateColT!45,
    minimum width=1.0cm
]
    (timeT) at (\midX,-3.95)
    {$t$};

\draw[flow]
    (latentS.east) --
    (latentInput.west);

\draw[
    fill=evolCol!60,
    trapez
]
    ([xshift=0.50cm]timeS.north east) --
    ([xshift=-0.50cm,yshift=0.30cm]latentDistribution.west) --
    ([xshift=-0.50cm,yshift=-0.30cm]latentDistribution.west) --
    ([xshift=0.50cm]timeT.south east) -- cycle;

\node[
    align=center,
    font=\small,
    text width=2.0cm
]
    at (5.30,\latentY)
    {latent evolution\\$m$};

\end{tikzpicture}
\end{center}
\caption{\emph{Schematic of a Stochastic Latent Twin (SLT) for stochastic dynamical systems.} The encoder maps the observed states $x_s$ and $x_t$ to their latent representations $z_s=e(x_s)$ and $z_t=e(x_t)$, respectively. Given $(z_s,s,t)$, the latent evolution map $m$ returns the transition law $m(z_s,s,t)=\operatorname{Dist}(Z_t\mid Z_s=z_s)$, which represents the  distribution of the latent state at time $t$ conditioned on $z_s$ at time $s$. The encoded state $z_t$ is one realization of this distribution corresponding to the observed state $x_t$. The decoder maps latent states back to physical space, with $\widetilde x_s=d(e(x_s))$ denoting the reconstruction of $x_s$. Applying the same decoder to the predicted latent distribution gives its pushforward in physical space, denoted by $d_\#$, so that the Stochastic Latent Twin $\Psi_{s,t}(x_s)=d_\#m(z_s,s,t)=d_\#m(e(x_s),s,t)$. Thus, the encoder--decoder pair defines the state representation, while the stochastic temporal evolution is contained entirely in $m$.}
\label{fig:stochasticLatentTwinDS}
\end{figure}

\paragraph{Contribution.} Stochastic Latent Twins can be trained directly from observed state transitions and evaluated at arbitrary, including irregularly spaced, time points without recursive small-step propagation. The main contribution of this work is the formulation and analysis of {Stochastic Latent Twins}, a latent solution-operator framework in which stochastic evolution is learned directly as a conditional probability distribution. We develop a Wasserstein approximation theory that connects latent stochastic evolution to the resulting physical-space solution law and introduce a distributional theory of physics conformity for stochastic dynamics, including conservation, dissipation, and latent-to-physical transfer of physical structure.

\paragraph{Organization.} The remainder of the paper is organized as follows. \Cref{sec:background} places Stochastic Latent Twins within the broader development of learned stochastic dynamics and distribution-valued solution operators, including neural and latent SDEs, stochastic flow maps, probabilistic operator learning, and stochastic transfer operators. \Cref{sec:stochastic-latent-twins} introduces the Stochastic Latent Twin formulation. We then develop its Wasserstein approximation theory (\Cref{subsec:end-to-end-approximation}), stability properties (\Cref{subsec:slt-stability}), and physics-conforming formulation (\Cref{subsec:physics}).

\Cref{sec:numerics} examines the framework across problems ranging from controlled SDEs (\Cref{sec:linear_sde_experiment,subsec:cir}) and fully data-driven stochastic dynamics (\Cref{subsec:lorenz}) to one-shot generative modeling (\Cref{subsec:swissroll}), physics conforming stochastic dynamics
(\Cref{subsec:pendulum}), and a two-dimensional stochastic Burgers partial differential equation (SPDE) (\Cref{subsec:burgers}), with direct comparisons against Neural Stochastic Flows and recent one-step generative methods. The appendices provide supporting theoretical derivations (\Cref{app:related-methods,app:linear-low-rank}) and detailed numerical setups, training and evaluation protocols, and additional diagnostics (\Cref{app:linear-numerics,app:cir-numerics,app:lorenz-numerics,app:swissroll-numerics,app:pendulum-numerics,app:burger-numerics}) to facilitate reproducibility.

\section{Related work}\label{sec:background}

Surrogate models for stochastic dynamics can be formulated in fundamentally different ways, depending on the mathematical object chosen to represent the evolution. These include local stochastic dynamics, conditional transition laws between source and target times, probability-density evolution, and stochastic operators.

One class of approaches learns the local stochastic dynamics. Neural SDEs parameterize the drift and diffusion coefficients \citep{li2020scalable,kidger2021neuralsdes}, while latent SDE approaches place such dynamics in a learned state space, often within variational or state-space formulations \citep{hasan2022identifying,course2023amortized,bartosh2025sdematching}. Prediction at a later time is then obtained by propagating the learned stochastic dynamics.

Rather than learning an infinitesimal model and propagating it forward, another class of methods represents stochastic evolution directly over finite time intervals. Stochastic flow-map learning constructs stochastic maps between states separated in time \citep{chen2024learning}, while transition-density approaches learn conditional transition densities or the associated Fokker--Planck evolution \citep{su2025transition,zeng2026transition}.  These methods share the direct finite-interval perspective of Stochastic Latent Twins, but differ in how stochastic evolution is represented and coupled to the state representation. Stochastic Latent Twins combine direct conditional-law learning with a separate latent representation, providing the basis for the approximation and physics-conformity developments of this work.

Neural Stochastic Flows (NSFs) \citep{kiyohara2026neural} provide the closest comparison to the present formulation. They learn conditional transition laws between arbitrary time points using conditional normalizing flows and encode structural properties of stochastic flows, including identity and Chapman--Kolmogorov consistency. Their latent extension is formulated through a variational state-space model. Stochastic Latent Twins instead retain a deterministic encoder--decoder representation and place stochasticity in the latent evolution itself, without prescribing a particular family for the conditional transition law. The theoretical developments also differ in scope: Neural Stochastic Flows focus primarily on structural consistency of the learned stochastic flow, whereas we develop end-to-end Wasserstein approximation and stability results together with a distributional theory of physics conformity and latent-to-physical transfer. Numerical comparisons are given in \Cref{subsec:lorenz,subsec:burgers}.

Probabilistic and stochastic neural operators approach uncertainty from a function-space perspective, learning distribution-valued or stochastic maps between input and output fields \citep{bulte2025probabilistic,bahmani2025neuralchaos,shi2026expanding}. Stochastic transfer- and Koopman-operator methods take a related operator viewpoint, focusing on the evolution of observables or learned representations adapted to the underlying stochastic dynamics \citep{kostic2024learning,ke2025learning}. Stochastic Latent Twins take a different route, centering the formulation on state-conditioned transition laws between source and target times while using the latent space to represent the evolving state.

A further connection arises with generative modeling. Diffusion and score-based models and flow matching construct transport between probability distributions \citep{song2021score,lipman2023flow}, while recent one-step methods seek to replace iterative generation by increasingly direct transport or solution maps \citep{frans2025one,geng2026mean,luo2026soflow}. These approaches are primarily designed to transform one distribution into another, whereas Stochastic Latent Twins learn state-conditioned transition laws. When the source state is itself distributed, however, these conditional laws induce a corresponding distribution-to-distribution evolution. This connection motivates the one-shot generative experiment in \Cref{subsec:swissroll}.

Stochastic Latent Twins also build directly on the deterministic Latent Twin,  Physics-Conforming Latent Twin, and Latent Twin Operators frameworks \citep{chung2026latent,chung2027physics, verma2026latent}. The present work extends the latent solution-operator viewpoint from point-valued to probability-valued evolution and generalizes physics conformity to stochastic solution laws, including distributional conservation and dissipation. A more detailed comparison with related stochastic-learning frameworks is provided in \Cref{app:related-methods}.

\section{Stochastic Latent Twins (SLTs)}\label{sec:stochastic-latent-twins}

The central objective is to approximate the stochastic solution operator $\Phi_{s,t}$ defined in \cref{eq:intro-stochastic-solution-operator}, which maps an initial state $x_s$ to the transition law $\Phi_{s,t}(x_s)$, the distribution of $X_t$ given $X_s=x_s$. In standard stochastic-process terminology, $\Phi_{s,t}$ is the transition kernel of the process \citep{ethier2009markov}. We use solution-operator terminology to emphasize the connection with deterministic solution operators and the Latent Twin framework.

\begin{definition}[Stochastic Latent Twin] \label{def:stochastic-latent-twin}
A \emph{Stochastic Latent Twin} of $\Phi_{s,t}$ is a surrogate stochastic solution operator
\begin{equation}\label{eq:slt-definition}
    \Psi_{s,t}(x_s) := d_\# m(e(x_s),s,t),
\end{equation}
defined through a latent space $\mathcal Z$ and
\begin{enumerate}
    \item[(i)] an encoder $e:\mathcal X\rightarrow\mathcal Z$,
    \item[(ii)] a \emph{latent evolution map} $ m:\mathcal Z\times[0,T]^2\rightarrow\mathcal P(\mathcal Z)$, and 
    \item[(iii)] a decoder $d:\mathcal Z\rightarrow\mathcal X$, inducing
    $d_\#:\mathcal P(\mathcal Z)\rightarrow\mathcal P(\mathcal X)$.
\end{enumerate}
\end{definition}
For a given $(z_s,s,t)$, the value $m(z_s,s,t)=\operatorname{Dist}(Z_t\mid Z_s=z_s)$ is the \emph{latent transition law}. The formulation does not prescribe a probability family for these laws; in this work, we use Gaussian, Gaussian-mixture, and noncentral-$\chi^2$ transition families. When referring to trainable realizations, $\theta$ denotes the parameters of the encoder, latent evolution map, and decoder collectively. We suppress it and write $e$, $m$, and $d$ unless the dependence matters.

\begin{remark}[Deterministic Latent Twins as a special case]
\label{rem:deterministic-lt}
The Stochastic Latent Twin framework naturally contains deterministic Latent Twins as a special case. Indeed, suppose that the latent evolution map returns the Dirac transition law, $m(z_s,s,t)=\delta_{M_{s,t}(z_s)}$, for some deterministic latent evolution map $M_{s,t}:\mathcal Z\rightarrow\mathcal Z$. Then $Z_t=M_{s,t}(z_s)$ almost surely, and the Stochastic Latent Twin reduces to $\Psi_{s,t}(x_s) = d_\#\,\delta_{M_{s,t}(e(x_s))} = \delta_{d(M_{s,t}(e(x_s)))}$. The stochastic formulation therefore extends point-valued solution operators to probability-valued ones, with direct access to predictive distributions and their statistics.

\end{remark}

Approximating the stochastic solution operator requires both faithful state representation and an accurate stochastic evolution, 
\begin{equation}\label{eq:slt-approximation-goals}
    d\circ e\approx\operatorname{id}_{\mathcal X},
    \qquad
    \Psi_{s,t}(x_s)
    =
    d_\#m(e(x_s),s,t)
    \approx
    \Phi_{s,t}(x_s),
\end{equation}
and the approximation error separates accordingly into a state-representation part and a stochastic-evolution part, as quantified in \Cref{sec:slt-approximation}.

\begin{remark}[Temporal consistency] \label{rem:temporal-consistency}
When the latent transition law admits a density, Markovian evolution satisfies the Chapman--Kolmogorov composition property
\begin{equation} \label{eq:latent-chapman-kolmogorov}
    p(z_t\mid z_s,s,t) = \int_{\mathcal Z} p(z_t\mid z_r,r,t)\, p(z_r\mid z_s,s,r)\, \mathrm dz_r, \qquad \text{for }s<r<t,
\end{equation}
that is, direct evolution from $s$ to $t$ agrees with evolution through an intermediate time after averaging over all possible intermediate states. This is the stochastic analogue of the composition property of deterministic solution operators, reducing to the familiar semigroup property in the time-homogeneous case \citep{ethier2009markov}. Chapman--Kolmogorov consistency provides a natural temporal consistency condition and may be promoted during training, as is also done in Neural Stochastic Flows \citep{kiyohara2026neural}. Since encoding need not preserve the Markov property when relevant state information is discarded, we regard \cref{eq:latent-chapman-kolmogorov} as an optional structural property rather than an assumption of the approximation theory.
\end{remark}

\section{Approximation and Structural Properties}
\label{sec:slt-approximation}

Since both the stochastic solution and its Stochastic Latent Twin approximation are probability distributions, their discrepancy is naturally measured in probability space. We denote by $\mathcal P_p(\mathcal X)\subset\mathcal P(\mathcal X)$ with $p\in [1,\infty)$, the space of probability measures on $\mathcal X$ with finite $p$-th moment. We use the $p$-Wasserstein distance $W_p$ to measure discrepancies between probability distributions; see, e.g., \citep{villani2009optimal}.

\begin{assumption}[Stochastic dynamics] \label{ass:stochastic-dynamics}
For every $x_s\in\mathcal X$ and $(s,t)\in\mathcal T$, the conditional stochastic solution operator is well defined and satisfies
\begin{equation}\label{eq:transition-p-moment}
    \Phi_{s,t}(x_s)\in\mathcal P_p(\mathcal X), \qquad\text{equivalently}\qquad \int_{\mathcal X} \|x\|_{\mathcal X}^p\, \Phi_{s,t}(x_s)(\mathrm dx) <\infty.
\end{equation}
\end{assumption}
Here we use $\Phi_{s, t}(x_s)(\d x)$ to denote integration against the conditional law which is not required to admit a density.  Moreover, no boundedness or compact-support condition is imposed on the stochastic state. The assumption only requires the conditional law to have a finite $p$-th moment. Transition laws with unbounded support, including Gaussian laws, are therefore admissible. This places the conditional laws in the natural domain of the $p$-Wasserstein distance used below.

The choice of $\mathcal T$ determines the admissible time orientation. For standard forward evolution, one may take $\mathcal T \subseteq \{(s,t)\in[0,T]^2:s\leq t\}$.  More generally, $\mathcal T$ may contain any time pairs for which the corresponding conditional solution laws are well defined, including backward conditional operators when such laws exist; see \Cref{subsec:forward-backward}. The approximation theory below is formulated directly in terms of these conditional probability laws and therefore does not require a prescribed time orientation.

We next impose the corresponding approximation properties on the state representation.

\begin{assumption}[Encoder and decoder]
\label{ass:encoder-decoder}
The encoder $e:\mathcal X\to\mathcal Z$ and decoder $d:\mathcal Z\to\mathcal X$ are Lipschitz continuous:
\begin{align}
    \|e(x)-e(y)\|_{\mathcal Z}
    &\leq
    L_e\|x-y\|_{\mathcal X},
    \qquad x,y\in\mathcal X,
    \label{eq:encoder-lipschitz}\\
    \|d(z)-d(w)\|_{\mathcal X}
    &\leq
    L_d\|z-w\|_{\mathcal Z},
    \qquad z,w\in\mathcal Z,
    \label{eq:decoder-lipschitz}
\end{align}
for constants $L_e,L_d>0$.
\end{assumption}
The encoder--decoder pair is assumed to be Lipschitz continuous, while its representation accuracy is quantified by the following reconstruction error.
\begin{definition}[Distribution-weighted reconstruction error]
\label{def:stochastic-ae-error}
For $x_s\in\mathcal X$ and $(s,t)\in\mathcal T$, define
\begin{equation}
\label{eq:stochastic-ae-error}
    \varepsilon_{\rm ae}(x_s;s,t)
    :=
    \left(
        \int_{\mathcal X}
        \|d(e(x))-x\|_{\mathcal X}^{p}\,
        \Phi_{s,t}(x_s)(\mathrm dx)
    \right)^{1/p}.
\end{equation}
\end{definition}
Since
\[
    X_t\mid X_s=x_s\sim\Phi_{s,t}(x_s),
\]
the reconstruction error can equivalently be written as
\[
    \varepsilon_{\rm ae}(x_s;s,t)
    =
    \left(
        \mathbb E\left[
            \|d(e(X_t))-X_t\|_{\mathcal X}^{p}
            \,\middle|\,
            X_s=x_s
        \right]
    \right)^{1/p}.
\]
Thus, $\varepsilon_{\rm ae}(x_s;s,t)$ measures reconstruction accuracy under the conditional distribution of the stochastic target state.

\subsection{Latent stochastic evolution}
\label{subsec:latent-stochastic-evolution}

The true conditional solution law induces a corresponding probability law in latent space through the encoder. Since Lipschitz maps preserve finite $p$-th moments under pushforward, \Cref{ass:stochastic-dynamics} and
\Cref{ass:encoder-decoder} imply $ e_\#\Phi_{s,t}(x_s)\in\mathcal P_p(\mathcal Z)$.

\begin{assumption}[Latent evolution map]
\label{ass:latent-map}
For every $x_s\in\mathcal X$ and $(s,t)\in\mathcal T$, the learned latent transition law matches
$
    m(e(x_s),s,t)\in\mathcal P_p(\mathcal Z).
$
\end{assumption}

We quantify how closely the learned latent transition matches the encoded conditional solution law through the following Wasserstein error.

\begin{definition}[Latent transition error] \label{def:latent-map-error}
For $x_s\in\mathcal X$ and $(s,t)\in\mathcal T$, define
\begin{equation} \label{eq:latent-map-error}
    \varepsilon_{\rm map}(x_s;s,t) := W_p\left( m(e(x_s),s,t), e_\#\Phi_{s,t}(x_s) \right).
\end{equation}
\end{definition}
We can now combine the reconstruction and latent transition errors to bound the resulting error of the Stochastic Latent Twin.

\subsection{End-to-end approximation}
\label{subsec:end-to-end-approximation}

We first recall the stability of the Wasserstein distance under Lipschitz pushforwards.

\begin{lemma}[Decoder pushforward] \label{lem:decoder-pushforward}
Let $\rho,\eta\in\mathcal P_p(\mathcal Z)$ and let $d:\mathcal Z\to\mathcal X$ be $L_d$-Lipschitz. Then
\begin{equation} \label{eq:decoder-pushforward}
    W_p(d_\#\rho,d_\#\eta) \leq L_d W_p(\rho,\eta).
\end{equation}
\end{lemma}

\begin{proof}
Let $(Z,W)$ be a pair of random variables whose joint distribution has marginals $Z\sim\rho$ and $W\sim\eta$, such a joint distribution is called a \emph{coupling} of $\rho$ and $\eta$ in optimal transport. Applying the decoder to both components gives a joint random variable $(d(Z),d(W))$ with marginals $d_\#\rho$ and $d_\#\eta$.

By the definition of the $p$-Wasserstein distance, any joint distribution with the prescribed marginals provides an upper bound on $W_p^p$. Hence,
\[
    W_p^p(d_\#\rho,d_\#\eta)
    \leq
    \mathbb E\left[
        \|d(Z)-d(W)\|_{\mathcal X}^p
    \right]\leq
    L_d^p
    \mathbb E\left[
        \|Z-W\|_{\mathcal Z}^p
    \right],
\]
where the latter inequality follows from $L_d$-Lipschitz continuity of the decoder.

Taking the infimum over all couplings of $\rho$ and $\eta$ yields
\[
    W_p^p(d_\#\rho,d_\#\eta)
    \leq
    L_d^p W_p^p(\rho,\eta),
\]
 and the desired estimate follows after taking the $p$-th root.
\end{proof}

We can now combine the latent approximation and reconstruction estimates to
obtain an end-to-end error bound for the Stochastic Latent Twin.

\begin{theorem}[Stochastic Latent Twin approximation] \label{thm:slt-approximation}
Suppose \Cref{ass:stochastic-dynamics,ass:encoder-decoder,ass:latent-map} hold. Then, for every $x_s\in\mathcal X$ and $(s,t)\in\mathcal T$,
\begin{equation} \label{eq:slt-approximation}
    W_p\left( \Psi_{s,t}(x_s), \Phi_{s,t}(x_s) \right)
    \leq
    L_d\,\varepsilon_{\rm map}(x_s;s,t) + \varepsilon_{\rm ae}(x_s;s,t).
\end{equation}
\end{theorem}

\begin{proof}
Fix $x_s\in\mathcal X$ and $(s,t)\in\mathcal T$. By the definition of the Stochastic Latent Twin and the triangle inequality,
\begin{align}
    W_p\left( \Psi_{s,t}(x_s), \Phi_{s,t}(x_s) \right) 
    \leq W_p\left( d_\#m(e(x_s),s,t), d_\#e_\#\Phi_{s,t}(x_s) \right)
    + W_p\left( d_\#e_\#\Phi_{s,t}(x_s), \Phi_{s,t}(x_s) \right).
\label{eq:slt-error-decomposition}
\end{align}

By \Cref{lem:decoder-pushforward} and \Cref{def:latent-map-error},
\[
    W_p\left( d_\#m(e(x_s),s,t), d_\#e_\#\Phi_{s,t}(x_s) \right) \leq L_d\,\varepsilon_{\rm map}(x_s;s,t).
\]

For the second term, the composition rule for pushforwards gives
\[
    d_\#e_\#\Phi_{s,t}(x_s) = (d\circ e)_\#\Phi_{s,t}(x_s).
\]
Since $(X_t\mid X_s=x_s)\sim\Phi_{s,t}(x_s)$, the pair $\bigl(d(e(X_t)),X_t\bigr)$ defines a coupling of $(d\circ e)_\#\Phi_{s,t}(x_s)$ and $\Phi_{s,t}(x_s)$. Therefore,
\begin{align*}
    W_p^p\left( (d\circ e)_\#\Phi_{s,t}(x_s), \Phi_{s,t}(x_s) \right)
    \leq
    \mathbb E\left[ \|d(e(X_t))-X_t\|_{\mathcal X}^p \,\middle|\, X_s=x_s \right]
    = \varepsilon_{\rm ae}(x_s;s,t)^p.
\end{align*}
Hence,
\[
    W_p\left( d_\#e_\#\Phi_{s,t}(x_s), \Phi_{s,t}(x_s) \right)
    \leq 
    \varepsilon_{\rm ae}(x_s;s,t).
\]
Combining the two estimates proves the result.
\end{proof}

\begin{corollary}[Uniform approximation] \label{cor:slt-uniform-approximation}
For any $\mathcal X_0\subseteq\mathcal X$ and $\mathcal T_0\subseteq\mathcal T$, suppose that
\[
    \sup_{\substack{x_s\in\mathcal X_0\\(s,t)\in\mathcal T_0}} \varepsilon_{\rm map}(x_s;s,t) \leq \overline{\varepsilon}_{\rm map},
    \quad \text{and} \qquad
    \sup_{\substack{x_s\in\mathcal X_0\\(s,t)\in\mathcal T_0}} \varepsilon_{\rm ae}(x_s;s,t) \leq \overline{\varepsilon}_{\rm ae}.
\]
Then
\begin{equation}
\label{eq:slt-uniform-approximation}
\sup_{\substack{x_s\in\mathcal X_0\\(s,t)\in\mathcal T_0}}
W_p\left(
    \Psi_{s,t}(x_s),
    \Phi_{s,t}(x_s)
\right)
\leq
L_d\overline{\varepsilon}_{\rm map}
+
\overline{\varepsilon}_{\rm ae}.
\end{equation}
\end{corollary}

Consequently, if $\overline{\varepsilon}_{\rm ae}\to0$ and $\overline{\varepsilon}_{\rm map}\to0$ along a sequence of Stochastic Latent Twins, then the corresponding surrogate solution operators converge uniformly to the stochastic solution operator in $W_p$ over those sets.
 
\begin{remark}[Relation to deterministic Latent Twins]
The deterministic Latent Twin approximation theory is commonly formulated through uniform estimates over prescribed state sets. Here, the basic approximation result is instead formulated directly for conditional probability laws. Uniform approximation follows when the distribution-weighted reconstruction and latent transition errors are uniformly controlled over a chosen set of conditioning states and time pairs. 
\end{remark}

The approximation theory developed above is formulated directly for probability-valued transition operators and does not require an explicit differential representation of the underlying stochastic dynamics. Thus, while SDEs and SPDEs provide the primary examples considered in this work, the results apply more generally to any family of conditional transition laws
\[
    \Phi_{s,t}(x)\in\mathcal P_p(\mathcal X),
    \qquad
    (s,t)\in\mathcal T,
\]
satisfying the stated moment and regularity assumptions. An SDE should therefore be regarded as one mechanism for generating such a transition family rather than as part of the definition of a Stochastic Latent Twin.

\subsection{Stability of the learned stochastic solution operator}
\label{subsec:slt-stability}

The sensitivity of the learned stochastic solution operator to perturbations of the conditioning state is governed by the stability of its components. Given two source states $x,y\in\mathcal X$ at time $s$, their separation bounds the distance between the corresponding predicted distributions at time $t$.

Under the Lipschitz assumptions on the encoder and decoder in
\Cref{ass:encoder-decoder}, suppose additionally that the learned latent evolution map is Wasserstein stable,
\[
    W_p\left(m(z,s,t),m(w,s,t)\right) \leq L_m(s,t)\|z-w\|_{\mathcal Z}, \qquad z,w\in\mathcal Z.
\]
Then, for all $x,y\in\mathcal X$,
\begin{equation} \label{eq:physical-stability}
    W_p\left( \Psi_{s,t}(x), \Psi_{s,t}(y) \right)
    \leq
    L_dL_m(s,t)L_e\|x-y\|_{\mathcal X}.
\end{equation}
Perturbations of the conditioning state therefore remain controlled under the learned stochastic evolution.

The same question can be posed when the state at time $s$ is itself random. Under the same hypothesis, for source-state distributions $\mu,\nu\in\mathcal P_p(\mathcal X)$,
\begin{equation}\label{eq:distribution-stability}
    W_p\left(
        \int_{\mathcal X}\Psi_{s,t}(x)\,\mu(\mathrm dx),
        \int_{\mathcal X}\Psi_{s,t}(x)\,\nu(\mathrm dx)
    \right)
    \leq
    L_dL_m(s,t)L_e\,W_p(\mu,\nu).
\end{equation}
Thus, closeness of two source-state distributions in $W_p$ implies
controlled closeness of their evolved distributions under the learned
Stochastic Latent Twin.

\subsection{Forward and backward conditional evolution}
\label{subsec:forward-backward}

A useful consequence of the conditional-law formulation is that the Stochastic Latent Twin framework does not intrinsically require a forward time-stepping formulation. Since the approximation theory above is formulated in terms of conditional laws rather than forward propagation, it applies naturally to both forward prediction and backward inference, connecting the latter to Bayesian inference and uncertainty quantification \citep{kaipio2005statistical,stuart2010inverse}.

For $s<t$, the forward and backward conditional laws are
\[
\Phi_{s,t}(x) = \operatorname{Dist}(X_t\mid X_s=x), \qquad \Phi_{t,s}(y) = \operatorname{Dist}(X_s\mid X_t=y).
\]
The backward conditional law is defined for almost every $y$ with respect to the law of $X_t$.
The backward conditional law is not obtained by simply inverting the forward solution operator. Rather, it characterizes the distribution of earlier states consistent with a later observation. Which earlier states are consistent depends on where the process was likely to be at time $s$, and not only on the forward transition. When transition and marginal densities exist, with $\rho_s$ and $\rho_t$ denoting the marginal densities of $X_s$ and $X_t$, respectively, Bayes' formula gives $ p(x_s\mid x_t) = \tfrac{p(x_t\mid x_s)\rho_s(x_s)}{\rho_t(x_t)}$, showing that backward inference depends on both the forward transition law and the marginal distributions at the conditioning times.

Backward evolution is therefore a conditional inference problem rather than negative-time integration of the forward SDE. Under suitable regularity assumptions, a time-reversed diffusion may itself admit an SDE representation. However, the reverse-time dynamics are determined by the backward conditional law and generally differ from the original forward-time dynamics \citep{haussmann1986time}.

\subsection{Physics conformity}
\label{subsec:physics}

For a deterministic solution operator, physics conformity can be expressed directly in terms of the predicted state evolution \citep{chung2027physics}. In the stochastic setting, the prediction is instead probability-valued, raising the question of how physics conformity should be defined for such stochastic evolution.

The appropriate notion depends on the physical property under consideration. Conservation laws require physical quantities to remain unchanged, whereas dissipative laws impose one-sided constraints on their evolution. In a stochastic system, such properties may hold for every realization or only on average. Other physical properties, such as positivity, or boundedness, constrain the set of admissible states and may be expressed through the support of the predicted distribution; see, e.g., \citep{oksendal2003stochastic,daprato2014stochastic}.

Let $\mathcal C:\mathcal X\rightarrow\mathbb R^q$ denote a collection of physical quantities or observables that characterize the physical structure of the dynamics, such as mass, energy, or momentum.

\begin{definition}[Distributional conformity]
\label{def:distributional-physics-conformity}
Let $\Phi_{s,t}(x_s) = \operatorname{Dist}(X_t\mid X_s=x_s)$ denote the stochastic solution operator, and let $\mathcal C$ represent physical quantities of interest. The stochastic evolution is \emph{distributionally conforming} with respect to $\mathcal C$ if
\begin{equation}\label{eq:distributional-conformity}
    \mathcal C_\#\Phi_{s,t}(x_s) = \delta_{\mathcal C(x_s)},
\end{equation}
where $\delta_{\mathcal C(x_s)}$ denotes the Dirac measure concentrated at the value $\mathcal C(x_s)$.
\end{definition}

The pushforward measure describes the distribution of the physical quantities under the stochastic evolution. Thus, \Cref{eq:distributional-conformity} requires this distribution to remain a Dirac measure concentrated at $\mathcal C(x_s)$, expressing exact preservation of the physical quantities under the stochastic evolution.

Physical laws may also take the form of inequalities rather than equalities. In particular, dissipative quantities such as energy may decrease along the evolution rather than remain constant. This leads to the corresponding distributional inequality. For clarity, we formulate dissipation for a single scalar physical quantity; multiple inequality constraints can be treated componentwise in the same way.

\begin{definition}[Distributional dissipation] \label{def:distributional-dissipation}
Let $\mathcal C:\mathcal X\rightarrow\mathbb R$ denote a dissipative physical quantity. The stochastic evolution is \emph{distributionally dissipative} with respect to $\mathcal C$ if, for every $x_s\in\mathcal X$ and every $(s,t)\in\mathcal T$ with $s\leq t$,
\begin{equation}\label{eq:distributional-dissipation}
    \operatorname{supp}\left( \mathcal C_\#\Phi_{s,t}(x_s) \right) \subseteq (-\infty,\mathcal C(x_s)].
\end{equation}
\end{definition}
The support condition in \Cref{eq:distributional-dissipation} requires the dissipative quantity not to exceed its value at $x_s$, i.e., $\mathcal C(X_t)\leq\mathcal C(x_s)$ almost surely given $X_s=x_s$. Thus, conservation concentrates the pushed-forward transition law at $\mathcal C(x_s)$, whereas dissipation restricts its support to values not exceeding $\mathcal C(x_s)$.

If $\Phi_{s,t}(x_s)=\delta_{x_t}$, these conditions reduce to $\mathcal C(x_t)=\mathcal C(x_s)$ and $\mathcal C(x_t)\leq\mathcal C(x_s)$, respectively, recovering the deterministic physics-conformity conditions.

\begin{remark}[Alternative notions of stochastic physics conformity] \label{rem:alternative-physics-conformity}
The distributional conditions above impose conservation or dissipation realization-wise. In other settings, weaker notions may be appropriate. For example, a conservation law may hold only in expectation, $\mathbb E\left[ \mathcal C(X_t)\mid X_s=x_s \right] = \mathcal C(x_s)$, with the analogous inequality for dissipative quantities. Other physical constraints can be imposed through support restrictions. For example, positivity requires the predicted distribution to be supported only on nonnegative states. Realization-wise conservation or dissipation implies the corresponding condition in expectation, but the converse need not hold.
\end{remark}

\paragraph{Conformity inherited from stochastic solution-operator approximation.}

For the exact dynamics satisfying the distributional conformity condition, the deviation of a learned stochastic solution operator $\Psi_{s,t}$ from this
physical property can be quantified by the conformity defect
\begin{equation}\label{eq:physics-defect}
    \mathfrak C_p(x_s;s,t) := W_p\left( \mathcal C_\#\Psi_{s,t}(x_s), \delta_{\mathcal C(x_s)} \right).
\end{equation}
The following result shows that this defect is controlled by the approximation
error of the stochastic solution operator.

\begin{corollary}[A posteriori physics conformity from solution-law approximation] \label{cor:physics-from-approx}
Assume that $\mathcal C:\mathcal X\rightarrow\mathbb R^q$ is $L_{\mathcal C}$-Lipschitz and that $\Phi$ satisfies \cref{eq:distributional-conformity}. Under the assumptions of \Cref{thm:slt-approximation}, for every $x_s\in\mathcal X$ and $(s,t)\in\mathcal T$,
\begin{equation}
\label{eq:physics-from-approx}
    \mathfrak C_p(x_s;s,t) \leq L_{\mathcal C} \left( L_d\varepsilon_{\rm map}(x_s;s,t) + \varepsilon_{\rm ae}(x_s;s,t) \right).
\end{equation}
\end{corollary}

\begin{proof}
Since $\mathcal C$ is $L_{\mathcal C}$-Lipschitz,
\[
    \mathfrak C_p(x_s;s,t)
    =
    W_p\left(
        \mathcal C_\#\Psi_{s,t}(x_s),
        \mathcal C_\#\Phi_{s,t}(x_s)
    \right)
    \leq
    L_{\mathcal C}
    W_p\left(
        \Psi_{s,t}(x_s),
        \Phi_{s,t}(x_s)
    \right),
\]
and the result follows from \Cref{thm:slt-approximation}.
\end{proof}

This result establishes that physical conformity follows naturally as a consequence of accurate stochastic solution-law approximation. However, such a posteriori conformity does not explicitly exploit known physical structure. A stronger approach is to incorporate physical constraints directly into the latent representation and evolution, providing an \emph{a priori} mechanism in which conformity is enforced by construction \citep{chung2027physics}.

\paragraph{Physics-conforming latent evolution.}

The construction of a physics-conforming stochastic latent evolution requires two ingredients. First, we introduce a latent representation of the physical quantity to be preserved or dissipated by the stochastic evolution. Second, this latent quantity must be compatible with its physical-space counterpart.

The conservation and dissipation conditions introduced above can likewise be imposed directly on the latent stochastic evolution. Let $\mathcal C_{\mathcal Z}:\mathcal Z\rightarrow\mathbb R^q$ denote the latent physical quantity corresponding to $\mathcal C:\mathcal X\rightarrow\mathbb R^q$. For conservation, exact latent distributional conformity requires
\begin{equation}\label{eq:latent-distributional-conformity}
    (\mathcal C_{\mathcal Z})_\# m(z_s,s,t) = \delta_{\mathcal C_{\mathcal Z}(z_s)}.
\end{equation}
Thus, although the latent state may evolve stochastically, the physical quantity represented by $\mathcal C_{\mathcal Z}$ remains concentrated at $\mathcal C_{\mathcal Z}(z_s)$. More generally, approximate latent conformity requires
\begin{equation}\label{eq:latent-as-conformity}
    W_p\left( (\mathcal C_{\mathcal Z})_\# m(z_s,s,t), \delta_{\mathcal C_{\mathcal Z}(z_s)} \right) \leq \varepsilon_{\rm latent}.
\end{equation}
Unlike the a posteriori conformity obtained from solution-law approximation, \Cref{eq:latent-as-conformity} imposes physical structure directly on the learned latent evolution.

To ensure that the latent physical quantity represents the same physical quantity after decoding, exact compatibility would require $\mathcal C(d(z)) = \mathcal C_{\mathcal Z}(z)$. More generally, we allow an approximate pullback relation,
\begin{equation}\label{eq:stochastic-pullback}
\sup_{z\in\mathcal Z}
\left\|
\mathcal C(d(z))-\mathcal C_{\mathcal Z}(z)
\right\|
\leq
\varepsilon_{\rm pull}.
\end{equation}
This condition anchors the latent physical quantity to the physical domain. Consequently, preservation of $\mathcal C_{\mathcal Z}$ under the latent stochastic evolution translates into preservation of the physical quantity $\mathcal C$ after decoding.

\begin{theorem}[Stochastic physics-conformity transfer] \label{thm:stochastic-physics-transfer}
Suppose \Cref{ass:latent-map} holds and that \Cref{eq:latent-as-conformity,eq:stochastic-pullback} hold. If $\mathcal C$ is $L_{\mathcal C}$-Lipschitz, then
\begin{equation}
\label{eq:physics-transfer-bound}
    \mathfrak C_p(x_s;s,t) \leq \varepsilon_{\rm latent} + 2\varepsilon_{\rm pull} + L_{\mathcal C}\varepsilon_{\rm ae}(x_s;s,s).
\end{equation}
\end{theorem}

\begin{proof}
Let $z_s=e(x_s)$ and $Z_t\sim m(z_s,s,t)$. By the definition of the Stochastic Latent Twin, $ \operatorname{Dist}(d(Z_t)) = d_\#m(z_s,s,t) = \Psi_{s,t}(x_s)$, and consequently $ \operatorname{Dist}(\mathcal C(d(Z_t))) = \mathcal C_\#\Psi_{s,t}(x_s)$. It follows that
\[
\begin{aligned}
    \mathfrak C_p(x_s;s,t) = W_p\left( \mathcal C_\#\Psi_{s,t}(x_s), \delta_{\mathcal C(x_s)} \right) 
    = W_p\left( \operatorname{Dist}(\mathcal C(d(Z_t))), \delta_{\mathcal C(x_s)} \right) 
    = \left( \mathbb E\left[ \left\| \mathcal C(d(Z_t)) -\mathcal C(x_s) \right\|^p \right] \right)^{1/p}.
\end{aligned}
\]
The final equality follows from the explicit form of the Wasserstein distance to a Dirac measure. To relate this physical-space discrepancy to the latent conformity condition, we insert the latent observable at $Z_t$ and $z_s$, together with its decoded value at $z_s=e(x_s)$. The triangle inequality gives

\begin{equation}
\begin{aligned}
    \left\| \mathcal C(d(Z_t))-\mathcal C(x_s) \right\|
    &\leq \left\| \mathcal C(d(Z_t))-\mathcal C_{\mathcal Z}(Z_t) \right\|
    + \left\| \mathcal C_{\mathcal Z}(Z_t)-\mathcal C_{\mathcal Z}(z_s) \right\| \\[1ex]
    &\qquad\qquad+\left\| \mathcal C_{\mathcal Z}(z_s)-\mathcal C(d(z_s)) \right\|
    + \left\| \mathcal C(d(e(x_s)))-\mathcal C(x_s) \right\|.
\end{aligned}
\end{equation}
The first and third terms are bounded by $\varepsilon_{\rm pull}$ through \Cref{eq:stochastic-pullback}. Since $Z_t\sim m(z_s,s,t)$, the Wasserstein distance to the Dirac measure in \Cref{eq:latent-as-conformity} can be written equivalently as
\[
    \left( \mathbb E\left[ \left\| \mathcal C_{\mathcal Z}(Z_t) - \mathcal C_{\mathcal Z}(z_s) \right\|^p \right] \right)^{1/p}
    \leq
    \varepsilon_{\rm latent}.
\]

For the identity transition, $\Phi_{s,s}(x_s)=\delta_{x_s}$, so that \Cref{def:stochastic-ae-error} gives
\[
    \varepsilon_{\rm ae}(x_s;s,s) = \|d(e(x_s))-x_s\|_{\mathcal X}.
\]
Hence, by the Lipschitz continuity of $\mathcal C$,
\[
    \left\|
        \mathcal C(d(e(x_s)))-\mathcal C(x_s)
    \right\|
    \leq
    L_{\mathcal C}\varepsilon_{\rm ae}(x_s;s,s).
\]
Since $\mathfrak C_p$ is defined through an $L^p$ norm, we use Minkowski's inequality to combine the four contributions in the preceding decomposition. Thus,
\[
    \mathfrak C_p(x_s;s,t) \leq \varepsilon_{\rm latent} +2\varepsilon_{\rm pull} +L_{\mathcal C}\varepsilon_{\rm ae}(x_s;s,s),
\]
which proves the result.
\end{proof}

When latent conformity is enforced exactly by the architecture, the latent conformity error disappears from the transfer bound.

\begin{corollary}[Physics conformity by construction]
\label{cor:physics-conformity-construction}
Under the assumptions of \Cref{thm:stochastic-physics-transfer}, suppose that the latent stochastic evolution is distributionally conforming by construction, so that $\varepsilon_{\rm latent}=0$. Then
\begin{equation}\label{eq:physics-conformity-construction}
    \mathfrak C_p(x_s;s,t) \leq 2\varepsilon_{\rm pull} + L_{\mathcal C}\varepsilon_{\rm ae}(x_s;s,s).
\end{equation}
If, in addition, the latent and physical quantities are exactly compatible, so that $\varepsilon_{\rm pull}=0$, then
\begin{equation}
    \mathfrak C_p(x_s;s,t) \leq L_{\mathcal C}\varepsilon_{\rm ae}(x_s;s,s).
\end{equation}
In particular, exact reconstruction implies exact distributional physics conformity.
\end{corollary}

The distinction between the a posteriori and physics-conforming estimates is particularly apparent in this case. Without explicit physical structure, conformity is inherited from the accuracy of the learned stochastic solution operator and therefore depends on $\varepsilon_{\rm map}$. When latent conformity is enforced by construction, this dependence disappears: physical conformity is controlled solely by the compatibility of the latent and physical quantities and by the reconstruction error. Thus, the learned stochastic evolution may retain approximation error while satisfying the prescribed physical law exactly in latent space.

\paragraph{Dissipative physical laws.}

The same latent-to-physical construction extends to the dissipative inequality in \Cref{def:distributional-dissipation}. For a scalar dissipative quantity, let $(a)_+:=\max\{a,0\}$. With $z_s=e(x_s)$ and $Z_t\sim m(z_s,s,t)$, define the physical and latent  dissipation defects by
\begin{equation}
    \mathfrak D_p(x_s;s,t) :=
\left(
\mathbb E\left[
\left(
\mathcal C(d(Z_t))
-\mathcal C(x_s)
\right)_+^p
\right]
\right)^{1/p}, \  \text{and} \
\mathfrak D_p^{\mathcal Z}(z_s;s,t)
:=
\left(
\mathbb E\left[
\left(
\mathcal C_{\mathcal Z}(Z_t)
-\mathcal C_{\mathcal Z}(z_s)
\right)_+^p
\right]
\right)^{1/p}.
\end{equation}
The latent evolution is approximately dissipative if $\mathfrak D_p^{\mathcal Z}(z_s;s,t) \leq \varepsilon_{\rm latent}^{\rm diss}$, and exactly dissipative when $\mathfrak D_p^{\mathcal Z}(z_s;s,t)=0$.

\begin{corollary}[Stochastic dissipation transfer]
\label{cor:stochastic-dissipation-transfer}
Suppose that $\mathcal C:\mathcal X\rightarrow\mathbb R$ is $L_{\mathcal C}$-Lipschitz, the pullback condition \Cref{eq:stochastic-pullback} holds, and, with $z_s=e(x_s)$, $\mathfrak D_p^{\mathcal Z}(z_s;s,t) \leq \varepsilon_{\rm latent}^{\rm diss}$. Then
\begin{equation}\label{eq:dissipation-transfer-bound}
    \mathfrak D_p(x_s;s,t)
    \leq
    \varepsilon_{\rm latent}^{\rm diss}
    +
    2\varepsilon_{\rm pull}
    +
    L_{\mathcal C}\varepsilon_{\rm ae}(x_s;s,s).
\end{equation}
\end{corollary}

\begin{proof}
Using the same decomposition as in \Cref{thm:stochastic-physics-transfer} together with $(a+b)_+\leq a_+ + |b|$, the physical dissipation violation is bounded by the latent dissipation violation, the two pullback errors, and the reconstruction error. Taking the $L^p$ norm and using the assumed latent dissipation bound, \Cref{eq:stochastic-pullback}, and $\varepsilon_{\rm ae}(x_s;s,s) =\|d(e(x_s))-x_s\|_{\mathcal X}$ gives the result.
\end{proof}

Analogous latent-to-physical arguments apply to the expectation-level and admissibility constraints discussed in \Cref{rem:alternative-physics-conformity}.

The preceding conformity defects also provide natural regularization terms for training; these are given in \Cref{sec:numerics}.

\subsection{Linear stochastic systems}
\label{subsec:linear-stochastic}

The linear SDE provides a setting in which the approximation properties of a
Stochastic Latent Twin can be characterized explicitly. Consider
\begin{equation} \label{eq:linear-sde}
    \mathrm dX_t = AX_t\,\mathrm dt + G\,\mathrm dW_t,
\end{equation}
where $A\in\mathbb R^{n_x\times n_x}$ and $G\in\mathbb R^{n_x\times n_w}$. For $t\geq s$, the stochastic solution operator is
\begin{equation} \label{eq:linear-transition}
    \Phi_{s,t}(x_s) = \mathcal N\!\left( M_{s,t}x_s, Q_{s,t} \right),
    \quad \text{with }
    M_{s,t} = \mathrm e^{(t-s)A}
    \text{ and } 
    Q_{s,t} = \int_s^t \mathrm e^{(t-r)A} GG^\top \mathrm e^{(t-r)A^\top} \,\mathrm dr .
\end{equation}

We consider a linear-Gaussian Stochastic Latent Twin with trainable linear encoder and decoder $e(x)=Ex$, $d(z)=Dz$, and Gaussian latent transition law $m(z_s,s,t) = \mathcal N\!\left( M_{s,t}^{\mathcal Z}z_s, Q_{s,t}^{\mathcal Z} \right)$. The corresponding conditional transition distribution in physical space is therefore
\begin{equation} \label{eq:linear-slt-transition}
    \Psi_{s,t}(x_s) = \mathcal N\!\left( DM_{s,t}^{\mathcal Z}Ex_s, DQ_{s,t}^{\mathcal Z}D^\top \right).
\end{equation}
An immediate consequence of \cref{eq:linear-slt-transition} is that a Gaussian latent transition law followed by a linear decoder produces a Gaussian physical transition distribution, independently of the latent dimension $n_z$. Increasing $n_z$ can therefore improve the representation of a Gaussian transition law, but it cannot enlarge the predicted distributional family beyond Gaussian distributions. This distinction becomes important in the Cox--Ingersoll--Ross (CIR) experiment in \cref{subsec:cir}, where the exact transition law is non-Gaussian.

Here, $E$ and $D$ define the trainable encoder--decoder representation used across all considered source and target times, whereas $M_{s,t}^{\mathcal Z}$ and $Q_{s,t}^{\mathcal Z}$ determine the mean and covariance of the learned latent transition law for each time pair $(s,t)$.

At full latent dimension, $n_z=n_x$, the exact stochastic solution operator belongs to this model class. In particular, the choice $E=D=I$, $ M_{s,t}^{\mathcal Z}=M_{s,t}$, and $Q_{s,t}^{\mathcal Z}=Q_{s,t}$ gives
\[ 
    \Psi_{s,t}(x_s) = \Phi_{s,t}(x_s) \qquad \text{for every } s\leq t.
\]
Thus, a Stochastic Latent Twin can exactly recover the stochastic solution operator at full latent dimension. The preceding choice provides one explicit exact realization; the representation itself need not be unique.

For $n_z<n_x$, the exact full-dimensional realization is no longer available, and the approximation problem becomes genuinely more involved. Nevertheless, the reduced problem admits an important analytical simplification: for any candidate encoder $E$ and decoder $D$, the Wasserstein-optimal latent Gaussian transition law can be obtained explicitly, so that only the encoder and decoder remain to be optimized, as in a variable-projection reduction \citep{chung2023variable,golub1973differentiation}. The complete derivation is given in \cref{app:linear-low-rank}.

Because the encoder and decoder must represent all time pairs simultaneously, their optimization is generally not available in closed form. To obtain an analytical reference for the approximation achievable at a given latent dimension, we therefore consider one fixed source time $s$ and target time $t$ and allow the encoder and decoder to be optimized specifically for that time pair. Let $C_s:=\mathbb E[X_sX_s^\top]$ denote the second-moment matrix of the source state. For this reference problem, the remaining optimization also admits a closed-form solution. In particular, the optimal decoded $n_z$-dimensional space is determined by the dominant eigenspace of
\begin{equation} \label{eq:linear-H}
    H_{s,t} = M_{s,t}C_sM_{s,t}^\top + Q_{s,t}.
\end{equation}
If $\eta_1\geq\cdots\geq\eta_{n_x}\geq0$ are the eigenvalues of $H_{s,t}$, then the minimum expected squared $2$-Wasserstein transition error is
\begin{equation} \label{eq:linear-transition-optimal-error}
    \mathcal E_{s,t}^{\star}(n_z) = \sum_{i>n_z}\eta_i.
\end{equation}

Hence, for each fixed time pair, the eigenvalue expression provides a closed-form lower reference value for the transition error attainable with latent dimension $n_z$. In contrast, a Stochastic Latent Twin trained across multiple source and target times uses one encoder--decoder representation for all of these transitions.

The analytical results above provide reference solutions for the numerical experiments in \cref{sec:linear_sde_experiment}, where we investigate how closely trained linear-Gaussian Stochastic Latent Twins approach these predicted approximation limits as the latent dimension and prediction interval are varied.

\section{Numerical experiments}
\label{sec:numerics}

We assume that training data are given as transition samples $\{(x_{s_i},x_{t_i},s_i,t_i)\}_{i=1}^N$. For an observed transition $(x_s,x_t,s,t)$, the encoder gives $z_s=e(x_s)$ and $z_t=e(x_t)$. Since $x_t$ is a realization of the conditional law $\Phi_{s,t}(x_s)$, the encoded target $z_t$ is a realization of $e_\#\Phi_{s,t}(x_s)$. Thus, observed transitions provide direct samples from the conditional latent law targeted by $m(z_s,s,t)$.

For notational clarity, the training objectives below are written in expectation form; in practice, all expectations are replaced by empirical averages over the available samples.

The Stochastic Latent Twin is trained by
\begin{equation} \label{eq:training-objective}
    \min_{\theta}\; \lambda_{\rm ae}\mathcal L_{\rm ae} + \lambda_{\rm pred}\mathcal L_{\rm pred} + \lambda_{\rm reg}\mathcal L_{\rm reg},
\end{equation}
where $\theta$ collects the parameters of the encoder, latent evolution map, and decoder, $\mathcal L_{\rm ae}$ targets the distribution-weighted reconstruction error introduced in \Cref{def:stochastic-ae-error}, $\mathcal L_{\rm pred}$ targets the conditional transition law, and $\mathcal L_{\rm reg}$ denotes optional auxiliary or structural regularization. Depending on the experiment, the latter may represent temporal or physics-conformity conditions or a problem-specific distribution-level objective.

The deterministic encoder--decoder representation is trained using the reconstruction loss
\begin{equation} \label{eq:ae-training-loss}
    \mathcal L_{\rm ae} = \mathbb E \left[ \left\| d(e(X_\tau))-X_\tau \right\|_{\mathcal X}^2 \right],
\end{equation}
where $X_\tau$ is a state drawn from the training transitions. This corresponds to the case $p = 2$ in \Cref{def:stochastic-ae-error} with squared error. This objective then averages the same squared reconstruction error that defines the distribution-weighted reconstruction error over the state--time samples used during training. For other $p$, the exponent $2$ is replaced by $p$.

\paragraph{Distributional prediction losses.}

Our theoretical analysis in \Cref{sec:slt-approximation} measures the approximation error between the learned and exact stochastic solution laws in Wasserstein distance. When the exact transition law is available, this motivates the squared $2$-Wasserstein objective
\[
    \mathcal L_{\rm pred} = \mathbb E_{X_s}\!\left[ W_2^2\!\left( \Psi_{s,t}(X_s), \Phi_{s,t}(X_s) \right) \right].
\]

We use direct Wasserstein objectives in the controlled linear-SDE experiment \Cref{sec:linear_sde_experiment} and for several of the CIR models \Cref{subsec:cir,app:cir-models}.

For empirical distributions in higher dimensions, direct computation of the Wasserstein distance can become expensive. A computationally attractive alternative is the sliced $2$-Wasserstein distance, which averages one-dimensional Wasserstein distances between projections of the two distributions onto directions on the unit sphere \citep{rowland2019orthogonal}: in one dimension, empirical $W_2$ can be computed by matching sorted samples, avoiding a multidimensional optimal-transport problem. For simplicity, we refer to this as the sliced-Wasserstein distance (SWD).   We use sliced-Wasserstein distance as an auxiliary distribution-matching objective for the one-shot generative endpoint in the Swiss-roll experiment (\Cref{subsec:swissroll,app:swissroll-training}). The sliced-Wasserstein distance is also the evaluation metric in the Swiss-roll, pendulum, and Burgers experiments.  Throughout, we estimate this metric directly from samples of respective distributions; the number of projection directions for each experiment is given in the corresponding appendix.

Estimation of the sliced-Wasserstein distance can be inaccurate with few samples per conditioning state. Such data are common in fully data-driven settings, where the data consist of observed transitions $(X_s, X_t)$, often with only a single realization $X_t$ for each conditioning state $X_s$. Proper scoring rules are well suited to this setting: they evaluate a predictive distribution against a single observed outcome, and their expected score is minimized by the true conditional distribution \citep{gneiting2007strictly}.

When the latent transition law admits a tractable density $p(\cdot\mid z_s,s,t)$, the logarithmic score gives the negative log-likelihood
\begin{equation} \label{eq:nll-loss}
    \mathcal L_{\rm pred} = -\mathbb E\!\left[ \log p(Z_t\mid Z_s,s,t) \right],
\end{equation}
where $Z_s=e(X_s)$ and $Z_t=e(X_t)$. For a Gaussian latent transition law, this specializes to
\begin{equation} \label{eq:gaussian-nll-loss}
    \mathcal L_{\rm pred} = \mathbb E\!\left[ \tfrac12 (Z_t-\mu)^\top\Sigma^{-1}(Z_t-\mu) +\tfrac12\log\det\Sigma \right],
\end{equation}
up to an additive constant, where $\mu=\mu(Z_s,s,t)$ and $\Sigma=\Sigma(Z_s,s,t)$. Likelihood-based training is used for the Gaussian-mixture and learned noncentral-$\chi^2$ CIR models; see
\Cref{subsec:cir,app:cir-models}.

When samples can be drawn from the predicted conditional law but its density is unavailable or inconvenient to evaluate, we use the energy score. For an observed transition $(X_s,X_t)$, let
\[
    X_t^{(1)},X_t^{(2)} \overset{\mathrm{i.i.d.}}{\sim} \Psi_{s,t}(X_s).
\]
The corresponding loss is
\begin{equation} \label{eq:energy-score-loss}
    \mathcal L_{\rm pred} = \mathbb E\!\left[ \left\|X_t^{(1)}-X_t\right\|_{\mathcal X} - \tfrac12 \left\|X_t^{(1)}-X_t^{(2)}\right\|_{\mathcal X} \right].
\end{equation}
The  first term rewards agreement with the observed target, while the second accounts for the spread of the predictive distribution. We use the physical-space energy score for the stochastic Lorenz, Swiss-roll, stochastic pendulum and Burgers equation experiments \Cref{subsec:lorenz,subsec:swissroll,subsec:pendulum,subsec:burgers}.

\paragraph{Additional structural regularization.}

Beyond these prediction objectives, $\mathcal L_{\rm reg}$ may promote additional structure in the learned evolution. For instance, it can enforce the identity condition and promote Chapman--Kolmogorov consistency \Cref{eq:latent-chapman-kolmogorov}; see also \citep{kiyohara2026neural}.

For experiments involving known physical structure, the conformity defects introduced in \cref{subsec:physics} provide corresponding regularization terms. With the latent physical quantity $\mathcal C_{\mathcal Z}:\mathcal Z\rightarrow\mathbb R^q$, conservation may be promoted through
\begin{equation} \label{eq:latent-conservation-loss}
    \mathcal L_{\rm con} = \mathbb E \left[ \left\| \mathcal C_{\mathcal Z}( Z_t) - \mathcal C_{\mathcal Z}(z_s) \right\|^p \right], \qquad Z_t\sim m(z_s,s,t).
\end{equation}
For a scalar dissipative observable, the corresponding one-sided regularization is
\begin{equation} \label{eq:latent-dissipation-loss}
    \mathcal L_{\rm diss} = \mathbb E \left[ \left( \mathcal C_{\mathcal Z}(Z_t) - \mathcal C_{\mathcal Z}(z_s) \right)_+^p \right].
\end{equation}
Compatibility between the latent and physical observables may likewise be encouraged through
\begin{equation} \label{eq:pullback-loss}
    \mathcal L_{\rm pull} = \mathbb E \left[ \left\| \mathcal C(d(Z)) - \mathcal C_{\mathcal Z}(Z) \right\|^2 \right].
\end{equation}

For Gaussian transitions, the expectations in \Cref{eq:latent-conservation-loss,eq:latent-dissipation-loss} are approximated using Monte Carlo samples from the predicted latent Gaussian law. When conservation or dissipation is exactly enforced by the latent architecture, the corresponding regularization term is omitted. Likewise, the pullback term is unnecessary when latent-to-physical compatibility is built directly into the representation.

The architecture, training data, and choice of structural regularization used in each numerical example are specified below.

\paragraph{Sample-path generation.}

A Stochastic Latent Twin represents conditional transition laws between arbitrary time points. Discrete sample paths can therefore be generated by sequentially sampling these conditional laws. Given times $t_0<t_1<\cdots<t_M$ and an initial state $x_0$, we use
\[ z_k=e(x_k), \qquad Z_{k+1}\sim m(z_k,t_k,t_{k+1}), \qquad x_{k+1}=d(Z_{k+1}),
\]
for $k=0,\ldots,M-1$. Alternatively, sample paths may be generated entirely in latent space via
\[
    Z_{k+1}\sim m(Z_k,t_k,t_{k+1}),
\]
followed by decoding only when physical-space states are required. This latent recursion is appropriate when the learned latent evolution remains consistent under recursively generated latent states.

\subsection{Linear SDE: verification of the analytical approximation limits} \label{sec:linear_sde_experiment}

The linear-Gaussian setting isolates the effects of latent dimension and state representation when the probability family is already correctly specified. We therefore use the linear SDE \Cref{eq:linear-sde} as a controlled test of the approximation results of \Cref{subsec:linear-stochastic,app:linear-low-rank}. We consider a randomly generated $n_x=6$ dimensional stable, nonnormal system; the nonnormality avoids trivially aligned dominant subspaces across prediction horizons. Further details of the system construction and numerical optimization are given in \cref{app:linear-numerics}.

\Cref{fig:linear-sde} (left) illustrates the stochastic dynamics. Source and target times $s\leq t$ are selected over $[0,2]$, defining the conditional transitions used for training and evaluation. No sampled data are used in this experiment: the objective is evaluated in closed form from the Gaussian transition laws (see \cref{app:linear-numerics}).

The principal comparison is with the relaxed problem of \Cref{prop:transition-wise-optimum}, in which the encoder and decoder may depend jointly on $(s,t)$. Optimizing the representation separately for each transition yields an optimal lower bound on the transition error attainable by representations shared across multiple transitions. We compare this optimal lower bound with a representation $(E,D)$ shared by all transitions and with the time-dependent extension $(E_s,D_t)$ of \cref{app:linear-extensions}. In both trainable cases, the optimal latent Gaussian transition is eliminated analytically as described in \cref{app:linear-low-rank}; only the representation itself is optimized numerically. The reported error is the average over all forward time pairs of $\sqrt{\mathcal J_{s,t}}$, where $\mathcal J_{s,t}$ is the expected squared Wasserstein error of \cref{eq:linear-gaussian-w2-objective}.

\begin{figure}
    \centering
    \includegraphics[height=0.32\linewidth]{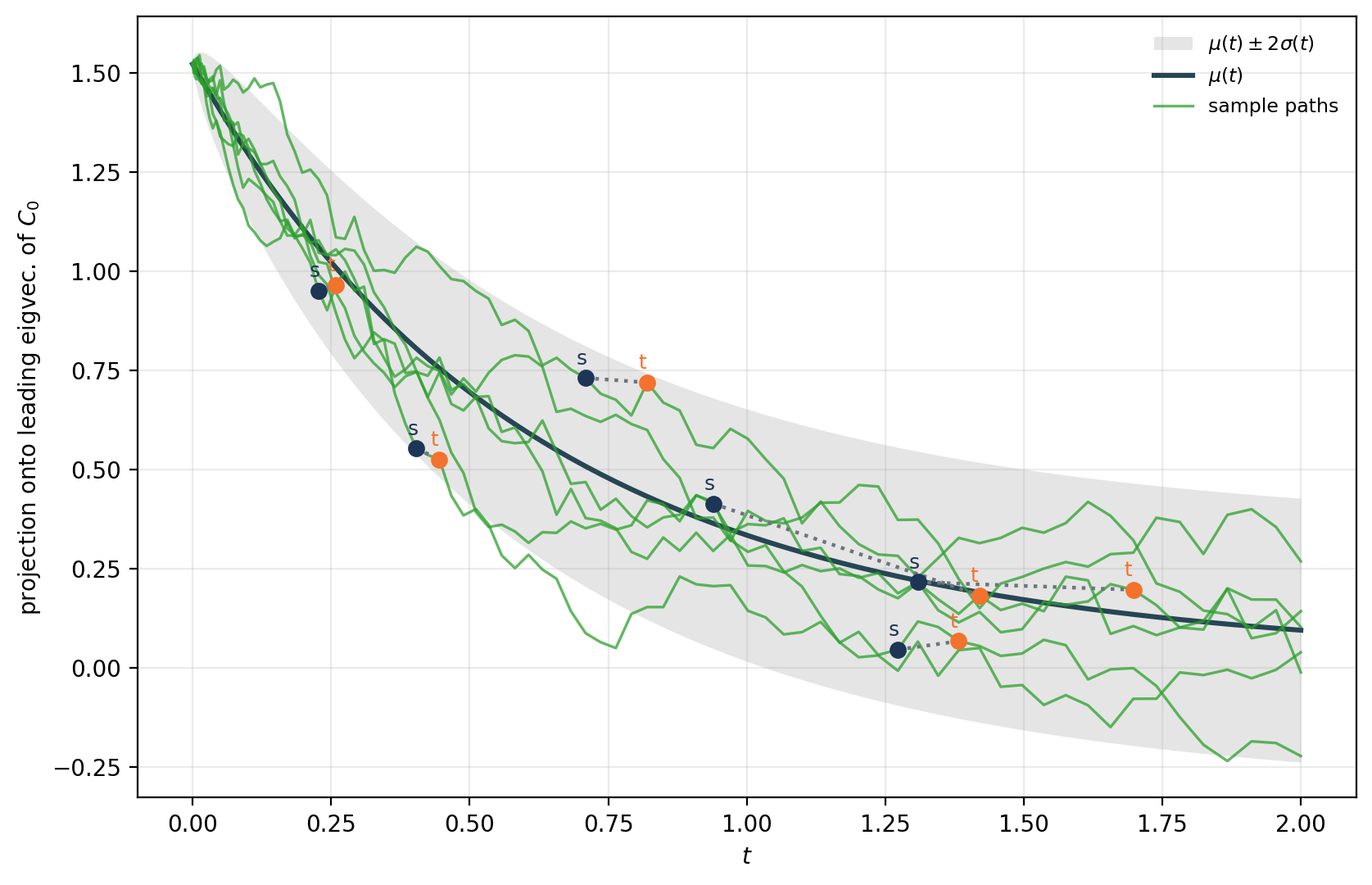}\quad
    \includegraphics[height=0.32\linewidth]{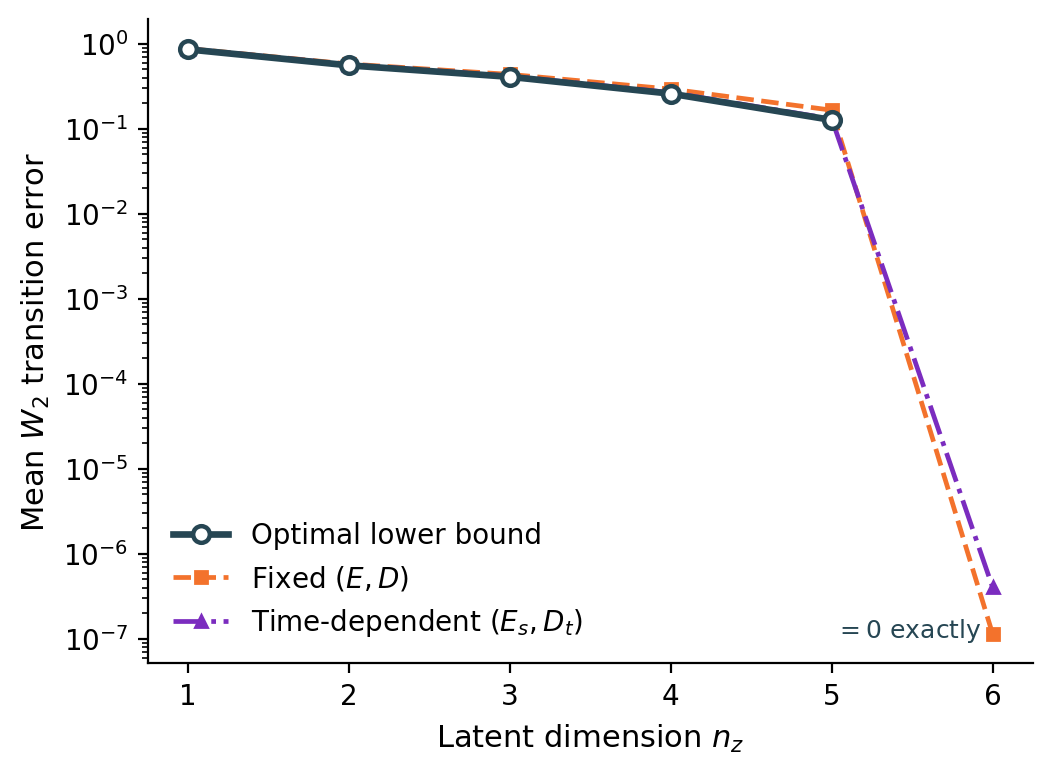}
    \caption{Linear stochastic system used for the verification experiment. \emph{Left:} stochastic dynamics, shown as the projection onto the leading eigenvector of the initial second-moment matrix $C_0$, with the mean $\pm 2$ standard deviations (shaded band), sample paths, and representative conditional transitions between source and target times $s<t$ used for training and evaluation. \emph{Right:} mean $2$-Wasserstein transition error versus latent dimension for the optimal lower bound, the fixed representation $(E,D)$, and the time-dependent representation $(E_s,D_t)$. The error is the mean of $\sqrt{\mathcal J_{s,t}}$ over all forward pairs (\cref{app:linear-numerics}); the full-dimensional case $n_z=n_x=6$ shows the exact-representation limit, where the optimal lower bound is exactly zero.}
    \label{fig:linear-sde}
\end{figure}

The numerical results \cref{fig:linear-sde} (right) reproduce the analytical approximation hierarchy. Increasing the latent dimension systematically reduces the transition error, and both trainable representations remain close to the optimal lower bound. Allowing the representation to vary as $(E_s,D_t)$ provides an additional improvement, although a common $(E,D)$ captures most of the reduction. At $n_z=n_x=6$, the transition error vanishes to numerical precision, confirming the exact-representation result of \Cref{subsec:linear-stochastic}.

Representative recursively generated trajectories are reported in \cref{app:linear-numerics}; as the latent dimension increases, the predicted mean and uncertainty approach the exact Gaussian evolution, with full recovery at $n_z=6$.

Although the results above exploit the analytical linear-Gaussian structure for controlled verification, this analytical elimination is not required by the Stochastic Latent Twin framework. In the subsequent experiments, the latent evolution map is instead learned directly from observed transition data.

\subsection{Cox--Ingersoll--Ross (CIR) model
} \label{subsec:cir}
We use the Cox--Ingersoll--Ross (CIR) model to examine what changes when the exact transition law is non-Gaussian. In particular, it allows us to separate three sources of approximation power that coincide in the linear-Gaussian example: latent dimension, nonlinear state representation, and the probability family used for the latent transition law. 

We consider the one-dimensional CIR model \citep{cox1985theory},
\begin{equation} \label{eq:cir-sde}
    \mathrm dX_t = \kappa(\beta-X_t)\,\mathrm dt + \sigma\sqrt{X_t}\,\mathrm dW_t.
\end{equation}
For $s \leq t$, its transition law is known explicitly and is a scaled noncentral chi-squared distribution \citep{cox1985theory},
\begin{equation}
\label{eq:cir-transition}
    \Phi_{s,t}(x_s) = c_{s,t}\, \chi^2_{\nu}\!\left(\lambda_{s,t}(x_s)\right),
\end{equation}
where $\chi^2_{\nu}(\lambda)$ denotes the noncentral chi-squared distribution with $\nu$ degrees of freedom and noncentrality parameter $\lambda$, $c_{s,t}>0$ is a scaling factor, and $\lambda_{s,t}(x_s)$ is the state-dependent noncentrality parameter. Throughout this experiment we take $\kappa=1$, $\theta=1$, and $\sigma=1.2$. Exact sampling from \Cref{eq:cir-transition} is available and is used to generate both the training data and reference solutions. The resulting CIR dynamics are illustrated in \Cref{fig:cir-summary} (left); further numerical details are given in \Cref{app:cir-numerics}.

\paragraph{Best Gaussian approximation.}
We first quantify how well the exact CIR transition law can be approximated by a Gaussian distribution. For a fixed source state $x_s$ and time pair $(s,t)$, consider
\begin{equation} \label{eq:cir-best-gaussian-problem}
    \mathcal E_{\rm G}^{\star}(x_s,s,t) = \min_{\mu\in\mathbb R,\;\sigma>0} W_2^2\!\left( \Phi_{s,t}(x_s), \mathcal N(\mu,\sigma^2) \right).
\end{equation}
We call the minimizing Gaussian the \emph{best Gaussian approximation} of $\Phi_{s,t}(x_s)$, and $\mathcal E_{\rm G}^{\star}$ its squared error. Using the one-dimensional quantile representation of the Wasserstein distance \citep{vallender1974calculation} and the corresponding characterization of the $W_2$-nearest Gaussian \citep{wang2026relative}, the minimizing mean $\mu^\star_{s,t}(x_s)$ is the mean of $\Phi_{s,t}(x_s)$, and the minimizing standard deviation is $\sigma^\star_{s,t}(x_s)=\int_0^1 \Phi_{s,t}^{-1}(u;x_s)\, \Phi_{\rm N}^{-1}(u)\,\mathrm du$, where $\Phi_{s,t}^{-1}(\,\cdot\,;x_s)$ and $\Phi_{\rm N}^{-1}$ denote the quantile functions of the CIR transition law and the standard normal distribution, respectively. The minimum squared Wasserstein error is
\begin{equation}
\label{eq:cir-optimal-gaussian-error}
    \mathcal E_{\rm G}^{\star}(x_s,s,t)
    = \int_{\mathbb R} \left( x-\mu^\star_{s,t}(x_s) \right)^2 \,\mathrm d\Phi_{s,t}(x_s)(x) - \bigl(\sigma^\star_{s,t}(x_s)\bigr)^2 .
\end{equation}
This value provides a lower reference for any model whose predicted physical transition law is Gaussian; the derivation is given in \cref{app:optimal-gaussian}.

We next test whether the Gaussian Stochastic Latent Twin considered in the linear-SDE experiment attains this reference. With a Gaussian latent transition law and a linear decoder, the decoded physical transition remains Gaussian for every latent dimension $n_z$. Since the physical state is one-dimensional, $n_z=1$ already suffices to represent an arbitrary scalar Gaussian transition, and increasing $n_z$ only provides a redundant parameterization of the same physical Gaussian family. We nevertheless train models with $n_z=1,2,$ and $4$ using the Wasserstein objective to verify this limitation numerically: all three models closely approach the error of the best Gaussian approximation, and their mean squared Wasserstein errors, averaged over source states and prediction horizons, differ by only approximately $0.03\%$. Thus, increasing latent dimension cannot remove the approximation error caused by restricting the physical transition law to the Gaussian family.

There are two distinct ways to enlarge the class of physical transition laws: one may enrich the latent probability family while retaining a linear representation, or retain a simple Gaussian latent law and enrich the encoder--decoder representation. We examine these two mechanisms separately.

\paragraph{Enriching the transition family.}

First, we enrich the latent probability family while keeping the encoder--decoder representation linear. Specifically, we replace the Gaussian transition by a $K=4$ Gaussian mixture,
\[
    m(z_s,s,t) = \sum_{j=1}^{K} \pi_j(z_s,s,t)\, \mathcal N\!\left( \mu_j(z_s,s,t), \Sigma_j(z_s,s,t) \right),
\]
with mixture weights, means, and covariances learned from the transition data. The resulting physical transition is therefore a Gaussian mixture as well.

As shown in \Cref{fig:cir-distributional-expressivity}, enriching the latent probability family substantially reduces the Wasserstein error relative to the Gaussian model. Under the common evaluation protocol, the mean squared Wasserstein error decreases from approximately $5.89\times10^{-2}$ for the single-Gaussian model to $1.85\times10^{-3}$ for the Gaussian mixture. This comparison shows that non-Gaussian physical transition laws can be obtained by enriching the latent probability family. 

\paragraph{Nonlinear representation.}

We next retain a single Gaussian latent transition law but replace the linear encoder and decoder by small multilayer perceptrons. The predicted physical transition is then
\begin{figure}
    \centering
    \includegraphics[width=\linewidth]{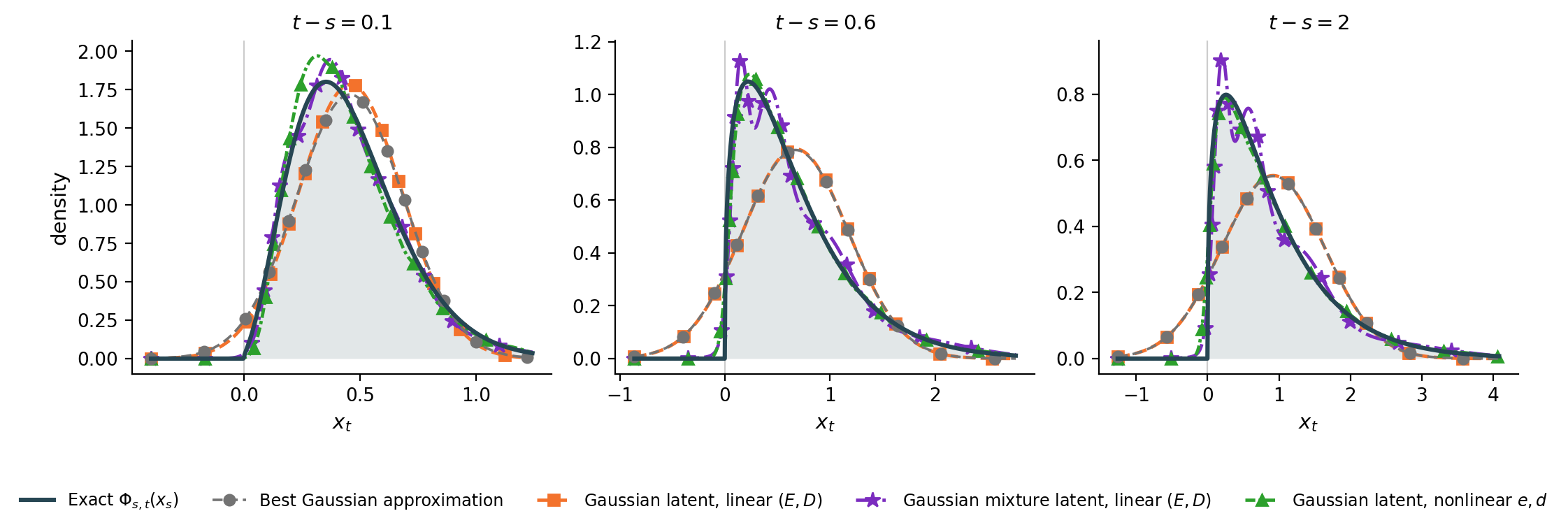}
    \caption{Distributional expressivity for the CIR transition. The linear-decoder Gaussian Stochastic Latent Twin approaches the best Gaussian approximation but cannot overcome its distributional restriction by increasing the latent dimension. A Gaussian mixture improves the approximation while retaining a linear representation, whereas a nonlinear decoder produces a highly accurate non-Gaussian physical transition from a single Gaussian latent law.}
    \label{fig:cir-distributional-expressivity}
\end{figure}
\begin{equation} \label{eq:cir-nonlinear-gaussian}
    \Psi_{s,t}(x_s) = d_\# \mathcal N\!\left( \mu(e(x_s),s,t), \Sigma(e(x_s),s,t) \right),
\end{equation}
which is no longer restricted to be Gaussian.

This change produces a much larger reduction in transition error. With $n_z=1$, the nonlinear representation reproduces the skewness and one-sided structure of the CIR transition substantially more accurately than the linear-decoder Gaussian models, see \cref{fig:cir-distributional-expressivity}. Under the common evaluation protocol, the mean squared Wasserstein error is approximately $2.69\times10^{-4}$, more than two orders of magnitude below the best Gaussian approximation. Thus, probability-family expressivity and representation expressivity are distinct: under a nonlinear encoder--decoder representation, a Gaussian latent law can induce a non-Gaussian physical transition.

\paragraph{Matching the CIR transition family.}

The previous comparisons obtain non-Gaussian physical transition laws either by enriching the latent distribution or by applying a nonlinear decoder. For CIR, however, the exact transition family itself is known. Since the problem is already one-dimensional, no separate state representation is required and we identify $z=x$. We therefore parameterize the transition directly as a scaled noncentral chi-squared law, $m(x_s,s,t) = c(x_s,s,t)\, \chi^2_{\nu(x_s,s,t)}\!\left(\lambda(x_s,s,t)\right)$.

We learn the parameters of this transition directly from sampled transitions, without supplying the analytical CIR formulas, using negative log-likelihood. The learned model is evaluated in the same $2$-Wasserstein metric used for the preceding comparisons and closely approximates the exact conditional transition laws across the tested source states and prediction horizons. Under the common evaluation protocol, the learned model attains a mean squared $2$-Wasserstein error of approximately $5.47\times10^{-5}$.

This comparison separates approximation from estimation. Restricting the physical transition to a Gaussian family introduces an irreducible approximation error, whereas the noncentral-$\chi^2$ family contains the exact CIR transition law. For the learned noncentral-$\chi^2$ model, the remaining error therefore reflects finite-data learning and numerical error rather than misspecification of the transition family.

The dependence on prediction horizon is summarized in \Cref{fig:cir-summary} (right), which reports the mean $2$-Wasserstein distance, without squaring, at each horizon.

\begin{figure}
    \centering
    \includegraphics[height=0.32\linewidth]{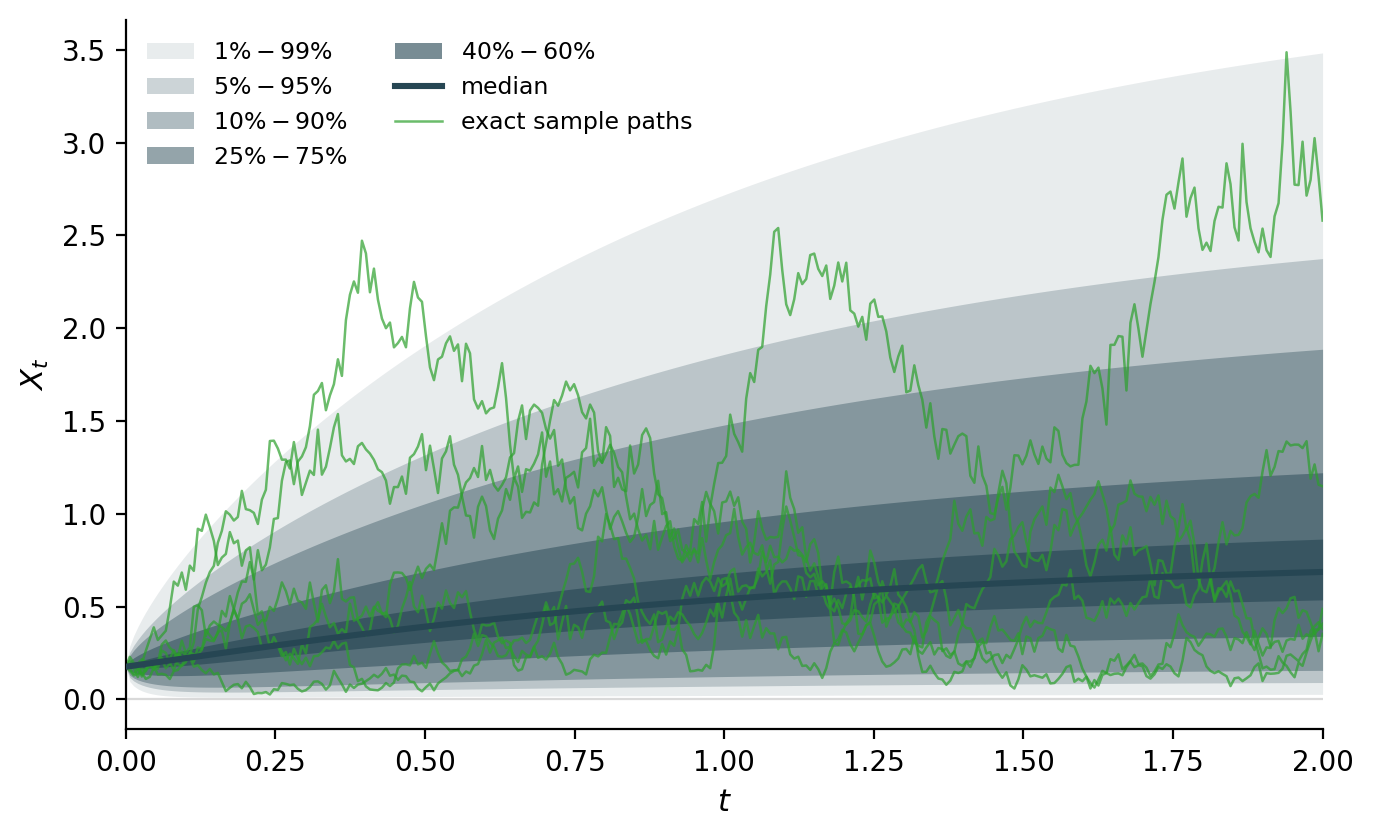}\quad
    \includegraphics[height=0.32\linewidth]{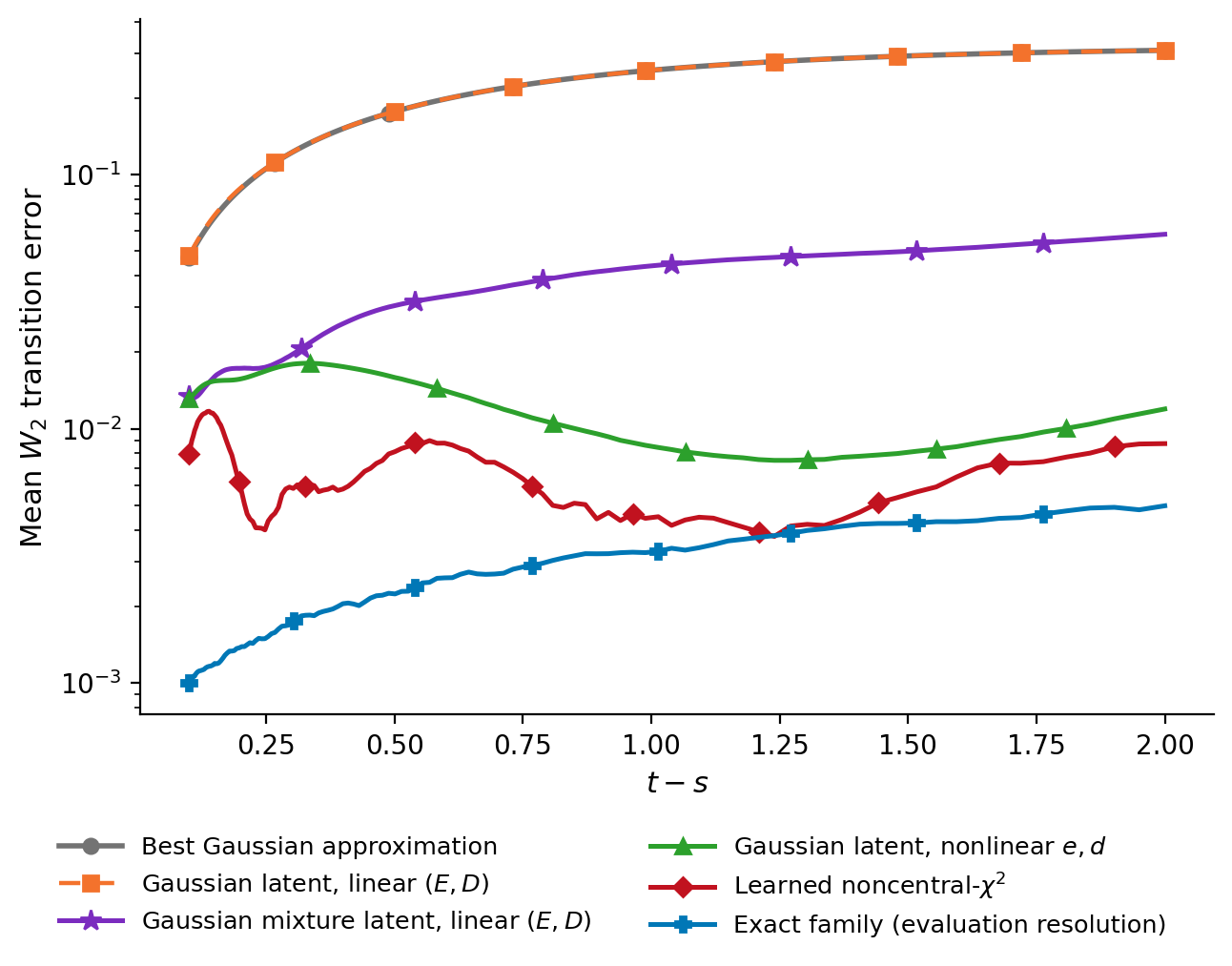}
    \caption{CIR experiment. \emph{Left:} exact CIR transition quantiles and representative sample paths; the shaded regions show central probability intervals of the conditional transition law as the prediction horizon increases. \emph{Right:} mean $2$-Wasserstein transition error as a function of prediction horizon for the CIR models, averaged over a common set of source states. 
    }
    \label{fig:cir-summary}
\end{figure}

Taken together, the CIR experiment separates three sources of approximation power that coincide in the linear-Gaussian setting. Approximation quality depends jointly on latent dimension, state representation, and the probability family used for stochastic evolution.

Additional distributional diagnostics with pointwise deviation of the predicted conditional distribution \Cref{fig:cir-cdf-deviation} and recursively generated sample paths are reported in \Cref{app:cir-diagnostics,app:cir-recursion}.

\subsection{Stochastic Lorenz: data-driven transition learning and comparison with Neural Stochastic Flows}
\label{subsec:lorenz}

The linear-SDE and CIR experiments study stochastic solution operators in settings where the relevant transition laws are analytically available. We next remove this analytical information and learn the stochastic transition directly from simulated trajectories of a nonlinear dynamical system. This setting also enables a direct comparison with Neural Stochastic Flows \citep{kiyohara2026neural}.
\begin{wrapfigure}{r}{0.45\textwidth}
    \centering
    \includegraphics[width=0.95\linewidth,clip,trim=0 0 0 2cm]{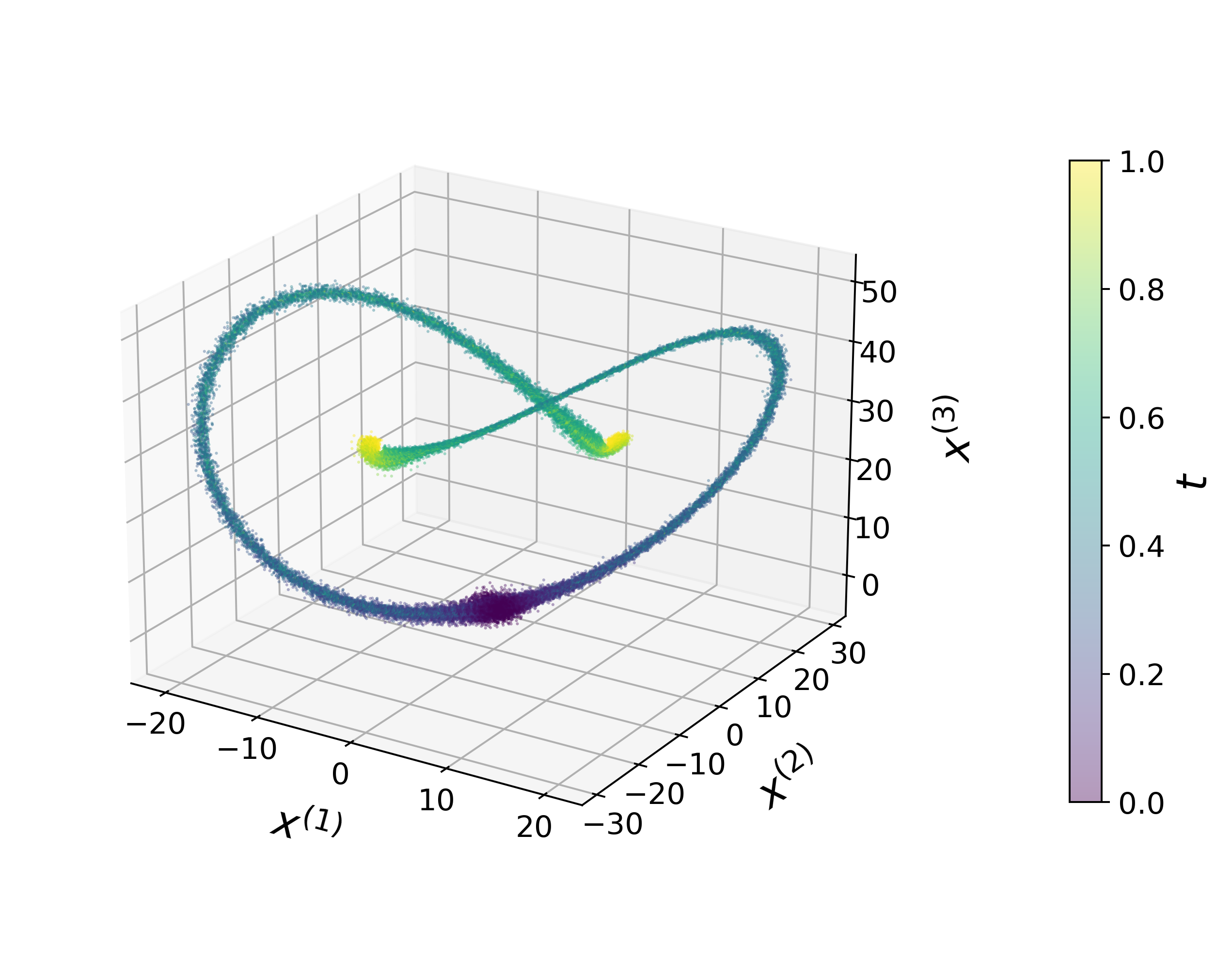}
    \caption{Evolution of the stochastic Lorenz state distribution over $t\in[0,1]$. Observed states are colored by time. The initially concentrated population near the unstable equilibrium evolves toward the bimodal structure associated with the two lobes of the Lorenz attractor.}
    \label{fig:lorenz-pointcloud}
    \vspace{-1cm}
\end{wrapfigure}

We follow the stochastic Lorenz-63 setup of \citep{li2020scalable}, used by Neural Stochastic Flows \citep{kiyohara2026neural}. 
Although the Lorenz system is a classic example of chaotic dynamics \citep{lorenz1963deterministic}, the finite-time experiment considered here focuses on the evolution from a Gaussian distribution centered at the unstable equilibrium to a pronounced bimodal distribution associated with the two lobes of the Lorenz attractor; see \cref{fig:lorenz-pointcloud}. This finite-time distributional bifurcation provides a data-driven test of the learned conditional transition laws.

The physical state $X_t\in\mathbb R^3$ satisfies
\begin{equation}\label{eq:lorenz-sde}
    \mathrm dX_t =
    \begin{pmatrix}
        \sigma(X_t^{(2)}-X_t^{(1)}) \\
        X_t^{(1)}(\rho-X_t^{(3)})-X_t^{(2)} \\
        X_t^{(1)}X_t^{(2)}-\beta X_t^{(3)}
    \end{pmatrix}
    \mathrm dt + \alpha\,\mathrm dW_t,
\end{equation}
with $\sigma=10$, $\rho=28$, $\beta=8/3$, and $\alpha=0.15$. Initial states are sampled from $X_0\sim\mathcal N(0,I)$, and trajectories are observed at $41$ equally spaced times on $[0,1]$. Training uses transition pairs from all pairs of observation times $s<t$. Further data generation details are given in \cref{app:lorenz-data}.

In contrast to the preceding linear-in-state latent transitions, we use
\[
    m(z_s,s,t) = \mathcal N\!\left( \mu(z_s,s,t), L(z_s,s,t)L(z_s,s,t)^\top \right),
\]
where a neural network maps $(z_s,s,t)$ directly to the conditional mean and covariance factor, allowing both to depend nonlinearly on the source state. The model is trained from observed transition pairs using reconstruction together with the physical-space energy score \Cref{eq:energy-score-loss}; the latent dimension is $n_z=8$, and architecture and training details are given in \Cref{app:lorenz-training}.

On $4{,}000$ test transitions sampled across the complete set of source and target times, the Stochastic Latent Twin attains predictive mean root-mean-squared error $1.69$ and energy score  $0.238$ in physical coordinates (\cref{tab:lorenz-nsf}). In addition to these transition-level measures, we examine the direct $0\to 1$ prediction, for which the initially concentrated population has developed the strongly bimodal structure of the Lorenz attractor. \Cref{fig:lorenz-t1density} shows that a single evaluation of the learned stochastic solution operator reproduces the two dominant modes and their spatial distribution without intermediate time stepping.

\begin{figure}
    \centering
    \includegraphics[width=\linewidth]{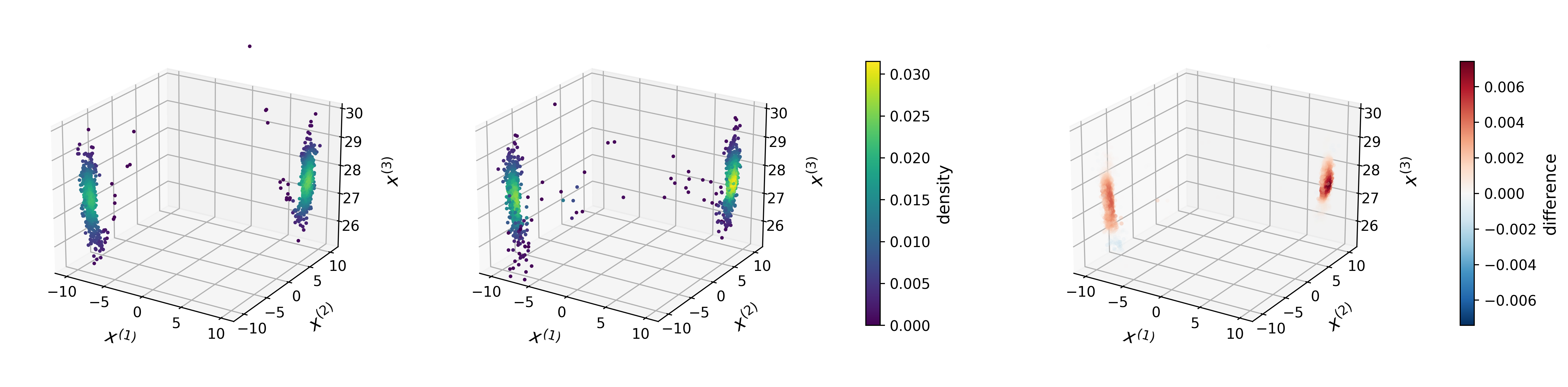}
    \caption{Distribution of the Lorenz state at $t=1$ for the $1{,}024$ test trajectories colored by kernel density estimates. Left: reference stochastic Lorenz process. Middle: direct Stochastic Latent Twin prediction $\Psi_{0,1}(x_0)$ from the same initial states. Right: difference between the estimated densities. The learned transition reproduces the bimodal population associated with the two lobes of the Lorenz attractor.}
    \label{fig:lorenz-t1density}
\end{figure}

For a direct comparison, we train Neural Stochastic Flows \citep{kiyohara2026neural} using the authors' reference implementation on the same data; its flow and bridge networks have approximately the same number of trainable parameters as the Stochastic Latent Twin above. At a common $100$-epoch training budget, the Stochastic Latent Twin attains smaller predictive-mean root-mean-squared error and energy score than the Neural Stochastic Flow, and a smaller $1$-Wasserstein distance between the $x^{(1)}$ marginals at $t=1$. Additional implementation details and results at shorter training budgets are reported in \cref{app:lorenz-competitors}.

\begin{table}
\centering
\caption{Stochastic Lorenz prediction at a common $100$-epoch training budget. Both models have approximately the same number of trainable parameters. All errors are evaluated in physical coordinates; root-mean-squared error of the predictive mean (RMSE) and energy score use the same $4{,}000$ test transitions and are computed from $64$ one-shot predictive samples per transition. The final column reports the $1$-Wasserstein distance between the predicted and reference $x^{(1)}$ marginals at $t=1$, from direct $0\to1$ predictions for the $1{,}024$ test trajectories. Wall-clock times refer to the present CPU implementations and are reported for reference rather than as optimized computational benchmarks.}
\label{tab:lorenz-nsf}
\begin{tabularx}{\linewidth}{@{}Xrrrrr@{}}
\toprule
Method & Parameters & Training time & RMSE & Energy & $W_1$ at $t=1$ \\
\midrule
Stochastic Latent Twin
    & $50{,}743$ & $4{,}522$ s
    & $\mathbf{1.69}$ & $\mathbf{0.238}$ & $\mathbf{0.112}$ \\
Neural Stochastic Flows
    & $51{,}540$ & $18{,}719$ s
    & $1.89$ & $0.310$ & $0.255$ \\
\bottomrule
\end{tabularx}
\end{table}

The Lorenz experiment demonstrates data-driven arbitrary-time prediction for a nonlinear stochastic dynamical system. We next use the same conditional solution-operator viewpoint in the backward direction, where the learned transition serves as a one-shot generative map.
\subsection{Swiss roll: Stochastic Latent Twin as a one-shot generative model} \label{subsec:swissroll}

The preceding experiments consider forward conditional evolution. We next turn to the backward conditional setting discussed in
\Cref{subsec:forward-backward} and test whether a Stochastic Latent Twin can learn a non-Gaussian backward transition entirely from data and use the resulting stochastic solution operator for one-shot generation. 

The target is a two-dimensional spiral projection of the Swiss-roll benchmark \citep{tenenbaum2000global,pedregosa2011scikit}. We corrupt the data on $t\in[0,1]$ with the Ornstein--Uhlenbeck (OU) process
\begin{equation}\label{eq:swiss-roll-forward}
    \mathrm dX_t = -\lambda X_t\,\mathrm dt + \sqrt{2\lambda}\,\mathrm dW_t, \qquad \lambda=2, 
\end{equation}
for which the distribution at $t=1$ is numerically on the same scale as $\mathcal N(0,I)$ under the finite-sample diagnostic used below. Details are given in \Cref{app:swissroll-data}.

The Stochastic Latent Twin is trained backward: later, more strongly corrupted states are used as sources and earlier states as targets. Although the forward OU transition is Gaussian, the backward conditional law depends on the non-Gaussian data distribution and is learned here only from sampled transitions. We use a Gaussian latent transition law with latent dimension $n_z=6$ and a nonlinear encoder and decoder, as motivated by the CIR experiment. The prediction loss is the physical energy score \Cref{eq:energy-score-loss}, together with an auxiliary sliced-Wasserstein term that matches one-shot samples from $\mathcal N(0,I)$ to the data at $t=0$, described in \Cref{app:swissroll-training}.

In the Stochastic Latent Twin, for one-shot generation, we draw
\[
    X_1\sim\mathcal N(0,I),
    \qquad X_0\sim\Psi_{1,0}(X_1).
\]
Thus each generated sample requires one direct source-to-target evaluation; no intermediate denoising or recursively composed transitions are used. \Cref{fig:swiss-roll-t0-heatmap} shows the resulting distribution. The one-shot Stochastic Latent Twin attains an average sliced-Wasserstein distance of $0.086\pm0.017$, compared with $0.090\pm0.014$ between independent target samples and $0.231\pm0.011$ for the untransformed standard-normal source.  The reference reflects sampling noise at 800 samples rather than a lower bound: discrepancy between the model and the target becomes apparent with larger samples (\cref{fig:swiss-roll-t0-heatmap}).

\begin{figure}
    \centering
    \includegraphics[width=\linewidth]{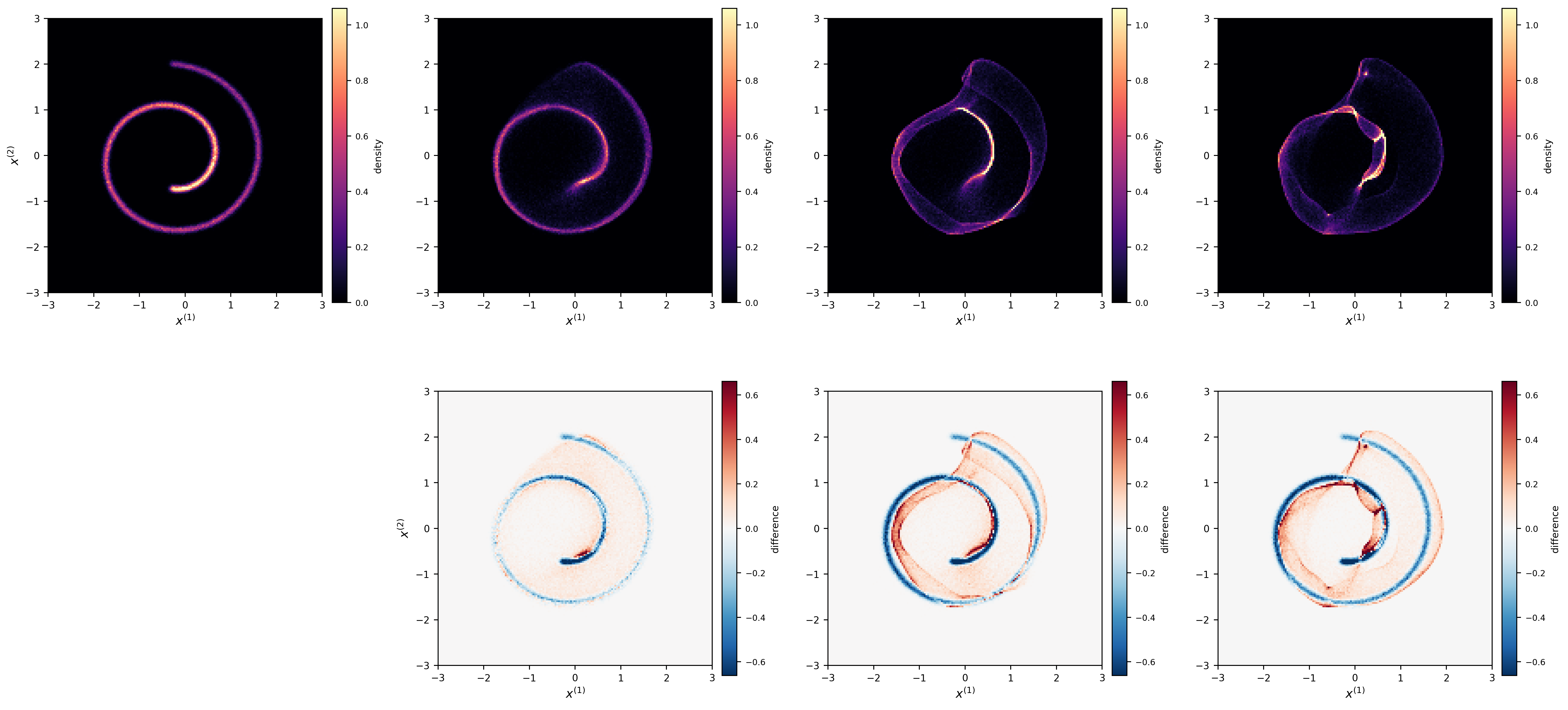}
    \caption{One-shot generation of the Swiss-roll target from $\mathcal N(0,I)$. Top row, from left to right: target density and generated densities for the Stochastic Latent Twin, MeanFlow, and SoFlow.  The images show two-dimensional histograms of 100,000 samples on a $160 \times 160$ grid over $[-3,3]^2$ for each respective setting. Bottom row: differences between each generated density and the target density.}
    \label{fig:swiss-roll-t0-heatmap}
\end{figure}

For context, we additionally compare the same one-shot generation task with MeanFlow \citep{geng2026mean} and SoFlow \citep{luo2026soflow}, trained from scratch on the same data at closely matched parameter counts, each with its own objective and without a sliced-Wasserstein term. The three methods all generate from $\mathcal N(0,I)$ in one network evaluation; their learned objects are different, however: MeanFlow learns an average velocity, SoFlow learns a deterministic solution-flow map, whereas the Stochastic Latent Twin learns conditional probability-valued transitions over arbitrary source and target times. \Cref{tab:swiss-roll-comparison} summarizes the one-shot generation results; implementation and training details are given in \cref{app:swissroll-competitors}.

\begin{table}
\centering
\caption{One-shot generation of the Swiss-roll target from $\mathcal N(0,I)$. Global sliced-Wasserstein distance (128 projections) is reported as mean $\pm$ one standard deviation over $10$ repeated batches of $800$ samples; lower is better. Wall-clock training times are reported for reference; the implementations were not optimized for computational efficiency.}
\label{tab:swiss-roll-comparison}
\begin{tabular}{lrrrr}
\toprule
Method & Parameters & Steps & Training time & Global SWD \\
\midrule
Stochastic Latent Twin & $14{,}017$ & $20{,}000$  & $243.7$ s & \textbf{$0.086\pm0.017$} \\
MeanFlow                & $14{,}188$ & $20{,}000$ & $60.8$ s  & $0.156\pm0.025$ \\
SoFlow                  & $14{,}188$ & $20{,}000$ & $47.0$ s  & $0.148\pm0.023$ \\
\bottomrule
\end{tabular}
\end{table}

The comparison here is primarily intended to assess one-shot generation accuracy at closely matched model sizes.

\subsection{Stochastic pendulum: physics conformity by construction} \label{subsec:pendulum}

We consider the stochastic pendulum
\begin{equation}\label{eq:stochastic-pendulum}
    \begin{aligned}
        \mathrm dq_t &= \left(p_t-\tfrac{\sigma^2}{2}\sin q_t\right)\mathrm dt +\sigma p_t\,\mathrm dW_t,\\
        \mathrm dp_t &= \left(-\sin q_t-\tfrac{\sigma^2}{2}p_t\cos q_t\right)\mathrm dt -\sigma\sin q_t\,\mathrm dW_t,
    \end{aligned}
    \qquad \text{with }\sigma=0.70.
\end{equation}

The stochastic forcing is tangent to the Hamiltonian energy manifolds, and the additional It\^o drift terms cancel the quadratic-variation contribution. Thus, the dynamics preserve
\[
    H(q_t,p_t)=\tfrac12p_t^2+1-\cos q_t
\]
pathwise: for every realization of the exact stochastic dynamics,
$H(q_t,p_t)=H(q_s,p_s)$ for all $t\geq s$. Hence, the dynamics remain genuinely stochastic: randomness changes the position along the orbit without changing its energy. This provides a direct test of the distributional physics-conformity framework of \cref{subsec:physics}. Initial conditions are sampled so that $H(q_0,p_0)\in[0.15,1.80]$, providing trajectories across a range of invariant energy manifolds; further sampling details are given in \cref{app:pendulum-data}.

We compare three Stochastic Latent Twins with the same representation and stochastic-transition architecture. The baseline \emph{Stochastic Latent Twin} uses no physical constraint. The \emph{Stochastic Latent Twin with soft conformity} promotes both latent invariance and latent-to-physical compatibility through the penalties of \Cref{eq:latent-conservation-loss,eq:pullback-loss}. In the \emph{Stochastic Latent Twin with conformity by construction}, the first latent coordinate is instead preserved exactly by the stochastic transition, while the remaining latent coordinates evolve stochastically. The pullback relation then aligns this invariant latent observable with the physical Hamiltonian. Thus, latent conformity holds exactly by construction, while its correspondence with the physical conserved quantity is learned from data. This is the setting of \Cref{cor:physics-conformity-construction} with $\varepsilon_{\rm latent}=0$ and a learned, nonzero pullback error. All three models are trained using reconstruction and the physical-space energy score \Cref{eq:energy-score-loss}; implementation details are given in \cref{app:pendulum-training}.

\begin{table}
\centering
\caption{Stochastic pendulum at horizon $t-s=1.5$. Conditional sliced-Wasserstein distance and physical energy defect are evaluated over eight held-out source states with $H\in[0.15,1.80]$; lower is better. The latent conformity defect is reported for the physics-conforming models, with zero corresponding to exact latent conservation by construction.}
\label{tab:pendulum-headline}
\begin{tabularx}{\linewidth}{@{}Xrrr@{}}
\toprule
Model
& Conditional SWD
& Physical energy defect
& Latent conformity defect \\
\midrule
SLT
& $0.1192$
& $1.144\times10^{-1}$
& -- \\
SLT with soft conformity
& $0.1096$
& $1.629\times10^{-2}$
& $3.551\times10^{-3}$ \\
SLT with conformity by construction
& $\mathbf{0.1026}$
& $\mathbf{1.538\times10^{-2}}$
& $\mathbf{0}$ (exact) \\
\bottomrule
\end{tabularx}
\end{table}

Both physics-conforming variants reduce the physical energy defect by approximately a factor of seven relative to the unconstrained Stochastic Latent Twin (\Cref{tab:pendulum-headline}) while retaining comparable conditional-law accuracy. Enforcing the latent invariant by construction removes the latent conformity defect exactly and gives the smallest physical energy defect among the three variants. Additional conformity diagnostics are reported in \Cref{app:pendulum-diagnostics}.

\subsection{Scalar two-dimensional stochastic Burgers equation}
\label{subsec:burgers}

The preceding experiments consider low-dimensional stochastic systems. We conclude our numerical investigations by testing whether the Stochastic Latent Twin framework can learn conditional laws for a spatially distributed stochastic PDE while retaining an exact physical conservation law. To this end, we consider the two-dimensional viscous Burgers equation with stochastic forcing,
\begin{equation} \label{eq:2dBurg}
    \d U_t = \left( -U_t\,\partial_{\xi_1}U_t -U_t\,\partial_{\xi_2}U_t +\nu\Delta U_t \right)\d t + \sigma\,\nabla\cdot \left( f(\xi_1,\xi_2)\,\d W_t \right),
\end{equation}
on $\Omega=[0,2\pi)^2$ with periodic boundary conditions. Here $W_t=(W_t^{(1)},\ldots,W_t^{(4)})$ is a standard Brownian motion in $\mathbb R^4$, and the stochastic forcing acts through the fixed smooth periodic flux matrix
\[
    f({\xi_1},{\xi_2}) =
    \begin{pmatrix}
        \cos {\xi_1} & 0 & \cos({\xi_1}+{\xi_2}) & \sin({\xi_1}-{\xi_2}) \\
        0 & \sin {\xi_2} & \cos({\xi_1}+{\xi_2}) & -\sin({\xi_1}-{\xi_2})
    \end{pmatrix}.
\]
Thus, $f({\xi_1},{\xi_2})\,\d W_t$ defines a stochastic flux with fixed spatial structure and random temporal amplitudes.

A key feature of \Cref{eq:2dBurg} is that it preserves mass pathwise. Defining the conserved spatial mean by $\mathcal C(U_t)=|\Omega|^{-1}\int_\Omega U_t(\xi_1,\xi_2)\,\d\xi_1\,\d\xi_2$, integration of \Cref{eq:2dBurg} over the periodic domain gives $\d \mathcal C(U_t)=0$, since the deterministic and stochastic flux terms have zero spatial integral by periodicity; see \Cref{app:burgers-numerics} for details. Hence, $\mathcal C(U_t)= \mathcal C(U_s)$ for any $s, t>0$, so the stochastic dynamics preserve mass exactly along every realization. This extends the physics-conformity test of \Cref{subsec:pendulum} from a low-dimensional invariant Hamiltonian system to a spatially distributed nonlinear stochastic PDE.

Trajectories of \Cref{eq:2dBurg} are simulated on a $n_{\mathrm g}\times n_{\mathrm g}$ periodic grid with $n_{\mathrm g}=64$, using Euler--Maruyama time integration and a conservative finite-volume discretization with local Lax--Friedrichs flux; further details are given in \Cref{app:burgers-numerics}. For simplicity, we denote the grid-discretized state at time $t$ by $u_t\in\mathbb R^{n_{\mathrm g}^2}$ and do not distinguish it from the continuous state unless relevant. Initial conditions are smooth random fields consisting of several large-scale blobs with random amplitudes, so different trajectories carry different conserved masses.

Because every term in the finite-volume update is a difference of face fluxes that cancel in pairs on the periodic grid, the discrete dynamics preserve the cell sum up to floating-point rounding. The discrete counterpart of the conserved quantity is $\mathcal C(u) = \frac{1}{n_{\mathrm g}^2} \sum_{i=1}^{n_{\mathrm g}^2}u^{(i)}$. Thus, the discretization preserves the same mass-conservation property as the continuous dynamics.

We compare a standard Stochastic Latent Twin with a physics-conforming variant that preserves $\mathcal C(u_t)$ by construction. Define $\mathcal C_{\mathcal Z}(z)=\tfrac{1}{n_z}\sum_{i=1}^{n_z}z^{(i)}$. Given the encoder output $e(u_s)$, we set
\[
    z_s = e(u_s) - \bigl(\mathcal C_{\mathcal Z}(e(u_s))-\mathcal C(u_s)\bigr)\mathbf 1_{n_z},
\]
so that $\mathcal C_{\mathcal Z}(z_s)=\mathcal C(u_s)$. Likewise, a sampled latent state $\widetilde Z_t$ is corrected by
\[
    Z_t = \widetilde Z_t - \bigl(\mathcal C_{\mathcal Z}(\widetilde Z_t)-\mathcal C_{\mathcal Z}(z_s)\bigr)
    \mathbf 1_{n_z}.
\]
The decoder output is shifted analogously so that $\mathcal C(d(z)) = \mathcal C_{\mathcal Z}(z)$ for any latent state $z$. Hence, $\mathcal C(d(Z_t))=\mathcal C_{\mathcal Z}(Z_t)=\mathcal C_{\mathcal Z}(z_s)=\mathcal C(u_s)$, and every predicted sample conserves mass exactly by construction.

We compare the standard Stochastic Latent Twin, the physics-conforming Stochastic Latent Twin, and the Latent Neural Stochastic Flow (Latent NSF) of \citet{kiyohara2026neural} across $32$ held-out initial conditions; see \Cref{fig:burgers2d}. For each initial condition and evaluation horizon, the reference conditional law is represented by $256$ independent solver realizations. As a finite-sample baseline, a second independent reference ensemble of the same size is compared with the first. Training, implementation and evaluation details are given in \Cref{app:burgers-training,app:burgers-eval}.

The left panel of \Cref{fig:burgers2d} shows the pointwise conditional medians for a representative held-out initial condition. All three models recover the dominant spatial structure over the training window, with the largest discrepancies concentrated near sharp fronts, where uncertainty in front location produces large pointwise errors. These discrepancies are most pronounced for the present Latent Neural Stochastic Flow implementation.

The right column of \Cref{fig:burgers2d} summarizes four metrics across all $32$ held-out initial conditions: pointwise-median error, sliced-Wasserstein distance, spread ratio, and conservation defect of $\mathcal C(u)$. Both Stochastic Latent Twin variants attain smaller median errors and sliced-Wasserstein discrepancies than the present Latent Neural Stochastic Flow implementation over the tested horizons. Their spread ratios remain close to one at short horizons and fall slightly below one at longer horizons, indicating mild underdispersion. Since this statistic averages the pointwise variance over space, it does not measure where the uncertainty is located. By construction, the physics-conforming Stochastic Latent Twin preserves $\mathcal C(u)$ to floating-point accuracy at every horizon, whereas the unconstrained Stochastic Latent Twin and Latent Neural Stochastic Flow exhibit conservation defects.

Thus, the physics-conforming Stochastic Latent Twin learns the conditional law of this spatially distributed nonlinear stochastic PDE while enforcing exact mass conservation by construction, without degrading predictive accuracy.

\begin{figure}
    \centering
    \includegraphics[width = 0.9\linewidth]{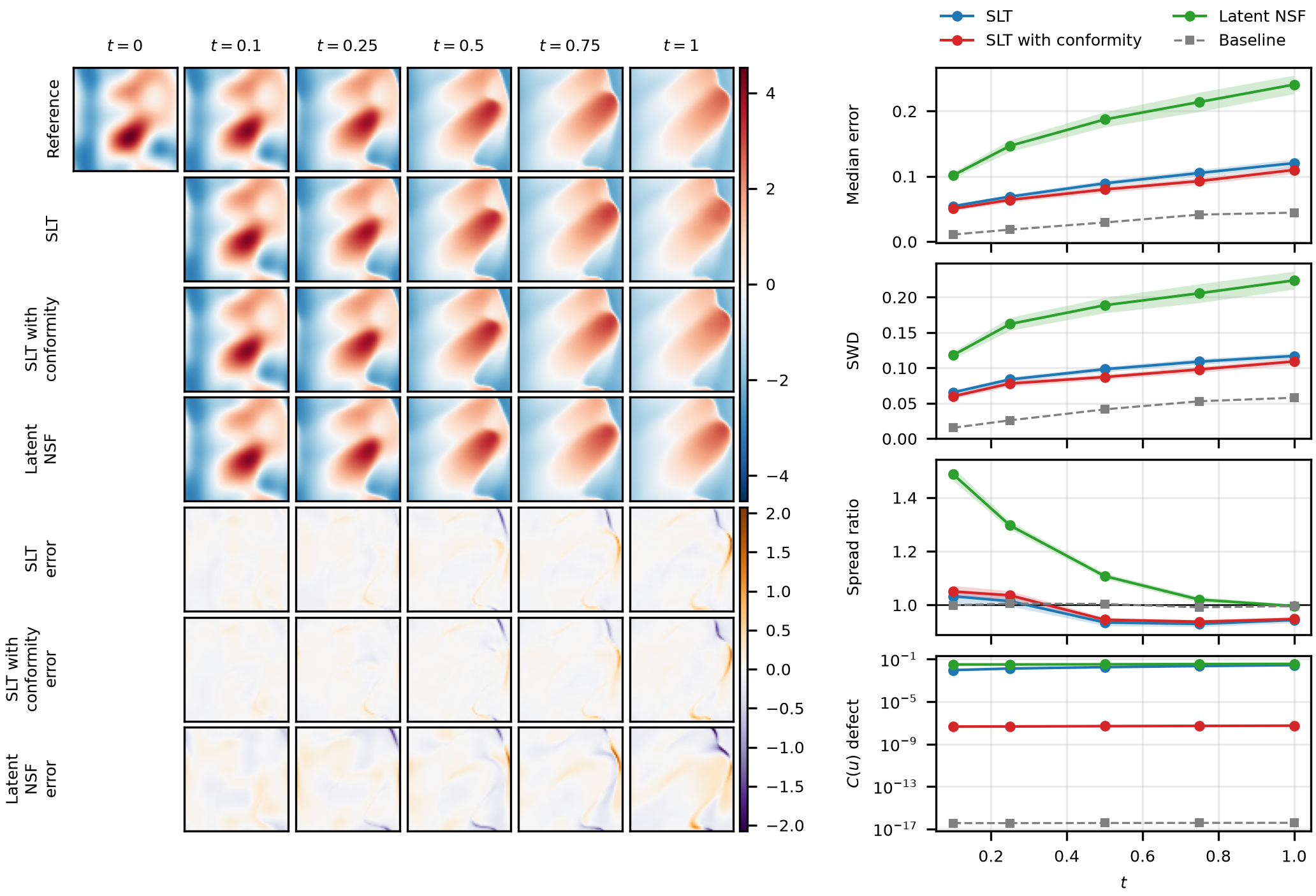}
    \caption{\emph{Stochastic two-dimensional Burgers equation.} \emph{Left:} Pointwise conditional medians for one held-out initial condition at $t=0.1,0.25,0.5,0.75,1.0$, with the initial condition at $t=0$ in the first column. The first row shows the reference ensemble median, the next three rows the medians predicted by the standard Stochastic Latent Twin, physics-conforming Stochastic Latent Twin, and Latent NSF, and the final three rows their errors relative to the reference. \emph{Right:} Median error, sliced-Wasserstein distance, spread ratio, and conservation defect $\mathcal C(u)$, averaged over $32$ held-out initial conditions. Each comparison uses $256$ predicted and $256$ reference samples; shading denotes $\pm1$ standard error. The baseline compares two independent reference ensembles of the same size. Lower is better for the median, Wasserstein, and conservation errors, while a spread ratio of one is ideal.}    
    \label{fig:burgers2d}
\end{figure}

\section{Conclusion and outlook}\label{sec:conclusion}

In this work, we introduced \emph{Stochastic Latent Twins}, a latent solution-operator framework for learning conditional probability laws of stochastic dynamical systems between arbitrary source and target times. A deterministic encoder--decoder separates state representation from stochastic temporal evolution, while the latent evolution map is probability-valued and may be realized by different transition families. We established an end-to-end Wasserstein approximation bound that separates reconstruction and latent-transition errors, proved stability with respect to perturbations of both source states and source distributions, and developed a distributional theory of physics conformity for conservation and dissipation. This theory provides latent-to-physical transfer bounds and shows how prescribed physical structure can be enforced by construction rather than recovered only through accurate solution-law approximation.

The numerical experiments were designed to isolate complementary aspects of this formulation. The linear-SDE example verifies the analytical reduced-approximation limits, while the Cox--Ingersoll--Ross experiment separates the roles of latent dimension, nonlinear state representation, and transition-family expressivity. The stochastic Lorenz experiment removes analytical access to the transition law and demonstrates fully data-driven arbitrary-time prediction, including a direct comparison with Neural Stochastic Flows. The Swiss-roll experiment uses the same conditional-law formulation in the backward direction for one-shot generation and compares it with recent one-step generative models. Finally, the stochastic pendulum and two-dimensional stochastic Burgers equation test physics conformity by construction, first for a low-dimensional Hamiltonian invariant and then for a spatially distributed stochastic PDE. In both settings, imposing the known physical structure substantially improves conformity without sacrificing conditional-law accuracy. For the pendulum, the latent invariant is conserved exactly and the physical energy defect falls by about a factor of seven; for Burgers, the construction preserves mass to round-off for every predicted sample.

Several directions remain open. These include connecting empirical distributional training more directly to the approximation theory, understanding long-time behavior and Chapman--Kolmogorov consistency, developing more expressive non-Gaussian transition models, and extending physics conformity beyond conserved quantities to dissipative structure. Scalable architectures for high-dimensional stochastic systems are another important direction.

Stochastic Latent Twins thus provide a path to scientific surrogates that carry uncertainty and physical structure along with the predicted state, which makes them candidates for data assimilation, inverse problems, and control.

\section*{Acknowledgments}
The authors thank Elizabeth Newman for helpful early discussions that contributed to the development of this work.

\section*{Code availability}
Code supporting the numerical experiments will be made publicly available on GitHub upon acceptance of the manuscript.

\bibliographystyle{cas-model2-names}

\bibliography{cas-refs}

@article{chung2026latent,
  title={Latent Twins},
  author={Chung, Matthias and Verma, Deepanshu and Collins, Max and Subrahmanya, Amit N and Katti Sastry, Varuni and Rao, Vishwas},
  journal={Machine Learning: Science and Technology},
  volume={7},
  number={2},
  pages={025039},
  year={2026},
  publisher={IOP Publishing}
}

@article{chung2027physics,
  title={Physics-conforming latent twins},
  author={Chung, Matthias and Bu, Yutong and Verma, Deepanshu},
  journal={Computer Methods in Applied Mechanics and Engineering},
  volume={463},
  pages={119351},
  year={2027},
  publisher={Elsevier}
}

@misc{verma2026latent,
      title={Latent Twin Operator}, 
      author={Deepanshu Verma and Riley Chen and Matthias Chung},
      year={2026},
      eprint={2609.35531},
      archivePrefix={arXiv},
      primaryClass={cs.CE},
      url={https://arxiv.org/abs/2609.35531}, 
}

@article{chen2024learning,
  author  = {Chen, Yuan and Xiu, Dongbin},
  journal = {Journal of Computational Physics},
  pages   = {112984},
  title   = {Learning Stochastic Dynamical System via Flow Map Operator},
  volume  = {508},
  year    = {2024}
}

@article{hasan2022identifying,
  author  = {Hasan, Ali and Pereira, Jo{\~a}o M. and Farsiu, Sina and Tarokh, Vahid},
  journal = {IEEE Transactions on Signal Processing},
  pages   = {89--104},
  title   = {Identifying Latent Stochastic Differential Equations},
  volume  = {70},
  year    = {2022}
}

@book{quarteroni2014scientific,
  author    = {Quarteroni, Alfio and Saleri, Fausto and Gervasio, Paola},
  title     = {Scientific Computing with {MATLAB} and {O}ctave},
  edition   = {4},
  series    = {Texts in Computational Science and Engineering},
  publisher = {Springer Berlin, Heidelberg},
  year      = {2014},
  isbn      = {978-3-642-45366-3},
  doi       = {10.1007/978-3-642-45367-0},
}

@inproceedings{kostic2024learning,
  author    = {Kostic, Vladimir R. and Novelli, Pietro and Grazzi, Riccardo and Lounici, Karim and Pontil, Massimiliano},
  booktitle = {International Conference on Learning Representations},
  title     = {Learning Invariant Representations of Time-Homogeneous Stochastic Dynamical Systems},
  year      = {2024}
}

@article{ke2025learning,
  title={Learning stochastic nonlinear dynamics with embedded latent transfer operators},
  author={Ke, Naichang and Tanaka, Ryogo and Kawahara, Yoshinobu},
  journal={arXiv preprint arXiv:2501.02721},
  year={2025}
}

@article{shi2026expanding,
  title={Expanding the chaos: Neural operator for stochastic (partial) differential equations},
  author={Shi, Dai and Lin, Lequan and Han, Andi and Thompson, Luke and Hern{\'a}ndez-Lobato, Jos{\'e} Miguel and Wang, Zhiyong and Gao, Junbin},
  journal={arXiv preprint arXiv:2601.01021},
  year={2026}
}

@article{gneiting2007strictly,
  title={Strictly proper scoring rules, prediction, and estimation},
  author={Gneiting, Tilmann and Raftery, Adrian E},
  journal={Journal of the American Statistical Association},
  volume={102},
  number={477},
  pages={359--378},
  year={2007},
  publisher={Taylor \& Francis}
}

@misc{kiyohara2025nsfcode,
  author       = {Kiyohara, Naoki and Johns, Edward and Li, Yingzhen},
  title        = {Neural Stochastic Flows: Official JAX Implementation},
  year         = {2025},
  howpublished = {\url{https://github.com/nkiyohara/jax-nsf}},
  note         = {Accessed 2026}
}

@article{pedregosa2011scikit,
  title={Scikit-learn: Machine learning in Python},
  author={Pedregosa, Fabian and Varoquaux, Ga{\"e}l and Gramfort, Alexandre and Michel, Vincent and Thirion, Bertrand and Grisel, Olivier and Blondel, Mathieu and Prettenhofer, Peter and Weiss, Ron and Dubourg, Vincent and others},
  journal={the Journal of Machine Learning Research},
  volume={12},
  pages={2825--2830},
  year={2011},
  publisher={JMLR. org}
}

@article{tenenbaum2000global,
  title={A global geometric framework for nonlinear dimensionality reduction},
  author={Tenenbaum, Joshua B and Silva, Vin de and Langford, John C},
  journal={science},
  volume={290},
  number={5500},
  pages={2319--2323},
  year={2000},
  publisher={American Association for the Advancement of Science}
}

@article{lorenz1963deterministic,
  author  = {Lorenz, Edward N.},
  title   = {Deterministic Nonperiodic Flow},
  journal = {Journal of the Atmospheric Sciences},
  volume  = {20},
  number  = {2},
  pages   = {130--141},
  year    = {1963},
  doi     = {10.1175/1520-0469(1963)020<0130:DNF>2.0.CO;2}
}

@inproceedings{rowland2019orthogonal,
  title={Orthogonal estimation of wasserstein distances},
  author={Rowland, Mark and Hron, Jiri and Tang, Yunhao and Choromanski, Krzysztof and Sarlos, Tamas and Weller, Adrian},
  booktitle={The 22nd International conference on artificial intelligence and statistics},
  pages={186--195},
  year={2019},
  organization={PMLR}
}

@inproceedings{frans2025one,
  title={One step diffusion via shortcut models},
  author={Frans, Kevin and Hafner, Danijar and Levine, Sergey and Abbeel, Pieter},
  booktitle={International Conference on Learning Representations},
  volume={2025},
  pages={34668--34684},
  year={2025}
}

@article{kiyohara2026neural,
  title={Neural stochastic flows: Solver-free modelling and inference for SDE solutions},
  author={Kiyohara, Naoki and Johns, Edward and Li, Yingzhen},
  journal={Advances in Neural Information Processing Systems},
  volume={38},
  pages={16917--16962},
  year={2025}
}

@inproceedings{luo2026soflow,
  title={SoFlow: Solution flow models for one-step generative modeling},
  author={Luo, Tianze and Yuan, Haotian and Liu, Zhuang},
  booktitle={International Conference on Learning Representations},
  volume={2026},
  pages={101623--101641},
  year={2026}
}

@article{geng2026mean,
  title={Mean flows for one-step generative modeling},
  author={Geng, Zhengyang and Deng, Mingyang and Bai, Xingjian and Kolter, Zico and He, Kaiming},
  journal={Advances in Neural Information Processing Systems},
  volume={38},
  pages={75460--75482},
  year={2025}
}

@article{bulte2025probabilistic,
  author  = {B{\"u}lte, Christopher and Scholl, Philipp and Kutyniok, Gitta},
  journal = {Transactions on Machine Learning Research},
  title   = {Probabilistic Neural Operators for Functional Uncertainty Quantification},
  year    = {2025}
}

@inproceedings{li2020scalable,
  title     = {Scalable Gradients for Stochastic Differential Equations},
  author    = {Li, Xuechen and Wong, Ting-Kam Leonard and Chen, Ricky T. Q. and Duvenaud, David},
  booktitle = {Proceedings of the Twenty Third International Conference on Artificial Intelligence and Statistics},
  series    = {Proceedings of Machine Learning Research},
  volume    = {108},
  pages     = {3870--3882},
  year      = {2020},
  publisher = {PMLR}
}

@inproceedings{kidger2021neuralsdes,
  title     = {Neural {SDE}s as Infinite-Dimensional {GAN}s},
  author    = {Kidger, Patrick and Foster, James and Li, Xuechen and Oberhauser, Harald and Lyons, Terry},
  booktitle = {Proceedings of the 38th International Conference on Machine Learning},
  series    = {Proceedings of Machine Learning Research},
  volume    = {139},
  pages     = {5453--5463},
  year      = {2021},
  publisher = {PMLR}
}

@article{golub1973differentiation,
  title={The differentiation of pseudo-inverses and nonlinear least squares problems whose variables separate},
  author={Golub, Gene H and Pereyra, Victor},
  journal={SIAM Journal on numerical analysis},
  volume={10},
  number={2},
  pages={413--432},
  year={1973},
  publisher={SIAM}
}

@article{chung2023variable,
  title={A variable projection method for large-scale inverse problems with $\ell_1$ regularization},
  author={Chung, Matthias and Renaut, Rosemary A},
  journal={Applied Numerical Mathematics},
  volume={192},
  pages={297--318},
  year={2023},
  publisher={Elsevier}
}

@inproceedings{course2023amortized,
  title     = {Amortized Reparametrization: Efficient and Scalable Variational Inference for Latent {SDE}s},
  author    = {Course, Kevin and Nair, Prasanth B.},
  booktitle = {Advances in Neural Information Processing Systems},
  volume    = {36},
  year      = {2023}
}

@book{bhatia1997matrix,
  author    = {Rajendra Bhatia},
  title     = {Matrix Analysis},
  series    = {Graduate Texts in Mathematics},
  volume    = {169},
  publisher = {Springer},
  year      = {1997},
  doi       = {10.1007/978-1-4612-0653-8}
}

@inproceedings{bartosh2025sdematching,
  title     = {{SDE} Matching: Scalable and Simulation-Free Training of Latent Stochastic Differential Equations},
  author    = {Bartosh, Grigory and Vetrov, Dmitry and Naesseth, Christian A.},
  booktitle = {Proceedings of the 42nd International Conference on Machine Learning},
  series    = {Proceedings of Machine Learning Research},
  volume    = {267},
  pages     = {3054--3070},
  year      = {2025},
  publisher = {PMLR}
}

@book{oksendal2003stochastic,
  author    = {\O{}ksendal, Bernt},
  title     = {Stochastic Differential Equations: An Introduction with Applications},
  edition   = {6},
  publisher = {Springer},
  address   = {Berlin, Heidelberg},
  year      = {2003},
  series    = {Universitext},
  doi       = {10.1007/978-3-642-14394-6}
}

@book{kloeden1992numerical,
  author    = {Kloeden, Peter E. and Platen, Eckhard},
  title     = {Numerical Solution of Stochastic Differential Equations},
  publisher = {Springer},
  address   = {Berlin, Heidelberg},
  year      = {1992},
  series    = {Applications of Mathematics},
  volume    = {23},
  doi       = {10.1007/978-3-662-12616-5}
}

@book{ethier2009markov,
  title={Markov processes: characterization and convergence},
  author={Ethier, Stewart N and Kurtz, Thomas G},
  year={2009},
  publisher={John Wiley \& Sons}
}

@article{gelbrich1990formula,
  title={On a formula for the L2 Wasserstein metric between measures on Euclidean and Hilbert spaces},
  author={Gelbrich, Matthias},
  journal={Mathematische Nachrichten},
  volume={147},
  number={1},
  pages={185--203},
  year={1990},
  publisher={Wiley Online Library}
}

@article{peyre2019computational,
  title={Computational optimal transport with applications to data sciences},
  author={Peyr{\'e}, Gabriel and Cuturi, Marco},
  journal={Foundations and trends{\textregistered} in machine learning},
  volume={11},
  number={5-6},
  pages={355--607},
  year={2019},
  publisher={Emerald Publishing Limited}
}

@article{bhatia2019bures,
  author  = {Bhatia, Rajendra and Jain, Tanvi and Lim, Yongdo},
  title   = {On the Bures--Wasserstein Distance between Positive Definite Matrices},
  journal = {Expositiones Mathematicae},
  volume  = {37},
  number  = {2},
  pages   = {165--191},
  year    = {2019},
  doi     = {10.1016/j.exmath.2018.01.002}
}

@inproceedings{song2021score,
  title     = {Score-Based Generative Modeling through Stochastic Differential Equations},
  author    = {Song, Yang and Sohl-Dickstein, Jascha and Kingma, Diederik P. and
               Kumar, Abhishek and Ermon, Stefano and Poole, Ben},
  booktitle = {International Conference on Learning Representations},
  year      = {2021}
}

@inproceedings{lipman2023flow,
  title     = {Flow Matching for Generative Modeling},
  author    = {Lipman, Yaron and Chen, Ricky T. Q. and Ben-Hamu, Heli and
               Nickel, Maximilian and Le, Matthew},
  booktitle = {International Conference on Learning Representations},
  year      = {2023}
}

@book{kaipio2005statistical,
  title     = {Statistical and Computational Inverse Problems},
  author    = {Kaipio, Jari and Somersalo, Erkki},
  series    = {Applied Mathematical Sciences},
  volume    = {160},
  publisher = {Springer},
  year      = {2005}
}

@article{stuart2010inverse,
  title   = {Inverse Problems: A Bayesian Perspective},
  author  = {Stuart, Andrew M.},
  journal = {Acta Numerica},
  volume  = {19},
  pages   = {451--559},
  year    = {2010}
}

@article{haussmann1986time,
  author  = {Haussmann, U. G. and Pardoux, {\'E}tienne},
  title   = {Time Reversal of Diffusions},
  journal = {The Annals of Probability},
  volume  = {14},
  number  = {4},
  pages   = {1188--1205},
  year    = {1986},
  doi     = {10.1214/aop/1176992362}
}

@book{daprato2014stochastic,
  author    = {Da Prato, Giuseppe and Zabczyk, Jerzy},
  title     = {Stochastic Equations in Infinite Dimensions},
  edition   = {2},
  series    = {Encyclopedia of Mathematics and its Applications},
  volume    = {152},
  publisher = {Cambridge University Press},
  year      = {2014}
}

@book{villani2009optimal,
  author    = {Villani, C{\'e}dric},
  title     = {Optimal Transport: Old and New},
  series    = {Grundlehren der mathematischen Wissenschaften},
  volume    = {338},
  publisher = {Springer},
  address   = {Berlin, Heidelberg},
  year      = {2009},
  doi       = {10.1007/978-3-540-71050-9}
}

@misc{zeng2026transition,
  author        = {Zeng, Li and Wan, Xiaoliang and Wang, Yaobin and Nobile, Fabio and Zhou, Tao},
  title         = {A Transition-Density-Based Operator Learning Method for Fokker--Planck Equations with Various Initial Conditions},
  year          = {2026},
  eprint        = {2606.09434},
  archivePrefix = {arXiv}
}

@article{su2025transition,
  author  = {Su, Haozhe and Tretyakov, M. V. and Newton, David P.},
  title   = {Deep Learning of Transition Probability Densities for Stochastic Asset Models with Applications in Option Pricing},
  journal = {Management Science},
  volume  = {71},
  number  = {4},
  pages   = {2922--2952},
  year    = {2024},
  doi     = {10.1287/mnsc.2022.01448}
}

@article{bahmani2025neuralchaos,
  author  = {Bahmani, Bahador and Kevrekidis, Ioannis G. and Shields, Michael D.},
  title   = {Neural Chaos: A Spectral Stochastic Neural Operator},
  journal = {Journal of Computational Physics},
  volume  = {539},
  pages   = {114233},
  year    = {2025},
  doi     = {10.1016/j.jcp.2025.114233}
}

@article{cox1985theory,
  title={A theory of the term structure of interest rates},
  author={Cox, John C and Ingersoll, Jonathan E and Ross, Stephen A and others},
  journal={Econometrica},
  volume={53},
  number={2},
  pages={385--407},
  year={1985},
  publisher={World Scientific}
}

@article{vallender1974calculation,
  title={Calculation of the Wasserstein distance between probability distributions on the line},
  author={Vallender, SS},
  journal={Theory of Probability \& Its Applications},
  volume={18},
  number={4},
  pages={784--786},
  year={1974},
  publisher={SIAM}
}

@article{wang2026relative,
  author  = {Wang, Binshuai and Wei, Peng},
  title   = {Relative Wasserstein Angle and the Problem of the
             {$W_2$}-Nearest Gaussian Distribution},
  journal = {arXiv preprint arXiv:2601.22355},
  year    = {2026}
}

\appendix

\section{Relation to existing stochastic learning frameworks}
\label{app:related-methods}

The related-work discussion in \Cref{sec:background} organizes existing approaches according to the mathematical object used to represent stochastic evolution. Here, we provide a more detailed structural comparison of the methods most closely related to Stochastic Latent Twins. \Cref{tab:related-stochastic-methods} summarizes the learned object, temporal representation, state representation, and principal structural or theoretical emphasis of each approach.

\begin{table*}[t]
\centering
\caption{\emph{Structural comparison of representative approaches for learning stochastic dynamics.} The methods differ primarily in the object learned, the representation of temporal evolution, and the accompanying structural or theoretical emphasis.} \label{tab:related-stochastic-methods}

\renewcommand{\arraystretch}{1.12}
\setlength{\tabcolsep}{3pt}
\scriptsize

\begin{tabularx}{\textwidth}{
>{\raggedright\arraybackslash}p{0.16\textwidth}
>{\raggedright\arraybackslash}p{0.21\textwidth}
>{\raggedright\arraybackslash}p{0.14\textwidth}
>{\raggedright\arraybackslash}p{0.20\textwidth}
>{\raggedright\arraybackslash}X
}
\toprule

\rowcolor{tablehead}
\textbf{Method} & \textbf{Learned object} & \textbf{Temporal access} & \textbf{State representation} & \textbf{Structural / theoretical emphasis} \\
\midrule

\textbf{Neural / latent SDEs}
\citep{li2020scalable,kidger2021neuralsdes,
hasan2022identifying,course2023amortized,bartosh2025sdematching}
&
Drift and diffusion, possibly in a learned latent space
&
Local dynamics + numerical propagation
&
Physical or latent stochastic state
&
SDE and variational formulations; properties depend on the
parameterization and learning setup
\\

\addlinespace[0.4ex]

\textbf{Stochastic flow-map learning}
\citep{chen2024learning}
&
Finite-interval stochastic evolution map
&
Direct over prescribed intervals
&
Physical state
&
Stochastic flow-map structure and numerical distributional validation
\\

\addlinespace[0.4ex]

\textbf{Neural Stochastic Flows}
\citep{kiyohara2026neural}
&
Conditional transition density via normalizing flows
&
Direct arbitrary $(s,t)$
&
Physical state; latent variant uses a variational state-space model
&
Identity and stationarity by construction; Chapman--Kolmogorov consistency
through regularization
\\

\addlinespace[0.4ex]

\textbf{Transition-density / Fokker--Planck learning}
\citep{su2025transition,zeng2026transition}
&
Transition density or probability-density evolution operator
&
Direct transition or density evolution
&
Physical state or function space
&
Kolmogorov and Fokker--Planck structure; transition-density analysis and
validation
\\

\addlinespace[0.4ex]

\textbf{Probabilistic / stochastic neural operators}
\citep{bulte2025probabilistic,bahmani2025neuralchaos,shi2026expanding}
&
Distribution-valued or stochastic operator between function spaces
&
Problem dependent
&
Function-space representation
&
Distributional operator learning; proper scoring and stochastic
representation
\\

\addlinespace[0.4ex]

\textbf{Stochastic transfer / Koopman operators}
\citep{kostic2024learning,ke2025learning}
&
Evolution of observables or learned representations
&
Operator propagation
&
Feature or observable representation
&
Statistical and operator-theoretic analysis, including long-time forecasting
\\

\specialrule{1.0pt}{0.45em}{0.35em}

\rowcolor{python1!10}
\textbf{\textcolor{python1}{Stochastic Latent Twins}}
&
Conditional latent transition law
$m(z_s,s,t)\in\mathcal P(\mathcal Z)$
&
Direct arbitrary $(s,t)$
&
Deterministic encoder--decoder; stochastic latent evolution
&
End-to-end Wasserstein approximation and stability; distributional
conservation and dissipation; latent-to-physical conformity transfer
\\

\bottomrule
\end{tabularx}
\end{table*}

\paragraph{Closest comparison: Neural Stochastic Flows.}

Among the approaches in \Cref{tab:related-stochastic-methods}, Neural
Stochastic Flows \citep{kiyohara2026neural} provide the closest comparison
to Stochastic Latent Twins. Both learn conditional transition laws that can be
queried directly between arbitrary source and target times. Neural Stochastic
Flows represent these laws using conditional normalizing flows and incorporate
structural properties of stochastic flows into the model: identity and, for
autonomous systems, stationarity are encoded through the parameterization,
while Chapman--Kolmogorov consistency is promoted through a bidirectional
KL-based regularizer. Their latent extension is formulated as a variational
state-space model.

The comparison with Stochastic Latent Twins, including the difference in theoretical scope, is given in \Cref{sec:background}, and direct numerical comparisons are given in \Cref{subsec:lorenz,subsec:burgers}.

\paragraph{Theoretical scope.}

Beyond this closest comparison, related frameworks address different
theoretical questions. Probabilistic Neural Operators justify distributional
training through strict propriety of the energy score
\citep{bulte2025probabilistic}, while stochastic transfer- and
Koopman-operator methods provide statistical and operator-theoretic guarantees,
including results for long-time forecasting
\citep{kostic2024learning,ke2025learning}. Transition-density approaches,
in turn, are closely tied to the analytical structure of the Kolmogorov and
Fokker--Planck equations
\citep{su2025transition,zeng2026transition}.

These perspectives are complementary to the approximation and structural
analysis developed here. The Stochastic Latent Twin results are conditional
on the stated reconstruction and latent-transition approximation assumptions;
they do not constitute statistical learning guarantees that finite-data
optimization necessarily attains these errors. This distinction is important,
since related approaches provide complementary guarantees concerning
distributional training objectives, statistical estimation, or long-time
forecasting.

\section{Optimal reduced approximation for linear stochastic systems}
\label{app:linear-low-rank}

This section provides the derivation of the linear-Gaussian results for $n_z<n_x$ summarized in \cref{subsec:linear-stochastic}. We first formulate the reduced Wasserstein approximation problem. We then show that, for any candidate encoder and decoder, the latent Gaussian transition can be eliminated analytically. Finally, by allowing the encoder and decoder to be optimized for one fixed source and target time, we derive a closed-form reference solution that provides a lower bound on the transition error attainable by a Stochastic Latent Twin using one encoder--decoder representation across multiple time pairs.

\subsection{Reduced Wasserstein approximation problem}
\label{app:linear-wasserstein-problem}

For fixed source and target times $s<t$, recall from \cref{eq:linear-transition,eq:linear-slt-transition} that $\Phi_{s,t}(x_s)=\mathcal N(M_{s,t}x_s,Q_{s,t})$ and
\[
    \Psi_{s,t}(x_s) = \mathcal N\!\left( DM_{s,t}^{\mathcal Z}Ex_s, DQ_{s,t}^{\mathcal Z}D^\top \right).
\]
Let $C_s:=\mathbb E[X_sX_s^\top]$ denote the second-moment matrix of the source state. We consider the reduced approximation problem
\begin{equation} \label{eq:linear-gaussian-w2-objective}
\begin{aligned}
    \min_{E,D,M_{s,t}^{\mathcal Z},Q_{s,t}^{\mathcal Z}}    \qquad& \mathcal J_{s,t} \left( E,D, M_{s,t}^{\mathcal Z}, Q_{s,t}^{\mathcal Z} \right) = \mathbb E_{X_s} \left[ W_2^2\!\left( \Psi_{s,t}(X_s), \Phi_{s,t}(X_s) \right) \right].\\
    \text{subject to}\quad& Q_{s,t}^{\mathcal Z}\succeq0, \qquad \operatorname{rank}(E) = \operatorname{rank}(D) = n_z.
\end{aligned}
\end{equation}

Since both conditional transition distributions are Gaussian, the objective in \cref{eq:linear-gaussian-w2-objective} can be evaluated explicitly as
\begin{equation} \label{eq:linear-gaussian-w2-decomposition}
    \mathcal J_{s,t} =  \left\| \left( M_{s,t} - DM_{s,t}^{\mathcal Z}E \right) C_s^{1/2} \right\|_{\mathrm F}^2
    + d_{\rm B}^2 \left( DQ_{s,t}^{\mathcal Z}D^\top, Q_{s,t} \right),
\end{equation}
where
\begin{equation} \label{eq:bures-distance}
    d_{\rm B}^2(S,Q) = \operatorname{tr}(S) + \operatorname{tr}(Q) - 2\operatorname{tr} \left( \left( Q^{1/2}SQ^{1/2} \right)^{1/2} \right)
\end{equation}
is the squared Bures--Wasserstein distance between positive semidefinite covariance matrices \citep{gelbrich1990formula,bhatia2019bures,peyre2019computational}.

Thus, \cref{eq:linear-gaussian-w2-objective} is equivalently the constrained minimization of the sum of the conditional-mean error and the covariance error in \cref{eq:linear-gaussian-w2-decomposition}.

\subsection{Analytical elimination of the latent Gaussian transition}
\label{app:latent-transition-elimination}

We now consider the inner optimization over $M_{s,t}^{\mathcal Z}$ and $Q_{s,t}^{\mathcal Z}$ for a candidate encoder $E$ and decoder $D$. This minimization admits a closed-form solution. 

\begin{proposition}[Optimal latent Gaussian transition] \label{prop:optimal-latent-gaussian-transition}
Let $E\in\mathbb R^{n_z\times n_x}$ and $D\in\mathbb R^{n_x\times n_z}$ be given, with $D^\top D=I$, and suppose that $EC_sE^\top$ is nonsingular. Then, for this encoder and decoder, the unique minimizers of
\[
    \mathcal J_{s,t} \left( E,D, M_{s,t}^{\mathcal Z}, Q_{s,t}^{\mathcal Z} \right)
\]
over $M_{s,t}^{\mathcal Z}$ and $Q_{s,t}^{\mathcal Z}\succeq0$ are
\begin{equation} \label{eq:optimal-latent-gaussian}
    \widehat M_{s,t}^{\mathcal Z} = D^\top M_{s,t}C_sE^\top \left( EC_sE^\top \right)^{-1}
    \qquad \text{and} \qquad
    \widehat Q_{s,t}^{\mathcal Z} = D^\top Q_{s,t}D.
\end{equation}
\end{proposition}

\begin{proof}
We first note that the assumption $D^\top D=I$ entails no loss of generality. Indeed, any full-column-rank decoder $\widetilde D$ admits a thin QR factorization $\widetilde D=DR$ with $D^\top D=I$ and $R\in\mathbb R^{n_z\times n_z}$ invertible. Under the corresponding change of latent coordinates $E_R:=RE$, $M_{s,t,R}^{\mathcal Z}:=RM_{s,t}^{\mathcal Z}R^{-1}$, and $Q_{s,t,R}^{\mathcal Z}:=RQ_{s,t}^{\mathcal Z}R^\top$, we have
\[
    DM_{s,t,R}^{\mathcal Z}E_R = \widetilde D M_{s,t}^{\mathcal Z}E, \qquad DQ_{s,t,R}^{\mathcal Z}D^\top = \widetilde DQ_{s,t}^{\mathcal Z}\widetilde D^\top.
\]
Hence the induced physical-space conditional transition distribution is unchanged, so restricting the analysis to $D^\top D=I$ does not restrict the model class. 

From \cref{eq:linear-gaussian-w2-decomposition}, the objective can then be written as
\begin{align}
    \mathcal J_{s,t} \left( E,D,M_{s,t}^{\mathcal Z},Q_{s,t}^{\mathcal Z} \right)
    ={}& \left\| \left( M_{s,t}-DM_{s,t}^{\mathcal Z}E \right) C_s^{1/2} \right\|_{\mathrm F}^2 \nonumber\\
    &+ \operatorname{tr} \left( Q_{s,t}^{\mathcal Z} \right) + \operatorname{tr} \left( Q_{s,t} \right) - 2\operatorname{tr} \left[ \left( (Q_{s,t}^{\mathcal Z})^{1/2} D^\top Q_{s,t}D (Q_{s,t}^{\mathcal Z})^{1/2} \right)^{1/2} \right].
\label{eq:linear-gaussian-objective-expanded}
\end{align}
The first term depends on $M_{s,t}^{\mathcal Z}$ but not on $Q_{s,t}^{\mathcal Z}$, whereas the remaining terms depend on $Q_{s,t}^{\mathcal Z}$ but not on $M_{s,t}^{\mathcal Z}$. Hence the two latent transition parameters may be optimized separately.

\begin{enumerate}
    \item Expanding the conditional-mean contribution gives
    \begin{align}
        \left\| \left( M_{s,t}-DM_{s,t}^{\mathcal Z}E \right)C_s^{1/2} \right\|_{\mathrm F}^2 = \operatorname{tr} \left( M_{s,t}C_sM_{s,t}^\top \right) - 2\operatorname{tr} \left( M_{s,t}^{\mathcal Z} EC_sM_{s,t}^\top D \right) + \operatorname{tr} \left( M_{s,t}^{\mathcal Z} EC_sE^\top (M_{s,t}^{\mathcal Z})^\top \right).
    \end{align}
    The first term is independent of $M_{s,t}^{\mathcal Z}$. Differentiating the remaining quadratic expression with respect to $M_{s,t}^{\mathcal Z}$ and imposing stationarity gives
    \[
        \widehat M_{s,t}^{\mathcal Z}EC_sE^\top = D^\top M_{s,t}C_sE^\top.
    \]
    Since $EC_sE^\top$ is nonsingular, the quadratic term is strictly convex in $M_{s,t}^{\mathcal Z}$, so this stationary point is the unique minimizer. Therefore,
    \[
        \widehat M_{s,t}^{\mathcal Z} = D^\top M_{s,t}C_sE^\top \left( EC_sE^\top \right)^{-1}.
    \]
    \item We next minimize with respect to $Q_{s,t}^{\mathcal Z}$. Using the definition of the Bures--Wasserstein distance in \cref{eq:bures-distance}, the covariance contribution can be written as
    \begin{align} \operatorname{tr}\left(Q_{s,t}^{\mathcal Z}\right) + \operatorname{tr}\left(Q_{s,t}\right) - 2&\operatorname{tr} \left[ \left( (Q_{s,t}^{\mathcal Z})^{1/2} D^\top Q_{s,t}D (Q_{s,t}^{\mathcal Z})^{1/2} \right)^{1/2} \right] \nonumber\\
    &\qquad= d_{\rm B}^2 \left( Q_{s,t}^{\mathcal Z}, D^\top Q_{s,t}D \right) + \operatorname{tr}(Q_{s,t}) - \operatorname{tr} \left( D^\top Q_{s,t}D \right).
    \end{align}
    The remaining terms are independent of $Q_{s,t}^{\mathcal Z}$, so the minimization reduces to the Bures--Wasserstein distance above, which is uniquely minimized at
    \[
        \widehat Q_{s,t}^{\mathcal Z} = D^\top Q_{s,t}D.
    \]
\end{enumerate}
This proves \cref{eq:optimal-latent-gaussian}.
\end{proof}

By \cref{prop:optimal-latent-gaussian-transition}, the latent transition parameters can therefore be eliminated analytically, and the joint optimization reduces to
\[
    \min_{E,D} \mathcal J_{s,t} \left( E,D, \widehat M_{s,t}^{\mathcal Z}(E,D), \widehat Q_{s,t}^{\mathcal Z}(D) \right).
\]
Thus, the latent Gaussian transition is determined explicitly for every candidate encoder and decoder, leaving only the optimization over $E$ and $D$. This is precisely the structure of a variable-projection reduction \citep{golub1973differentiation,chung2023variable}.

\subsection{Closed-form benchmark for a fixed source and target time}
\label{app:transition-wise-optimum}

The preceding reduction applies to a Stochastic Latent Twin that uses the same encoder and decoder across all considered source and target times. Although the optimal latent Gaussian transition can be eliminated separately for each $(s,t)$, the remaining optimization over $E$ and $D$ must represent all transitions simultaneously and generally has no closed-form solution.

To obtain a closed-form lower bound at a prescribed latent dimension, we consider a relaxed reference problem in which the encoder and decoder may be optimized separately for each time pair $(s,t)$. We denote these pair-dependent maps by $E_{s,t}$ and $D_{s,t}$. Since $E_{s,t}$ may depend jointly on the source and target times, the latent mean-propagation matrix can be absorbed into the encoder and may, without loss of generality, be taken as $M_{s,t}^{\mathcal Z}=I$. The resulting reference problem is therefore
\begin{equation} \label{eq:pair-dependent-benchmark}
    \min_{\substack{ E_{s,t},\,D_{s,t},\,Q_{s,t}^{\mathcal Z}\succeq0\\ D_{s,t}^\top D_{s,t}=I}} \left\| \left( M_{s,t}-D_{s,t}E_{s,t} \right)C_s^{1/2} \right\|_{\mathrm F}^2 + d_{\rm B}^2 \left( D_{s,t}Q_{s,t}^{\mathcal Z}D_{s,t}^\top, Q_{s,t} \right).
\end{equation}

\begin{proposition}[Optimal pair-dependent linear-Gaussian approximation] \label{prop:transition-wise-optimum} 
Let $ H_{s,t} := M_{s,t}C_sM_{s,t}^\top + Q_{s,t}$, and $ \eta_1\geq\eta_2\geq\cdots\geq\eta_{n_x}\geq0 $ denote its eigenvalues. Further, let $U_{s,t,n_z}$ contain eigenvectors associated with the $n_z$ largest eigenvalues. Then an optimal solution of \cref{eq:pair-dependent-benchmark} is
\[
    D_{s,t}^{\star} = U_{s,t,n_z}, \qquad E_{s,t}^{\star} = U_{s,t,n_z}^\top M_{s,t}, \qquad Q_{s,t}^{\mathcal Z,\star} = U_{s,t,n_z}^\top Q_{s,t}U_{s,t,n_z}.
\]
and the minimum expected squared $2$-Wasserstein transition error is
\[
    \mathcal E_{s,t}^{\star}(n_z) = \sum_{i=n_z+1}^{n_x}\eta_i.
\]
\end{proposition}

\begin{proof}
For a given $D_{s,t}$, expanding the conditional-mean contribution and using $D_{s,t}^\top D_{s,t}=I$ gives
\[
    \left\| \left( M_{s,t}-D_{s,t}E_{s,t} \right) C_s^{1/2} \right\|_{\mathrm F}^2 = \operatorname{tr}\left(M_{s,t}C_sM_{s,t}^\top\right) -2\operatorname{tr}\left(E_{s,t}C_sM_{s,t}^\top D_{s,t}\right) +\operatorname{tr}\left(E_{s,t}C_sE_{s,t}^\top\right).
\]
The covariance contribution is independent of $E_{s,t}$, so the optimization over the encoder is determined entirely by the conditional-mean term. The quadratic expression is minimized by
\[
    E_{s,t}^{\star}=D_{s,t}^\top M_{s,t}.
\]
By \cref{prop:optimal-latent-gaussian-transition}, for a given $D_{s,t}$ the covariance contribution is minimized by
\[
    Q_{s,t}^{\mathcal Z,\star}=D_{s,t}^\top Q_{s,t}D_{s,t}.
\]
Hence, after eliminating $E_{s,t}$ and $Q_{s,t}^{\mathcal Z}$, the remaining optimization is
\[
    \min_{D_{s,t}^\top D_{s,t}=I} \quad \operatorname{tr}(H_{s,t}) - \operatorname{tr}\left(D_{s,t}^\top H_{s,t}D_{s,t}\right),
\]
where $H_{s,t}$ is defined in \cref{prop:transition-wise-optimum}. Since $\operatorname{tr}(H_{s,t})$ is independent of $D_{s,t}$, this is equivalent to
\[
    \max_{D_{s,t}^\top D_{s,t}=I} \operatorname{tr}\left(D_{s,t}^\top H_{s,t}D_{s,t}\right).
\]
For a symmetric positive semidefinite matrix $H_{s,t}$ with eigenvalues $\eta_1\geq\cdots\geq\eta_{n_x}$, the quantity $\operatorname{tr}(D_{s,t}^\top H_{s,t}D_{s,t})$ over all $D_{s,t}$ with $D_{s,t}^\top D_{s,t}=I$ is maximized by choosing the columns of $D_{s,t}$ as eigenvectors associated with the $n_z$ largest eigenvalues. This result is known as the Ky Fan maximum principle \citep{bhatia1997matrix}. Hence, the optimal decoder, encoder, and latent covariance are
\[
    D_{s,t}^{\star}=U_{s,t,n_z}, \qquad E_{s,t}^{\star}=U_{s,t,n_z}^\top M_{s,t}, \qquad Q_{s,t}^{\mathcal Z,\star} = U_{s,t,n_z}^\top Q_{s,t}U_{s,t,n_z}.
\]

The minimum value follows from the eigenvalues $\eta_1\geq\cdots\geq\eta_{n_x}\geq0$ of $H_{s,t}$. Since
\[
    \operatorname{tr}(H_{s,t}) = \sum_{i=1}^{n_x}\eta_i 
    \qquad \text{and} \qquad
    \max_{D_{s,t}^\top D_{s,t}=I} \operatorname{tr} \left( D_{s,t}^\top H_{s,t}D_{s,t} \right) = \sum_{i=1}^{n_z}\eta_i,
\]
the minimum expected squared $2$-Wasserstein transition error is
\[
    \mathcal E_{s,t}^{\star}(n_z) = \sum_{i=n_z+1}^{n_x}\eta_i.
\]
\end{proof}

Because the encoder and decoder in this reference problem may depend jointly
on $(s,t)$, \cref{eq:linear-transition-optimal-error} represents a best-case
transition error within the pair-dependent linear-Gaussian factorization, not
the error of a Stochastic Latent Twin with a common encoder--decoder
representation. It therefore provides a lower reference value for the latter.

\subsection{Extensions of the linear parameterization}
\label{app:linear-extensions}

The preceding analysis considers a time independent encoder and decoder and allows the latent Gaussian transition to depend directly on $(s,t)$. We briefly note two related structured parameterizations.

\paragraph{Time-dependent encoder and decoder.}
One may allow the representation to depend separately on the source and target times,
\[
    e_s(x)=E_sx, \qquad d_t(z)=D_tz.
\]
Then
\[
    \Psi_{s,t}(x_s) = \mathcal N\!\left( D_tM_{s,t}^{\mathcal Z}E_sx_s, D_tQ_{s,t}^{\mathcal Z}D_t^\top \right).
\]
If $D_t^\top D_t=I$ and $E_sC_sE_s^\top$ is nonsingular, the preceding elimination carries over directly:
\[
    \widehat M_{s,t}^{\mathcal Z} = D_t^\top M_{s,t}C_sE_s^\top \left(E_sC_sE_s^\top\right)^{-1}, \qquad \widehat Q_{s,t}^{\mathcal Z} = D_t^\top Q_{s,t}D_t.
\]
Thus the inner latent transition law remains available in closed form, while the outer representation problem is now posed over the time-dependent matrices $E_s$ and $D_t$. This differs from the pair-dependent benchmark in \cref{app:transition-wise-optimum}, where the representation may depend jointly on $(s,t)$.

\paragraph{Generator-constrained latent mean transitions.}
A different restriction is to impose
\[
    M_{s,t}^{\mathcal Z} = \mathrm e^{(t-s)A_{\mathcal Z}}
\]
for a single latent generator $A_{\mathcal Z}$. This couples the latent mean transitions across prediction intervals, and the pointwise optimal matrices from \cref{eq:optimal-latent-gaussian} need not be representable by one common matrix exponential. Hence the latent mean transition can generally no longer be eliminated independently for each $(s,t)$, and $E$, $D$, and $A_{\mathcal Z}$ must instead be optimized jointly.

\begin{remark}[Nonlinear representations]
The closed-form elimination above extends to nonlinear encoders as long as the decoder remains linear. Writing $Z_s=e(X_s)$, one obtains
\[
    \widehat M_{s,t}^{\mathcal Z} = D^\top M_{s,t}\mathbb E[X_sZ_s^\top] \left(\mathbb E[Z_sZ_s^\top]\right)^{-1}, \qquad \widehat Q_{s,t}^{\mathcal Z} = D^\top Q_{s,t}D,
\]
provided $\mathbb E[Z_sZ_s^\top]$ is nonsingular. For nonlinear decoders, however, the Gaussian latent law is generally mapped to a non-Gaussian physical-space distribution, and the corresponding closed-form elimination is lost.
\end{remark}


\section{Additional numerical details for the linear SDE}
\label{app:linear-numerics}

The linear numerical experiment uses a stable, nonnormal $n_x=6$ dimensional SDE. The nonnormal dynamics avoid a setting in which the dominant transition subspaces are trivially aligned across all prediction horizons. Let $C_0=\mathbb E[X_0X_0^\top]$ denote the initial second-moment matrix. For the linear SDE, the source-state second moments are propagated analytically as
\[
    C_s = M_{0,s}C_0M_{0,s}^\top+Q_{0,s},
\]
and these matrices are used in the Wasserstein objectives below.

Both representations are trained by minimizing the pair-averaged objective of \cref{eq:linear-gaussian-w2-objective}, with the latent transition parameters $M_{s,t}^{\mathcal Z},Q_{s,t}^{\mathcal Z}$ eliminated in closed form at every pair via \cref{prop:optimal-latent-gaussian-transition}, so only $(E,D)$ -- or $(E_s,D_t)$ for the time-dependent representation -- remain as free parameters. The decoder is parameterized as an unconstrained matrix re-orthonormalized by a QR factorization on every forward evaluation, so $D^\top D=I$ holds exactly throughout training rather than only approximately or only at initialization; the encoder is unconstrained. Both are initialized with i.i.d.\ standard Gaussian entries, the encoder additionally scaled by $1/\sqrt{n_x}$. Training pairs are the exhaustive set of forward pairs $s<t$ over a common $100$-point uniform time grid on $[0,t_{\max}]$ (no subsampling). Parameters are optimized with Adam at learning rate $0.03$, cosine-annealed to $2\%$ of that value over $800$ steps, from $3$ independent random restarts; the restart attaining the lowest final training objective is kept. This procedure is repeated independently for each reported latent dimension $n_z$. The relaxed solution of \cref{prop:transition-wise-optimum} requires no numerical optimization and provides the optimal lower bound used for comparison.

For \cref{fig:linear-sde} (right), the reported error at each latent dimension is the arithmetic mean of $\sqrt{\mathcal J_{s,t}}$ over all $4{,}950$ forward pairs on the $100$-point time grid. For the optimal lower bound, $\mathcal J_{s,t}$ takes the minimum value $\mathcal E_{s,t}^{\star}(n_z)$ of \cref{prop:transition-wise-optimum}; for the trainable representations, the same objective is evaluated at the optimized $(E,D)$ or $(E_s,D_t)$. The square root is taken for each time pair before averaging.

Beyond the fixed rank $n_z=n_x=6$ at which the transition law is recovered exactly, it is useful to see how the three references of \cref{fig:linear-sde} (right) compare across the prediction horizon itself rather than only in the single horizon-averaged number reported there. \cref{fig:linear-sde-method-comparison} plots $\sqrt{\mathcal J_{s,t}}$ against $t-s$ at a fixed source time $s=0$, for $n_z=1,\ldots,5$. The fixed representation $(E,D)$ tracks the pair-dependent lower bound and the time-dependent representation closely at short horizons for every rank, and the three remain close throughout at low rank; the gap that opens at longer horizons for the higher ranks is consistent with a single $(E,D)$ having to serve every horizon simultaneously, whereas the pair-dependent bound and $(E_s,D_t)$ are each allowed to specialize to the horizon being evaluated.

\begin{figure}
    \centering
    \includegraphics[width=0.85\linewidth]{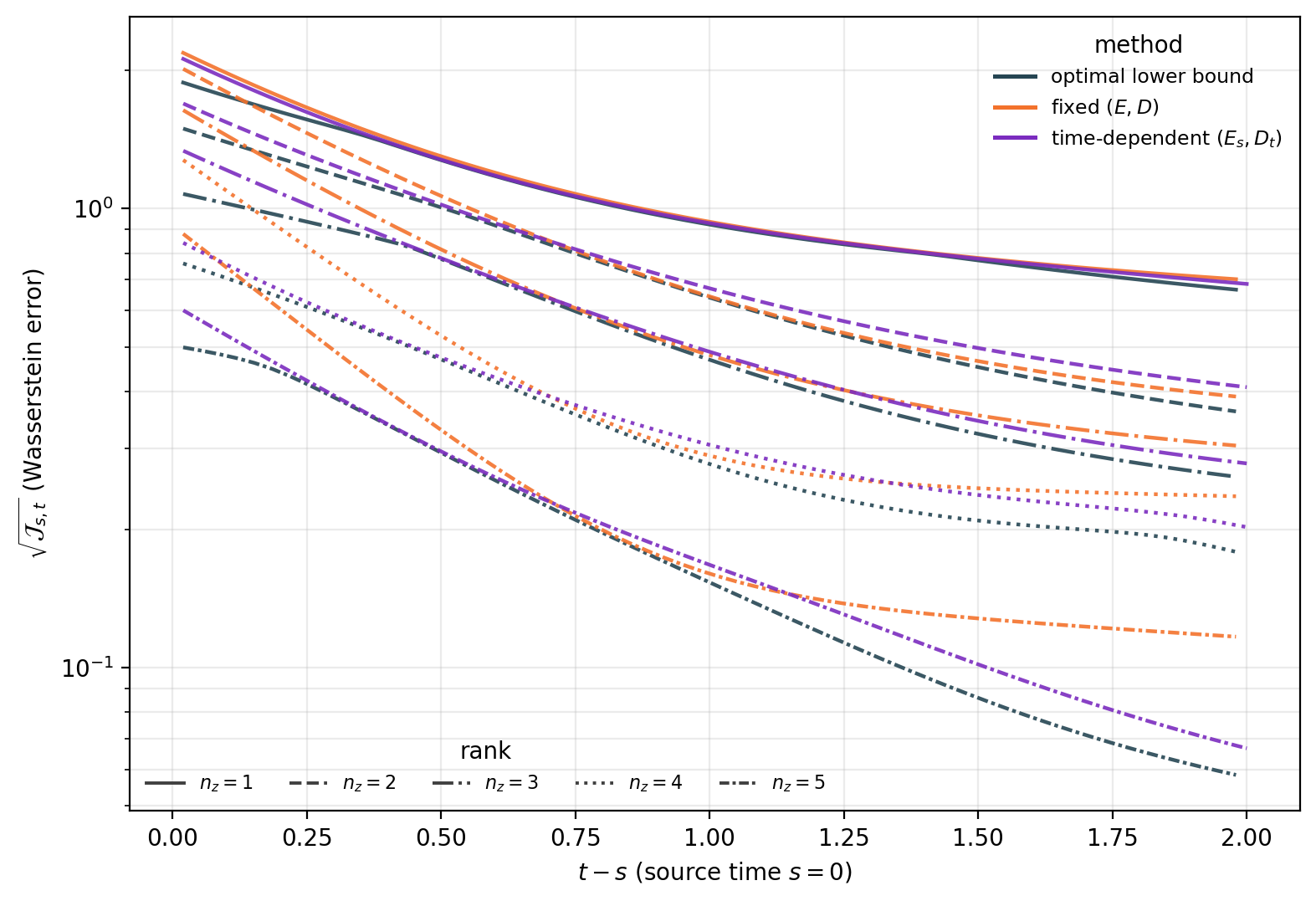}
    \caption{Wasserstein transition error $\sqrt{\mathcal J_{s,t}}$ as a function of prediction horizon $t-s$ at fixed source time $s=0$ for the optimal lower bound, the fixed representation $(E,D)$, and the time-dependent representation $(E_s,D_t)$, for latent dimensions $n_z=1,\ldots,5$.}
    \label{fig:linear-sde-method-comparison}
\end{figure}

We also illustrate trajectory generation using the direct latent recursion described above. \Cref{fig:linear-sde-latent-twin} shows representative trajectories for several latent dimensions. As $n_z$ increases, the predicted mean and uncertainty approach the exact conditional evolution, and at $n_z=6$ the full transition law is recovered.

\begin{figure}
    \centering
    \includegraphics[width=\linewidth]{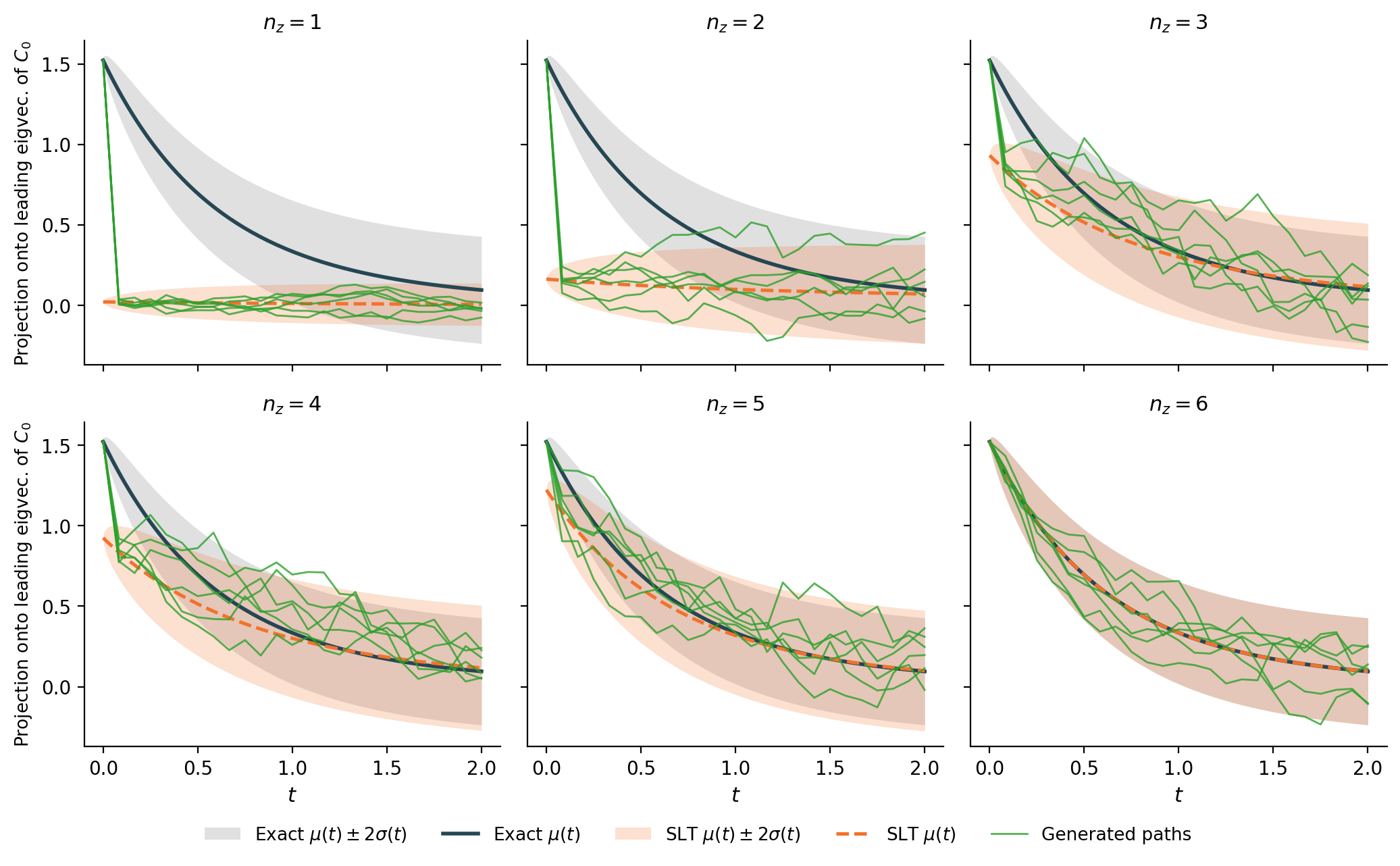}
    \caption{Recursive trajectory generation with the fixed representation $(E,D)$ for latent dimensions $n_z =1,\ldots, 6$. Exact and predicted conditional means and uncertainty regions are shown together with generated trajectories.}
    \label{fig:linear-sde-latent-twin}
\end{figure}
 
\section{Additional details for the CIR experiment}
\label{app:cir-numerics}

\subsection{Transition law, training design, and evaluation}
\label{app:cir-data}

For the CIR process in \cref{eq:cir-sde}, let $h=t-s$. The conditional
transition law in \cref{eq:cir-transition} can be written as
\[
    \Phi_{s,t}(x_s)
    =
    c_{s,t}\,
    \chi^2_{\nu}\!\left(\lambda_{s,t}(x_s)\right),
\]
with
\[
    c_{s,t}
    =
    \frac{\sigma^2\bigl(1-\mathrm e^{-\kappa h}\bigr)}{4\kappa},
    \qquad
    \nu
    =
    \frac{4\kappa\theta}{\sigma^2},
    \qquad
    \lambda_{s,t}(x_s)
    =
    \frac{\mathrm e^{-\kappa h}x_s}{c_{s,t}}.
\]
For $\kappa=1$, $\theta=1$, and $\sigma=1.2$, we have $\nu\approx2.78$. The condition
$2\kappa\theta/\sigma^2\approx1.39>1$ also places the process in the
strictly positive Feller regime.

Source states are sampled from the stationary CIR distribution and
prediction horizons from
\[
    h\sim\mathcal U[0.05,2].
\]
The same source-state and prediction-horizon design points are used for all
learned models.

The amount of conditional-distribution information supplied during training
is matched across the different objectives. For models trained directly with
a Wasserstein objective, the exact conditional law is represented by
128 quantile values at each design point. For likelihood-trained models,
128 independent target states are drawn from the exact conditional
transition at the same design point. Thus differences between the reported
models are not caused simply by providing one model with more realizations
of each conditional distribution.

All models are evaluated on a common collection of source states and
prediction horizons. When the predicted physical transition admits a
convenient quantile representation, the one-dimensional $2$-Wasserstein
distance is evaluated directly. Otherwise it is evaluated from empirical
quantiles of samples from the learned transition. Finite-sample quantiles
introduce a small positive evaluation error for highly accurate models.
We therefore apply the same sample-based estimator to samples drawn from the
exact CIR transition. The resulting nonzero value for the exact transition
should be interpreted as the numerical resolution of the evaluation
procedure rather than approximation error.

\subsection{Optimal Gaussian approximation}
\label{app:optimal-gaussian}

For a fixed source state $x_s$ and time pair $(s,t)$, the best Gaussian approximation is the minimizer in \cref{eq:cir-best-gaussian-problem}, and $\mathcal E_{\rm G}^{\star}(x_s,s,t)$ is its squared error.
Using the one-dimensional quantile representation of the Wasserstein
distance \citep{vallender1974calculation}, we have
\[
    W_2^2\!\left(
        \Phi_{s,t}(x_s),
        \mathcal N(\mu,\sigma^2)
    \right)
    =
    \int_0^1
    \left[
        \Phi_{s,t}^{-1}(u;x_s)
        -
        \mu
        -
        \sigma\Phi_{\rm N}^{-1}(u)
    \right]^2
    \,\mathrm du .
\]
Since
\[
    \int_0^1\Phi_{\rm N}^{-1}(u)\,\mathrm du=0,
    \qquad
    \int_0^1\bigl(\Phi_{\rm N}^{-1}(u)\bigr)^2\,\mathrm du=1,
\]
minimization with respect to $\mu$ gives
\[
    \mu^\star_{s,t}(x_s)
    =
    \int_0^1
    \Phi_{s,t}^{-1}(u;x_s)\,\mathrm du
    =
    \int_{\mathbb R}x\,\mathrm d\Phi_{s,t}(x_s)(x),
\]
while minimization with respect to $\sigma$ gives
\[
    \sigma^\star_{s,t}(x_s)
    =
    \int_0^1
    \Phi_{s,t}^{-1}(u;x_s)
    \Phi_{\rm N}^{-1}(u)
    \,\mathrm du .
\]
Substituting these minimizers gives \cref{eq:cir-optimal-gaussian-error}. This characterization is the one-dimensional $W_2$ projection onto the Gaussian family; see also \citet{wang2026relative}.

\subsection{Model parameterizations and training}
\label{app:cir-models}

The CIR experiment compares several ways of enlarging the class of conditional transition laws.

\paragraph{Gaussian latent transition law with linear representation.} For latent dimension $n_z$ we use
\[
    e(x)=Ex,
    \qquad
    d(z)=Dz,
\]
and
\[
    m(z_s,s,t)
    =
    \mathcal N\!\left(
        \mu_\theta(z_s,t-s),
        \Sigma_\theta(z_s,t-s)
    \right).
\]
A small neural network predicts the latent mean and a Cholesky factor of the latent covariance. Because the decoder is linear, the induced physical transition remains Gaussian. Models with $n_z=1,2,$ and $4$ are trained using the one-dimensional Wasserstein objective described above.

\paragraph{Gaussian-mixture latent transition law.} To enlarge the probability family while retaining a linear representation, we use
\[
    m(z_s,s,t)
    =
    \sum_{j=1}^{K}
    \pi_j(z_s,t-s)
    \mathcal N\!\left(
        \mu_j(z_s,t-s),
        \Sigma_j(z_s,t-s)
    \right),
\]
with $K=4$ in the reported experiment. The mixture parameters are learned from sampled transitions by maximum likelihood. A linear pushforward of this latent distribution remains a Gaussian mixture in physical space.

\paragraph{Gaussian latent transition law with nonlinear representation.} For the nonlinear-representation experiment, both the encoder and decoder are small multilayer perceptrons,
\[
    e:\mathbb R\rightarrow\mathbb R^{n_z},
    \qquad
    d:\mathbb R^{n_z}\rightarrow\mathbb R,
\]
while the latent transition law remains Gaussian. For $n_z=1$, standard-normal quantile nodes are transformed through the predicted latent Gaussian and the decoder and compared directly with the exact CIR quantiles. Because $d$ is nonlinear, the physical pushforward $d_\#m(z_s,s,t)$ is no longer restricted to the Gaussian family.

\paragraph{Learned noncentral-$\chi^2$ transition.} Finally, because the exact CIR transition family is known, we consider the identity representation $z=x$ and parameterize
\[
    m_\theta(x_s,s,t)
    =
    c_\theta(x_s,t-s)\,
    \chi^2_{\nu_\theta(x_s,t-s)}
    \!\left(
        \lambda_\theta(x_s,t-s)
    \right).
\]
A neural network maps
\[
    (x_s,t-s)
    \longmapsto
    \bigl(
        c_\theta,
        \nu_\theta,
        \lambda_\theta
    \bigr),
\]
with positive parameterizations for all three distribution parameters. The network is trained by negative log-likelihood using sampled target states and is not supplied with $\kappa$, $\theta$, $\sigma$, or the
analytical formulas for the CIR transition parameters.

\subsection{Distributional diagnostics}
\label{app:cir-diagnostics}

The mean Wasserstein errors reported in the main text summarize the discrepancy between complete conditional transition laws. To localize where these distributions differ, we additionally compare their conditional cumulative distribution functions. For a fixed source state $x_s$, define
\[
    F_{s,t}(x_t;x_s)
    :=
    \mathbb P(X_t\leq x_t\mid X_s=x_s).
\]
We evaluate the pointwise difference between the predicted and exact conditional CDFs at the representative source state $x_s=0.40$ over target values $x_t$ and prediction horizons $t-s$.

\Cref{fig:cir-cdf-deviation} shows that the linear-decoder Gaussian model exhibits the largest systematic discrepancy, reflecting its inability to capture the skewed CIR transition. The Gaussian mixture reduces this error, while the nonlinear representation and learned noncentral-$\chi^2$ transition provide substantially more accurate conditional distributions.

\begin{figure}
    \centering
    \includegraphics[width=\linewidth]{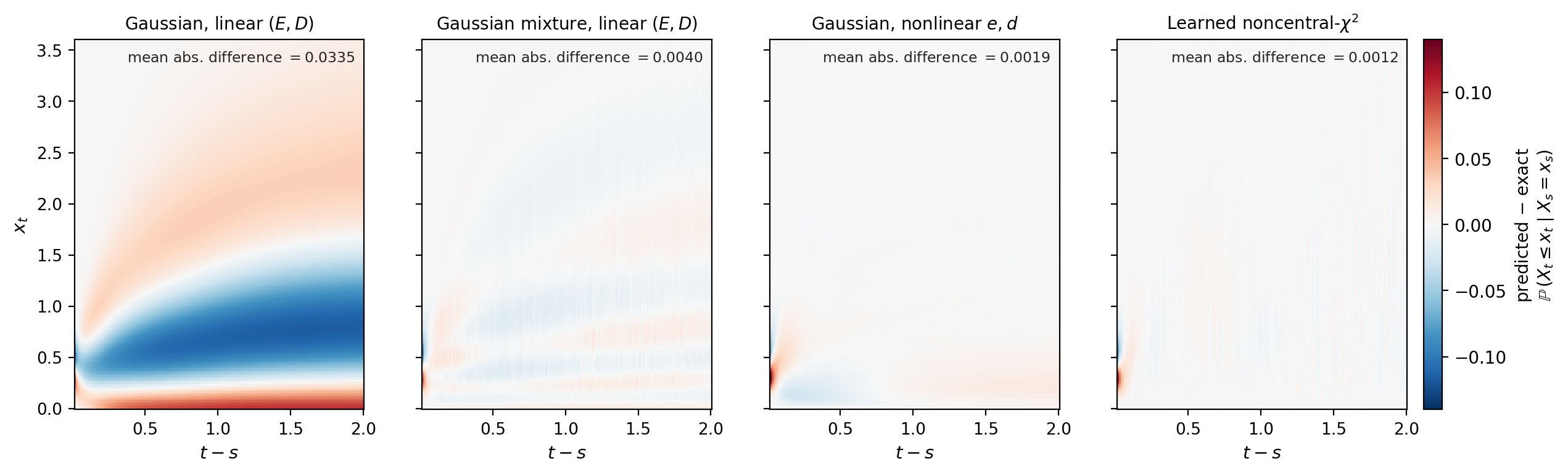}
    \caption{Pointwise conditional-CDF error for the CIR transition at $x_s=0.40$. The panels show $\widehat F_{s,t}(x_t;x_s)-F_{s,t}(x_t;x_s)$ as a function of target value $x_t$ and prediction horizon $t-s$ for the linear-Gaussian, Gaussian-mixture, nonlinear-representation, and learned noncentral-$\chi^2$ models.}
    \label{fig:cir-cdf-deviation}
\end{figure}

\subsection{Recursive generation of CIR trajectories}
\label{app:cir-recursion}

Although the Stochastic Latent Twin is trained as a conditional transition
model, repeated application of the learned physical transition operator can
also be used to generate discrete stochastic trajectories. Starting from
$x_k$ at time $t_k$, we use
\[
    z_k=e(x_k),
    \qquad
    Z_{k+1}\sim m(z_k,t_k,t_{k+1}),
    \qquad
    x_{k+1}=d(Z_{k+1}).
\]
For the direct noncentral-$\chi^2$ model, latent and physical states coincide
and the transition is applied directly to $x_k$.

The resulting ensembles are compared with the exact CIR marginals
$\Phi_{0,t}(x_0)$. The linear-Gaussian model retains the symmetric shape and
possible negative support inherited from its one-step transition family,
whereas the nonlinear-representation and noncentral-$\chi^2$ models more
closely reproduce the exact CIR marginals.

For diagnostic purposes, we additionally consider the pure latent recursion
\[
    Z_{k+1}\sim m(Z_k,t_k,t_{k+1}),
\]
without decoding and re-encoding at each step. This imposes the stronger
requirement that the learned latent evolution map remain accurate on latent
states generated by its own previous predictions. Since this property is not
imposed during training, we regard the experiment as a diagnostic of latent
closure rather than as part of the principal CIR comparison.

\begin{figure}
    \centering
    \includegraphics[width=\linewidth]{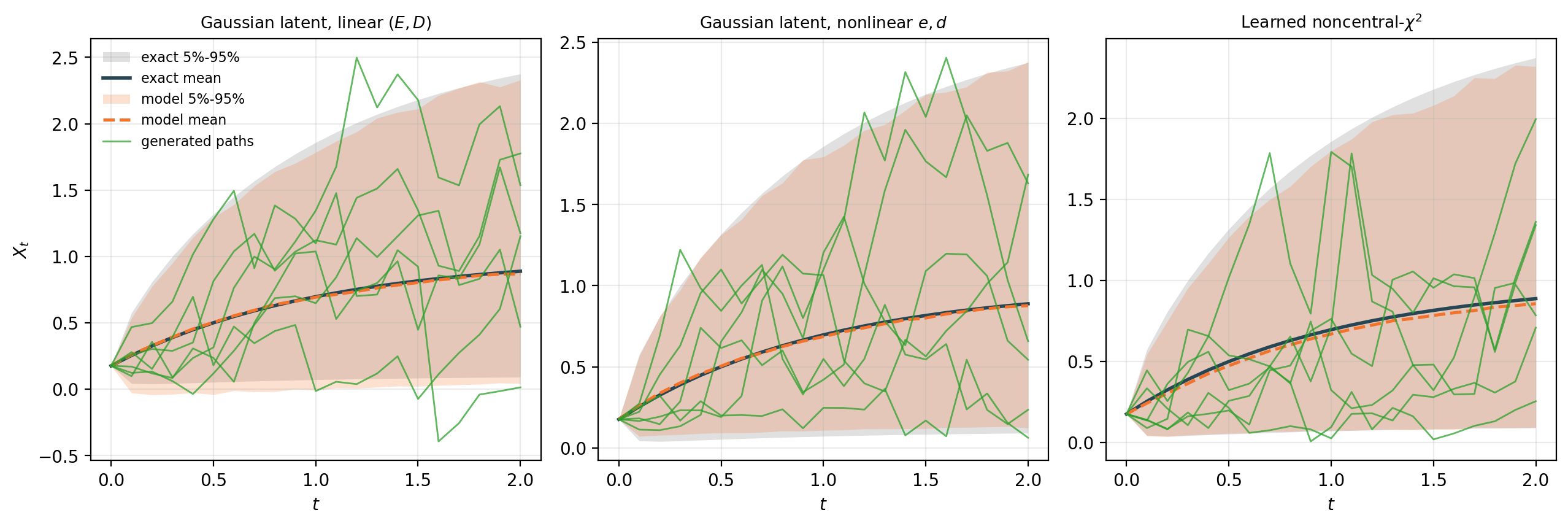}
    \caption{Recursive CIR trajectory generation. Ensembles obtained by
    repeated application of the learned physical transition operator are
    compared with the exact CIR marginals.}
    \label{fig:cir-recursive-paths}
\end{figure}

\section{Additional details for the stochastic Lorenz experiment}
\label{app:lorenz-numerics}

\subsection{Data generation and transition pairs}
\label{app:lorenz-data}

Following the Neural Stochastic Flows reference setup \citep{kiyohara2026neural},\footnote{\url{https://github.com/nkiyohara/jax-nsf/tree/main/examples/lorenz}} trajectories are generated from the stochastic Lorenz-63 system \cref{eq:lorenz-sde} using Euler--Maruyama with integration step size $10^{-3}$ and are recorded at $41$ equally spaced times on $[0,1]$. Initial states are sampled independently from $X_0\sim\mathcal N(0,I)$, with no additional observation noise. As in the reference implementation, we use $1{,}024$ training trajectories and $256$ validation trajectories; we additionally generate $1{,}024$ independent test trajectories.

States supplied to the Stochastic Latent Twin are standardized using training-set statistics only, whereas Neural Stochastic Flows is trained directly in physical coordinates. All direct comparisons are therefore reported in physical units.

Following the Neural Stochastic Flow training setup, every ordered pair of observation times $s_i<t_i$ is used as a transition pair, $(x_{s_i},x_{t_i},s_i,t_i)$,  giving $\binom{41}{2}=820$ source--target pairs per trajectory and $839{,}680$ training pairs. The same observed transition pairs are used for both methods.

\subsection{Architecture, training, and evaluation}
\label{app:lorenz-training}

The encoder and decoder are multilayer perceptrons with two hidden layers of width $128$ and SiLU activations. For the state-dependent Gaussian transition introduced in \cref{subsec:lorenz}, the transition network has two hidden layers of width $96$ with SiLU activations and takes $(z_s,s,t)$ as input. Its outputs parameterize the conditional mean and lower-triangular Cholesky factor, with a softplus applied to the diagonal entries. The resulting model has $50{,}743$ trainable parameters.

A latent-dimension screen over $n_z\in\{2,3,4,6,8,12\}$ showed that the validation energy score changes little from $n_z=4$ onward, while the reconstruction error continues to decrease. We therefore use $n_z=8$.

Training minimizes reconstruction of both states together with the energy score \cref{eq:energy-score-loss}, approximated with $16$ decoded predictive samples per conditioning state. We use Adam with batch size $512$, initial learning rate $10^{-3}$ with cosine annealing, and $100$ epochs.

Prediction accuracy is evaluated on $4{,}000$ held-out transitions drawn across all source and target times from the $1{,}024$ test trajectories. We additionally evaluate the direct $0\to1$ transition to assess the final bimodal state distribution and its $x^{(1)}$ marginal. For each test transition, we draw $64$ independent one-shot predictive samples. The predictive mean of the $64$ samples is computed, and the root-mean-squared error between the predictive mean and the observed target state is taken over all $4{,}000$ test transitions. The energy score is the Monte Carlo estimate of \cref{eq:energy-score-loss} from the same $64$ samples, averaged over the same transitions. The endpoint $W_1$ error uses a single one-shot sample from $0 \rightarrow 1$ for each test trajectory.

\Cref{fig:lorenz-density-evolution} illustrates the learned temporal evolution. Each distribution $\Psi_{0,t}(x_0)$ is evaluated directly at the target time rather than obtained by recursive propagation.

\begin{figure}
    \centering
    \includegraphics[width=0.29\linewidth]{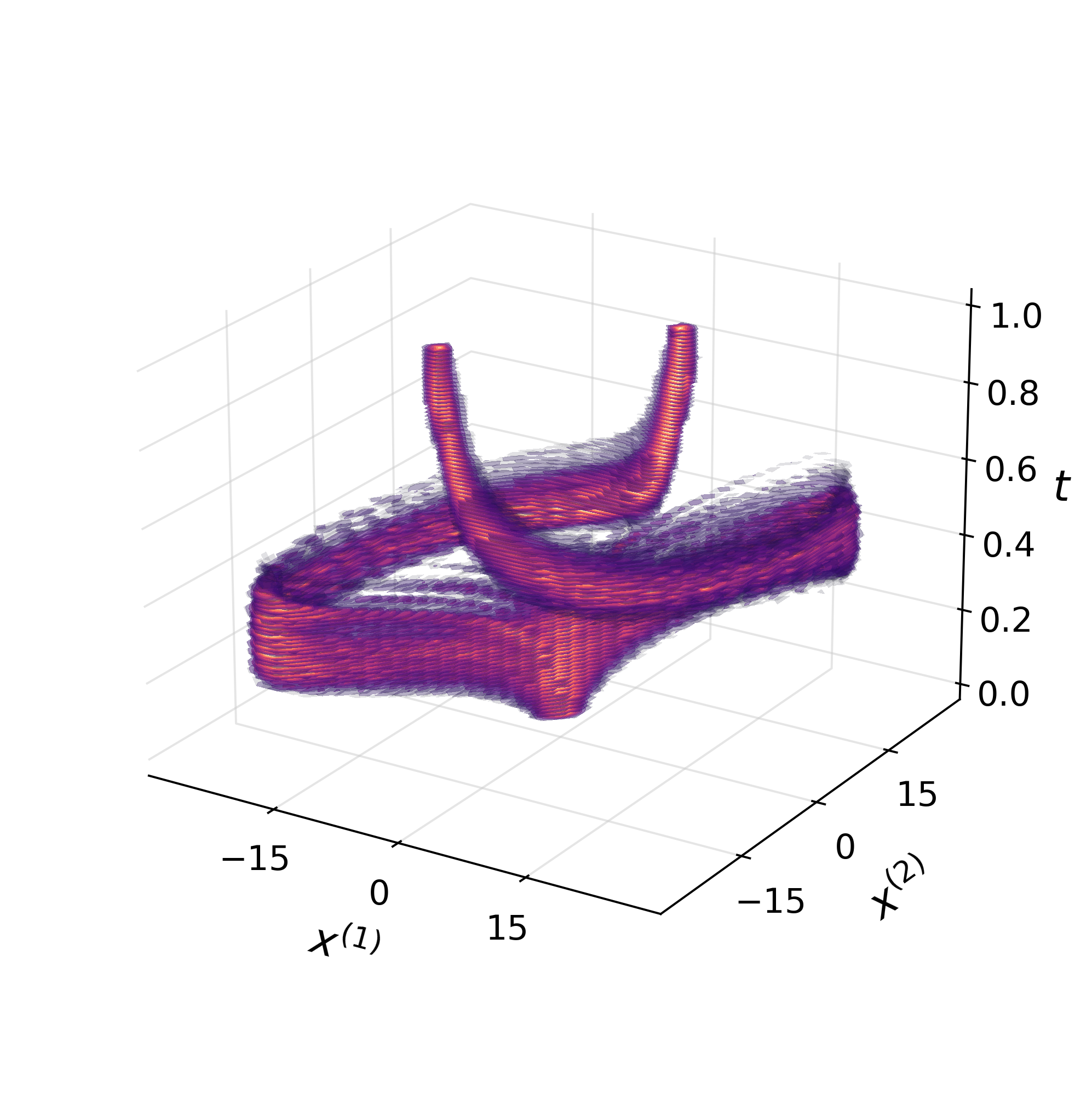}
    \includegraphics[width=0.29\linewidth]{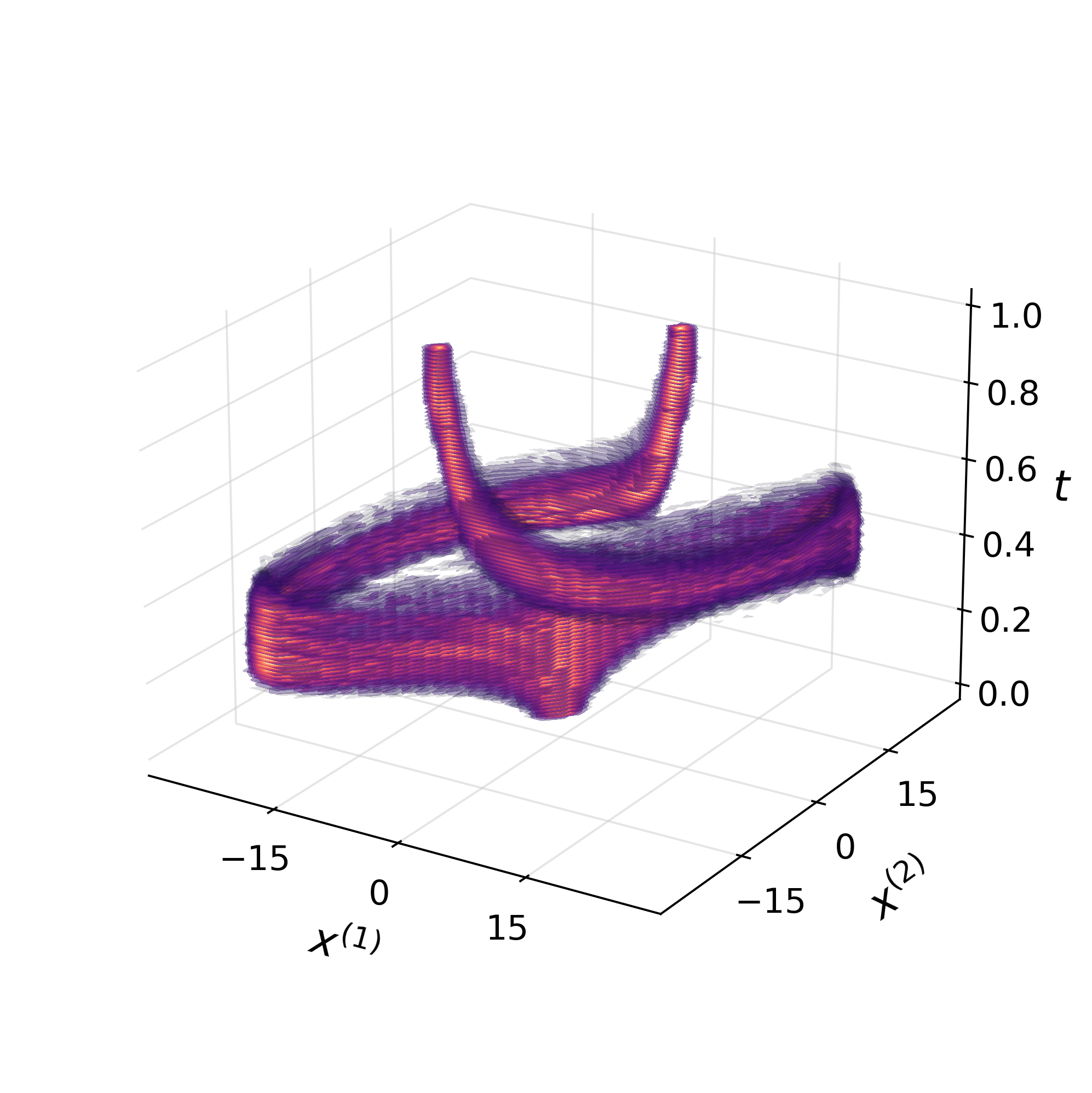}
    \includegraphics[width=0.38\linewidth]{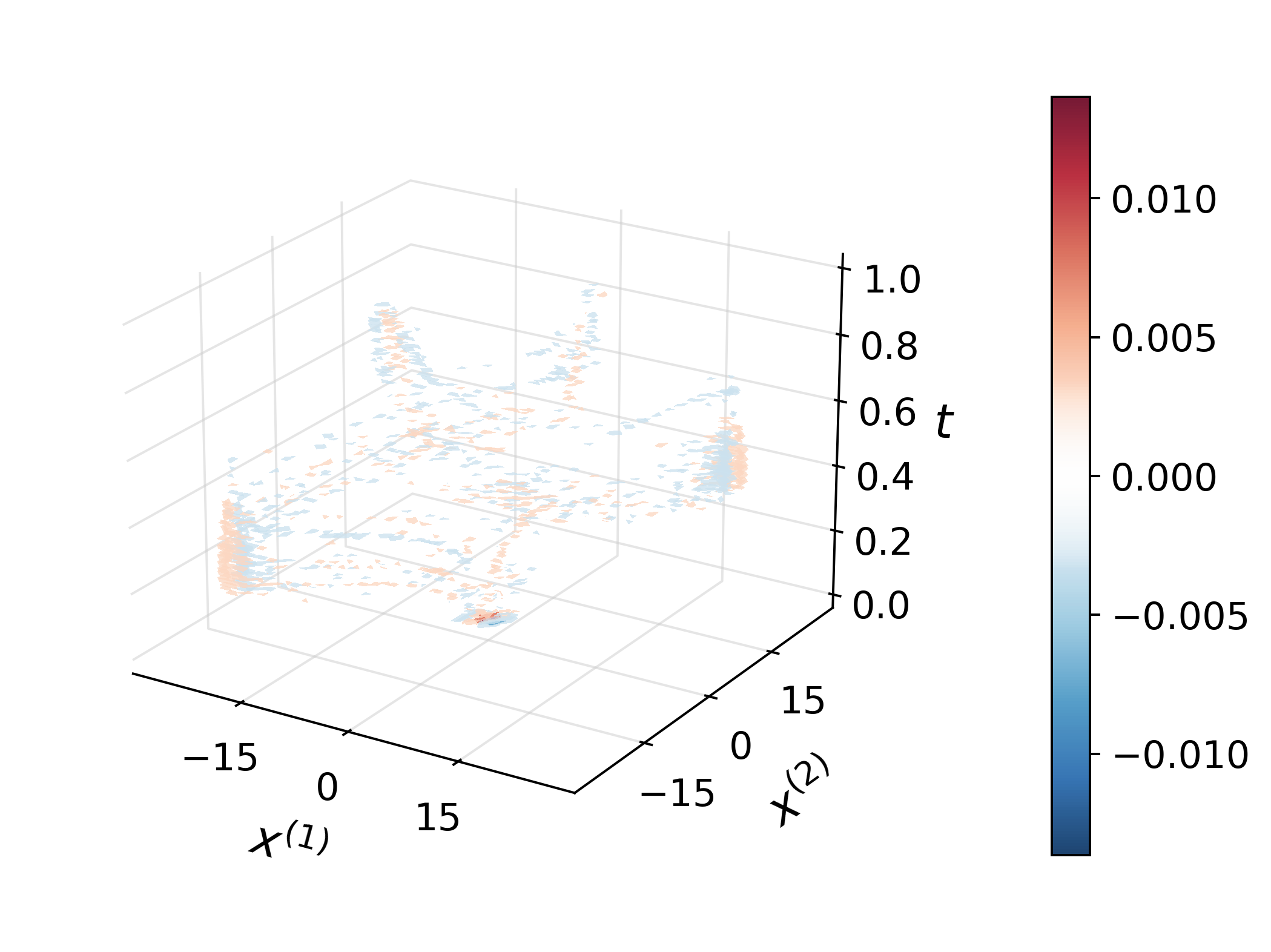}
    \caption{Space–time density of $(x^{(1)},x^{(2)})$ for the $1{,}024$ test trajectories over $t\in[0.025,1]$: reference process (left), Stochastic Latent Twin prediction $\Psi_{0,t}(x_0)$ (middle), and their difference (right). The twin is evaluated directly from the initial states at $101$ times, most of which lie between the observation times used for training. Since the test trajectories are observed only every $0.025$, the reference at intermediate times is obtained by linear interpolation. Densities are smoothed histograms.}
    \label{fig:lorenz-density-evolution}
\end{figure}

\subsection{Comparison with Neural Stochastic Flows}
\label{app:lorenz-competitors}

Neural Stochastic Flows \citep{kiyohara2026neural} is trained on the same stochastic Lorenz data using the authors' reference implementation. In addition to its conditional normalizing flow, Neural Stochastic Flow uses an auxiliary bridge model to construct intermediate-time transitions during training, reflecting the same transition-composition structure underlying the Chapman--Kolmogorov consistency, see \cref{eq:latent-chapman-kolmogorov}. The flow and bridge networks together contain $51{,}540$ trainable parameters.

Although both methods use the same observed source--target pairs, their training procedures differ: the Stochastic Latent Twin learns these conditional transitions directly, whereas Neural Stochastic Flow additionally uses its bridge model and intermediate-time flow objective. The Stochastic Latent Twin is trained on the energy score, which is one of the reported metrics, while the Neural Stochastic Flow is trained on its own objective. We compare the methods at common epoch budgets rather than by equating optimization steps.  

The same ordering observed in the main $100$-epoch comparison is already present after $10$ epochs: the Stochastic Latent Twin attains predictive-mean root-mean-squared error $1.71$, energy score $0.309$, and endpoint $W_1$ error $0.338$, compared with $2.14$, $0.497$, and $0.469$, respectively, for Neural Stochastic Flows.  The Neural Stochastic Flow model is evaluated with the same protocol as the Stochastic Latent Twin, using the final model.

\section{Additional details for the Swiss-roll experiment}
\label{app:swissroll-numerics}

\subsection{Data and backward transition pairs}
\label{app:swissroll-data}

The target data are generated using
\texttt{sklearn.datasets.make\_swiss\_roll} with noise level $0.30$ and
random seed $7$. We retain coordinates $0$ and $2$, yielding the
two-dimensional spiral projection used in the experiment. The data are
standardized using statistics computed from the training set only. We use
$4{,}096$ training samples and $512$ validation samples. 

The observation grid contains $65$ equally spaced time points
$t_0,\ldots,t_{64}$ on $[0,1]$. Training proceeds backward in time: for
grid indices $i<j$, the later and more strongly corrupted state $X_{t_j}$
is used as the conditioning state, and the model learns the conditional
distribution of the earlier state $X_{t_i}$. To sample all prediction horizons uniformly, the index difference $j-i$ is drawn uniformly from $\{1,\ldots,64\}$, after which a compatible earlier index $i$ is selected.

For $a<b$, the forward OU process in \cref{eq:swiss-roll-forward} has the
exact transition
\[
    X_b \mid X_a=x_a
    \sim
    \mathcal N\!\left(
        e^{-\lambda(b-a)}x_a,\,
        \left(1-e^{-2\lambda(b-a)}\right)I
    \right).
\]
Thus, a training pair is generated by first sampling the transition from
$0$ to $t_i$ and then from $t_i$ to $t_j$, so that the conditioning and
target states belong to the same stochastic realization without requiring
a full trajectory to be stored.

\subsection{Architecture, training, and evaluation}
\label{app:swissroll-training}

The model uses a Gaussian latent transition law with latent dimension $n_z=6$,
\[
m(z_s,s,t)
=
\mathcal N\!\left(
    M_{s,t}^{\mathcal Z} z_s,\,
    L_{s,t}^{\mathcal Z}
    \bigl(L_{s,t}^{\mathcal Z}\bigr)^\top
\right).
\]
The encoder and decoder are multilayer perceptrons with two hidden layers of
width $48$, SiLU activations, and LayerNorm after each hidden activation. The transition network has two hidden layers of width $64$ with SiLU activations and takes $(s,t)$ as input.  It produces $M_{s,t}^{\mathcal Z}$ and the lower-triangular entries of $L_{s,t}^{\mathcal Z}$, with positive diagonal entries enforced through
a softplus parameterization. The model has $14{,}017$ trainable parameters.

Training minimizes
\[
    \mathcal L_{\rm ae}
    +
    \mathcal L_{\rm es}
    +
    \lambda_{\rm marg}\mathcal L_{\rm marg},
\]
where $\mathcal L_{\rm ae}$ is the reconstruction loss of \cref{eq:ae-training-loss} summed over the conditioning and target state of each pair, and $\mathcal L_{\rm es}$ is the physical-space energy score of \cref{eq:energy-score-loss}. The energy score is approximated using $16$ decoded predictive samples per conditioning state. The auxiliary term $\mathcal L_{\rm marg}$ is a sliced-Wasserstein discrepancy between a batch of one-shot samples generated from $\mathcal N(0,I)$ and an independent batch of training samples at $t=0$, both of the training batch size. The reported model uses $\lambda_{\rm marg}=1.5$.

Training uses Adam with batch size $128$, initial learning rate $10^{-3}$ with cosine annealing, and $20{,}000$ optimization steps.

We use the sliced-Wasserstein distance: all reported sliced-Wasserstein distance values use $128$ projection directions. For the reported generation results, we average over $10$ independent repetitions. Each iteration samples $800$ fresh standard-normal source samples and $800$ target samples shared by all methods together with the projection directions. We report the mean and standard deviation over the repetitions. The finite-sample reference compares two
independent target samples, while the Gaussian-source baseline compares an
untransformed sample from $\mathcal N(0,I)$ with the target sample.

The auxiliary sliced-Wasserstein term emphasizes the particular
$t=1$ to $t=0$ distributional query used for one-shot generation.
Preliminary experiments without this term produced noticeably poorer
one-shot samples (data not shown).

\subsection{MeanFlow and SoFlow comparison} \label{app:swissroll-competitors}

For comparison with other one-step generative approaches, we use the authors' official implementations of MeanFlow \citep{geng2026mean} and SoFlow \citep{luo2026soflow}, whose training objectives and samplers are applied unchanged, and adapt them to the present two-dimensional setting. Both methods are trained from scratch on the same standardized Swiss-roll data used for the Stochastic Latent Twin and operate directly in physical space. To obtain closely matched model sizes, both competitors use multilayer perceptrons in place of the networks used in the original implementations.

MeanFlow \citep{geng2026mean} learns an average velocity associated with a straight-line interpolation between data and Gaussian noise. We follow the training objective of the reference implementation, which imposes the corresponding self-consistency relation through a Jacobian--vector product. For one-shot generation, the learned average velocity is evaluated directly between $t=1$ and $t=0$.  We use the settings of its ImageNet experiments and condition the network on $(t,r)$ as opposed to $(t, t-r)$ used in the reference code as this yielded better performance.

SoFlow \citep{luo2026soflow} learns a deterministic solution map between source and target times. We follow the reference implementation, which combines a flow-matching velocity objective with a solution-consistency objective based on a finite-difference approximation. We use the method settings of its ImageNet configuration without classifier-free guidance.   One-shot generation is obtained by directly evaluating the learned solution map between $t=1$ and $t=0$.

For MeanFlow and SoFlow, we use multilayer perceptrons with three hidden layers of width $82$ and SiLU activations, giving $14{,}188$ trainable parameters, closely matching the $14{,}017$ parameters of the Stochastic Latent Twin. Both models are trained with Adam, batch size $128$, initial learning rate $10^{-3}$ with cosine annealing, and $20{,}000$ optimization steps without weight averaging. The learning rate was selected using short preliminary hyperparameter screens; the remaining settings follow the reference configurations described above.  A summary of training configurations is given in \Cref{tab:swiss-roll-competitors-detailed}.

All three methods are evaluated using the same one-shot generation protocol.  The resulting parameter counts, wall-clock training times, and one-shot generation errors are reported in \cref{tab:swiss-roll-comparison}.

\begin{table}
\centering
\caption{Implementation details for the one-shot Swiss-roll comparison. All methods are evaluated on the same target distribution and generate from $\mathcal N(0,I)$ in one model evaluation. Wall-clock times refer to the present research implementations and are reported for reference.}
\label{tab:swiss-roll-competitors-detailed}
\small
\begin{tabularx}{\linewidth}{@{}lXXX@{}}
\toprule
 & Stochastic Latent Twin & MeanFlow & SoFlow \\
\midrule

Learned object
&
Latent evolution map
&
Average velocity
&
Deterministic solution map
\\

Training objective
&
Reconstruction + energy score + endpoint SWD
&
MeanFlow consistency objective
&
Velocity + solution consistency
\\

Network width
&
Encoder/decoder $48$; transition $64$
&
$82$
&
$82$
\\

Parameters
&
$14{,}017$
&
$14{,}188$
&
$14{,}188$
\\

Batch size
&
$128$
&
$128$
&
$128$
\\

Training steps
&
$20{,}000$
&
$20{,}000$
&
$20{,}000$
\\

Initial learning rate
&
$10^{-3}$
&
$10^{-3}$
&
$10^{-3}$
\\

Training time
&
$243.7$ s
&
$60.8$ s
&
$47.0$ s
\\

Global SWD
&
$\mathbf{0.086\pm0.017}$
&
$0.156\pm0.025$
&
$0.148\pm0.023$
\\

\bottomrule
\end{tabularx}
\end{table}

\section{Additional details for the stochastic pendulum experiment} \label{app:pendulum-numerics}

\subsection{Data generation and training pairs} \label{app:pendulum-data}

Initial states are rejection-sampled uniformly from
$q\in[-2.5,2.5]$ and $p\in[-1.8,1.8]$, subject to $ H(q,p)\in[0.15,1.80]$. This range remains below $H=2$, the energy at which the pendulum changes from oscillatory to rotational motion. We generate $2{,}048$ training,
$512$ validation, and $512$ test trajectories, recorded at $601$ equally spaced times on $[0,30]$ with spacing $\Delta t=0.05$.

To generate the stochastic dynamics of \cref{eq:stochastic-pendulum}, we exploit their random-time representation. At each nominal increment $\Delta t$, we draw $\tau=\Delta t+\sigma\sqrt{\Delta t}\,\xi$, with  $\xi\sim\mathcal N(0,1)$, and propagate the deterministic pendulum over the random elapsed time $\tau$. Thus, no stochastic time integrator is required for the reference trajectories; the stochasticity enters entirely through the sampled time increment. The deterministic flow is approximated using fourth-order Runge--Kutta with eight substeps. The resulting numerical energy drift has mean $5.79\times10^{-8}$ and maximum $8.92\times10^{-7}$, well below the energy defects reported for the learned models.

Training pairs $(x_s,x_t,s,t)$ are sampled from the same stochastic realization. To provide sufficient coverage of both local and longer transitions, half of the pairs use time differences up to $0.5$, while the remaining half are sampled over the full training interval $(0,30]$.

\subsection{Architecture and training}
\label{app:pendulum-training}

The encoder and decoder are multilayer perceptrons with two hidden layers of width $64$ and SiLU activations, and the latent dimension is $n_z=4$. The transition network has two hidden layers of width $96$ and takes $(z_s,s,t)$ as input.

The transition network parameterizes the Gaussian latent transition law. To enforce the identity property at $t=s$, the predicted mean increment is scaled by $t-s$ and the Cholesky factor of the covariance by $\sqrt{t-s}$. Hence,
$m(z,s,s)=\delta_z$ holds exactly by construction.

For the Stochastic Latent Twin and the Stochastic Latent Twin with soft conformity, all four latent coordinates are evolved by the Gaussian transition. For the Stochastic Latent Twin with conformity by construction, the first latent coordinate is copied unchanged and only the remaining three coordinates evolve stochastically. For the two physics-conforming models, the latent physical quantity is
\[
    \mathcal C_{\mathcal Z}(z)=c_{\rm pull}z_1,
\]
where $c_{\rm pull}$ is learned through the pullback relation.

All models are trained with reconstruction of both states in each transition pair together with the physical-space energy score \cref{eq:energy-score-loss}, approximated using $16$ predictive samples per conditioning state. The Stochastic Latent Twin with soft conformity additionally uses both $\mathcal L_{\rm pull}$ and $\mathcal L_{\rm con}$, each with unit weight, whereas the Stochastic Latent Twin with conformity by construction uses only $\mathcal L_{\rm pull}$, also with unit weight, since its latent invariant is already preserved exactly. Training uses Adam with batch size $256$, initial learning rate $10^{-3}$ with cosine annealing, and $40{,}000$ optimizer updates.

\subsection{Evaluation and conformity diagnostics}
\label{app:pendulum-diagnostics}

The conditional-law error is assessed using eight source states from the test set, selected across the prescribed energy range. For each source state and time interval $t-s$, we generate $256$ samples from both the reference conditional distribution and the Stochastic Latent Twin approximation. Reference samples are obtained from the random-time representation by propagating the deterministic pendulum over independently sampled times
$(t-s)+\sigma\sqrt{t-s}\,\xi$, with  $\xi\sim\mathcal N(0,1)$.
The sliced-Wasserstein distance is approximated using $128$ projection directions and averaged over the eight source states. Physical conformity is measured by the root-mean-square energy defect
\[
    \left( \mathbb E\!\left[ \bigl(H(X_t)-H(x_s)\bigr)^2 \right] \right)^{1/2},
\]
approximated from the same predicted samples.

\Cref{fig:pendulum-conditional-clouds} gives a representative comparison at $t-s=1.5$ for three source states with different energies. The unconstrained Stochastic Latent Twin produces visible deviations from the corresponding energy manifolds, whereas both physics-conforming variants remain much more closely concentrated around them.

\begin{figure}
    \centering
    \includegraphics[width=\linewidth]
    {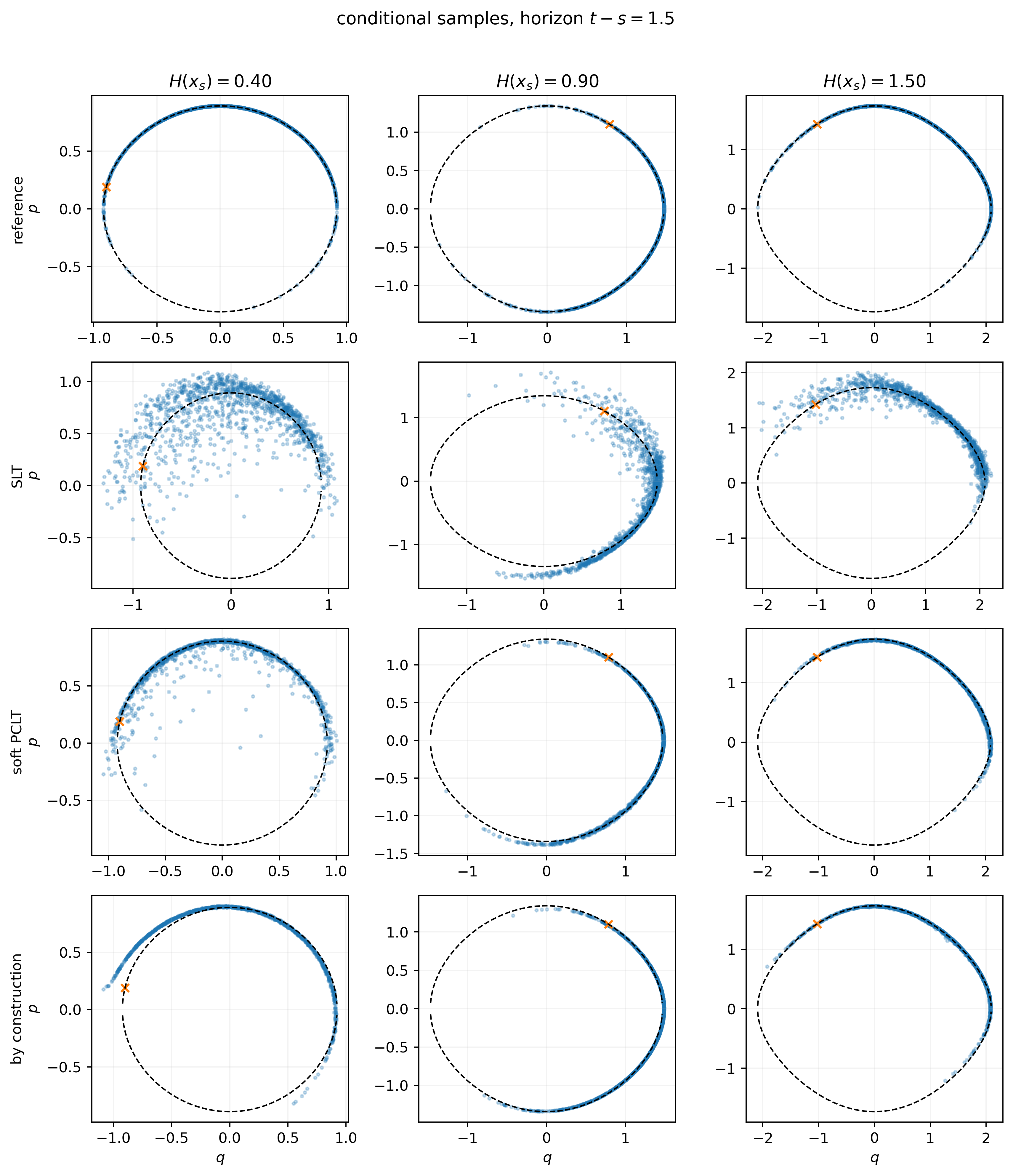}
    \caption{Reference and predicted conditional samples at $t-s=1.5$ for source states with
    $H(x_s)=0.40$, $0.90$, and $1.50$. Dashed curves denote the invariant energy manifolds and crosses the corresponding source states. The physics-conforming variants substantially reduce deviations from the prescribed energy manifold while retaining the stochastic spread of the conditional law.}
    \label{fig:pendulum-conditional-clouds}
\end{figure}

To examine the latent-to-physical conformity transfer, \cref{fig:pendulum-conformity-decomposition} separates the physical energy defect into the contributions associated with target pullback compatibility, latent conformity, and source calibration. 

\begin{figure}
    \centering
    \includegraphics[width=0.85\linewidth]
    {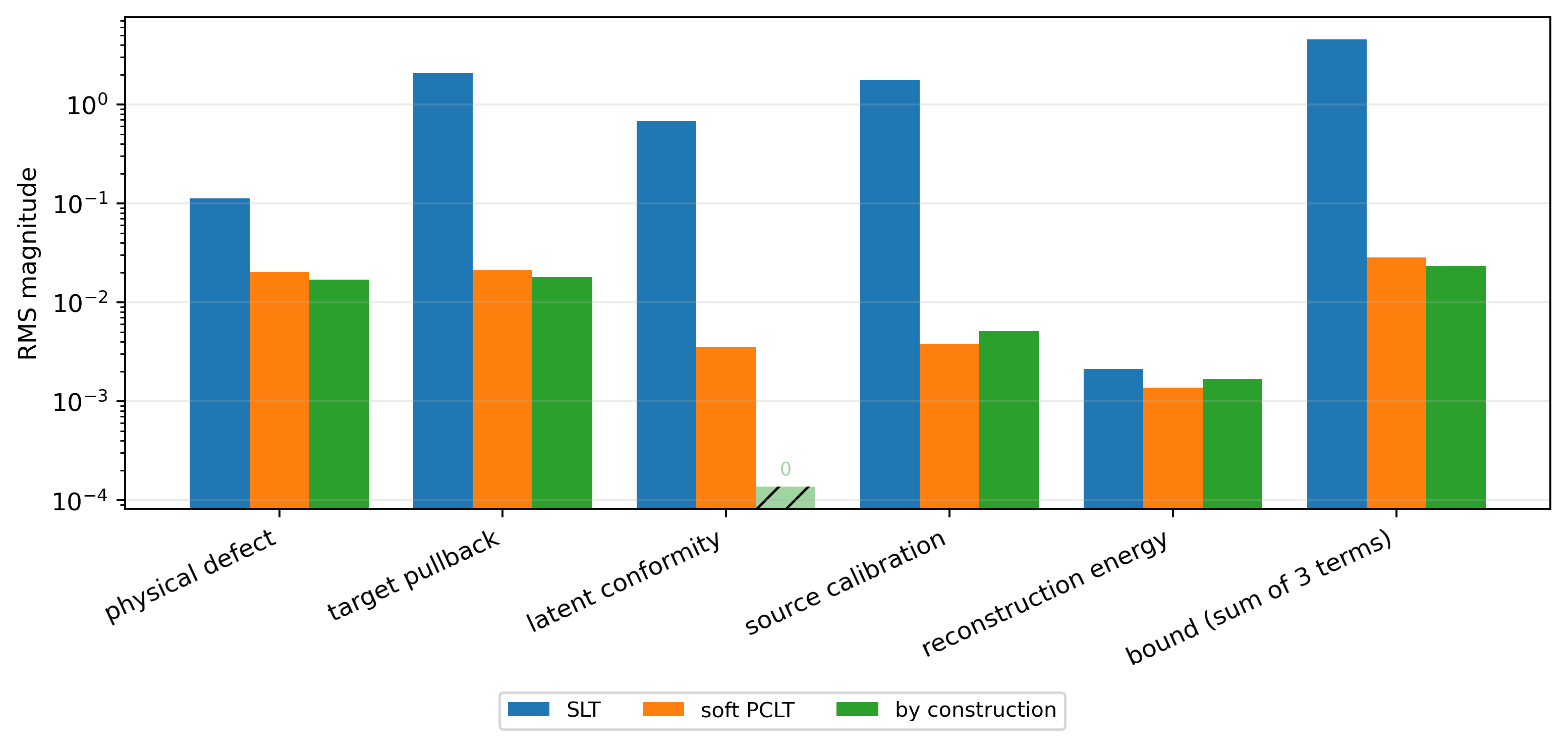}
    \caption{Sample-based decomposition of the physical energy defect at
    $t-s=1.5$ into target pullback, latent-conformity, and source-calibration
    contributions, together with their sum. For the Stochastic Latent Twin with conformity
    by construction, the latent-conformity contribution is zero.}
    \label{fig:pendulum-conformity-decomposition}
\end{figure}

\section{Additional details for the 2D stochastic Burgers Equation experiment}\label{app:burger-numerics}

\subsection{Data generation and training pairs}\label{app:burgers-numerics}

To simulate \Cref{eq:2dBurg}, the domain $[0, 2\pi)^2$ was first discretized into a $64 \times 64$ grid of cells. The advective and stochastic terms were computed as fluxes across the edges of the cells: the former was approximated using the local Lax--Friedrichs flux \citep{quarteroni2014scientific}, and the latter was evaluated using $f(\xi_1,\xi_2)$, multiplied by $\sigma$ and Brownian increments sampled independently at each time step and for each trajectory. The divergence of these fluxes was then approximated by a finite difference of their values at the cell edges. The diffusion term was approximated with the standard five-point finite-difference Laplacian.

Time integration used the Euler--Maruyama method with step size $\Delta t \le 2.5\times10^{-3}$: training uses $32$ equally spaced times on $[0, 1]$, including both endpoints, so time steps were adjusted to fall on each of these recorded times. In this scheme, $\sum_{i = 1}^{n_g^2} u^{(i)}$ is conserved to round-off at every time step, as the finite differences of the fluxes and the five-point Laplacian telescope to zero when summed over the periodic grid. Initial conditions were generated as smooth fields obtained by filtering white noise to its lowest wavenumbers, scaling by a random amplitude in $[1,2]$, and adding a random constant in $[-0.4, 0.4]$ so that the mass varied across trajectories.

With $\nu = 0.01$ and $\sigma = 0.15$, we generate $10{,}000$ initial conditions and simulate a trajectory from each over $t \in [0,1]$, recording the solution at $32$ equally spaced times for training. Each trajectory is a single realization of the noise, and the training pairs consist of $(u_s, u_t, s, t)$ samples drawn from these simulated trajectories with $s \leq t$ two of the 32 times. 

Both the training data and the reference ensembles are generated by the same $64\times64$ solver of \Cref{eq:2dBurg}. At this resolution, the simulated solutions can deviate noticeably from those of the continuous equation, particularly near sharp fronts. The models are therefore trained and evaluated on the conditional law of this discretized stochastic system, rather than on the conditional law of the continuous equation.

\paragraph{Mass conservation.} The pathwise mass conservation used in \Cref{subsec:burgers} follows directly from the conservative form of \Cref{eq:2dBurg}. For sufficiently regular solutions, 
\[
    u\,\partial_{\xi_1}u+u\,\partial_{\xi_2}u = \nabla\cdot \left( \tfrac{u^2}{2}, \tfrac{u^2}{2} \right), \qquad \Delta u=\nabla\cdot(\nabla u),
\]
and integration over the periodic domain gives
\[
    |\Omega| \,\cdot \, \d \mathcal C(u_t)  = -\int_\Omega \nabla\cdot \left( \tfrac{u^2}{2}, \tfrac{u^2}{2} \right) \d\xi_1\d \xi_2\,\d t\, +\nu\int_\Omega \Delta u\,\d \xi_1\d \xi_2\,\d t\,  +\sigma\int_\Omega \nabla\cdot\bigl(f(\xi_1,\xi_2)\,\d W_t\bigr) \d \xi_1\d \xi_2 =0,
\]
because the spatial integral of each periodic divergence vanishes. Therefore, $\mathcal C(u_t)=\mathcal C(u_s)$, for all $s,t > 0$ along every realization.

\subsection{Architecture and training}\label{app:burgers-training}

The architectures used for the Stochastic Latent Twin are identical between the standard and physics-conforming variant.  A latent space of $n_z = 512$ is used, with  CNN-based encoders and decoders.  To reflect the periodic setting, all convolutions are $3\times 3$ with circular padding.  The full network is built from residual blocks consisting of two group-normalization (GN) with 8 groups, SiLU, and convolution layers in sequence:
\[
   R:\ u \;\to\; \mathrm{GN} \;\to\; \mathrm{SiLU} \;\to\; \mathrm{Conv} \;\to\; \mathrm{GN} \;\to\;
  \mathrm{SiLU} \;\to\; \mathrm{Conv},
  \qquad u \mapsto u +  R(u).
\]
The encoder maps the standardized field to $z \in \mathbb{R}^{n_z}$ with an initial convolution followed by four of the above residual stages with stride-2 downsampling, a final residual block, and a final linear layer. The decoder mirrors it with nearest-neighbor upsampling. 

The latent transition law was implemented as a Gaussian with outputs given by an MLP.  Similar to the design of \citep{kiyohara2026neural}, we take
\[
    Z_t \mid Z_s = z_s \sim \mathcal{N}\big( \mu_{s,t},\, L_{s,t}L_{s,t}^\top \big)
\]
with
\begin{align*}
    [\tilde \mu, \tilde L] = {\rm MLP}(z_s, t-s),\qquad \mu_{s,t} = z_s + (t-s)\,\tilde\mu, \qquad L_{s,t} = \sqrt{t-s}\, \tilde L,
\end{align*}
where $\tilde L$ is lower triangular with diagonal entries passed through $\mathrm{softplus}(\cdot) + 10^{-6}$ to ensure positivity. The input to the MLP is the concatenation of the source code $z_s$ with a two-dimensional time encoding, either the pair $(s, t)$ or $(0, t-s)$, the latter used in our experiment. This input is passed through a single hidden layer of width 112 with SiLU activation, followed by a linear layer. The first $n_z$ outputs form $\tilde\mu$, and the remaining $n_z(n_z+1)/2$ fill the lower-triangular entries of $\tilde L$.

Both Stochastic Latent Twin variants are trained on the same $10{,}000$ training trajectories.  Each training step draws a batch of 64 pairs, $(u_s, u_t, s, t)$, $s < t$ from randomly chosen trajectories.  The time gap, $t - s$ is chosen uniformly in $1$ to $31$ with the start time chosen uniformly among possible start times.  The loss, averaged over the batch, consists of the reconstruction loss of both drawn states and the energy score with 8 samples drawn from the models predicted distribution.  Here, we use the norm $\|\cdot\| = \frac{1}{n_g}\|\cdot\|_2$.  We train for 10{,}000 steps with Adam (learning rate $10^{-3}$, cosine-annealed to $10^{-5}$). Both variants see the same batches and the same latent noise.  The two variants differ only in the mean shifts applied after encoding, after latent sampling and after decoding. 

Our Latent Neural Stochastic Flow implementation uses the authors' reference code \citep{kiyohara2025nsfcode}, with their Stochastic Moving MNIST configuration as a template: the model class, stochastic flow, bridge network, loss and training step are unchanged. We made the following adjustments for our setting. The posterior network was changed from the original GRU to the MLP multivariate normal posterior provided in the repository and a latent dimension of $n_z = 512$ was used. The encoder and decoder were replaced by those of the Stochastic Latent Twin so that the two methods share the same architectures and no scene encoder was used.  The decoder output is the mean of a Gaussian whose standard deviation is fixed at $1$, rather than learned as in the reference configuration. The hidden widths of the flow and bridge networks were increased from $64$ to $1024$.  These changes support comparison between learned latent dynamics as all implemented methods share similar parameter counts at inference--see \Cref{tab:Burg-model-comp}.

\begin{table}
\centering
\caption{Parameter count comparison of Stochastic Latent Twin, Stochastic Latent Twin with conformity and Latent Neural Stochastic Flow model components.}
\label{tab:Burg-model-comp}
\begin{tabular}{lccc}
\toprule
Model component & SLT & SLT with conformity & Latent NSF \\
\midrule
Encoder                   & $4{,}064{,}800$  & $4{,}064{,}800$  & $4{,}064{,}800$  \\
Decoder                   & $4{,}067{,}649$  & $4{,}067{,}649$  & $4{,}067{,}649$  \\
Latent transition network & $14{,}955{,}600$ & $14{,}955{,}600$ & $15{,}226{,}880$ \\
Posterior network         & --               & --               & $1{,}050{,}624$  \\
Bridge network            & --               & --               & $17{,}853{,}440$ \\
\midrule
Total                     & $23{,}088{,}049$ & $23{,}088{,}049$ & $42{,}263{,}393$ \\
Total without bridge      & $23{,}088{,}049$ & $23{,}088{,}049$ & $24{,}409{,}953$ \\
\bottomrule
\end{tabular}
\end{table}

The training configuration was modified as follows. We kept the reference value $\lambda = 0.1$ and set $\beta = \beta_{\rm skip} = 0.2$ based on a small hyperparameter sweep. We trained on the same $10{,}000$ training trajectories as the Stochastic Latent Twin: each training batch contained 32 trajectories, and the loss combined reconstruction at every frame, transition losses for all 31 one-step transitions, and 25 skip-ahead pairs per trajectory with larger time gaps. The number of skip-ahead pairs follows the reference configuration, but we sampled the gap of each pair uniformly, matching the Stochastic Latent Twin's sampling of training pairs. The weights $\beta$, $\beta_{\rm skip}$, and $\lambda$ were ramped linearly from zero over the first six epochs, per optimization step. The model was trained for 64 epochs on the $10{,}000$ training trajectories, producing approximately twice the number of optimization steps used for the Stochastic Latent Twin. AdamW was used with a constant learning rate of $10^{-3}$ as in the reference configuration. 

These adjustments were chosen to provide a controlled and representative comparison between the two methods in our setting. The Latent Neural Stochastic Flow was not subjected to an exhaustive architecture or hyperparameter search, so further tuning may improve its performance. The comparison across evaluation times is shown in \Cref{fig:burgers2d}, while \Cref{tab:Burg-metrics} reports the corresponding metrics at $t=0.5$.

\begin{table}
\centering
\caption{Comparison of Stochastic Latent Twin, Stochastic Latent Twin with mass conformity and Latent Neural Stochastic Flow for $t = 0.5$.  Metrics correspond to respective model approximations $\Psi_{0, 0.5}(x_0)$ and are evaluated according to \Cref{app:burgers-eval}. Values are means $\pm$ one standard error over the 32 initial conditions.}
\label{tab:Burg-metrics}
\begin{tabular}{lcccc}
\toprule
Metric & SLT & SLT with conformity & Latent NSF & Baseline \\
\midrule
Median error         & $0.0894 \pm 0.0044$ & $\mathbf{0.0801 \pm 0.0047}$ & $0.1873 \pm 0.0117$ & $0.0294 \pm 0.0013$ \\
SWD                  & $0.0983 \pm 0.0037$ & $\mathbf{0.0872 \pm 0.0038}$ & $0.1887 \pm 0.0109$ & $0.0414 \pm 0.0011$ \\
Spread ratio         & $0.935 \pm 0.013$   & $\mathbf{0.945 \pm 0.009}$   & $1.108 \pm 0.011$   & $1.004 \pm 0.007$   \\
$\mathcal C(u)$ defect & $1.80 \times 10^{-2}$ & {$\mathbf{4.86 \times 10^{-8}}$} & $3.32 \times 10^{-2}$ & $3.53 \times 10^{-17}$ \\
\bottomrule
\end{tabular}
\end{table}

For all methods, data are standardized with training-set statistics, all metrics are computed in physical units and the final models were used in evaluation.

\subsection{Evaluation and diagnostics}\label{app:burgers-eval}

For evaluation, we take $32$ held-out initial conditions and simulate $256$ independent trajectories for each time $t \in \{0.1, 0.25, 0.5, 0.75, 1\}$, keeping the resulting states for each time. This provides samples of the conditional law of the simulated system at time $t$ given the initial condition. As before, the time steps are no larger than $2.5\times10^{-3}$ and chosen so that the simulation lands exactly on each $t$.  To provide a baseline for the metrics, we repeat this simulation with new noise realizations.

For each of the $32$ initial states $u_0$, all models are evaluated on its one-shot prediction of $U_t \mid U_0 = u_0$, drawing 256 samples from its predicted distribution at each specified $t$.  In each of the following metrics, the 256 reference samples are compared to the 256 model samples and aggregated into the reported metrics as follows.
\begin{itemize}
    \item \textbf{Median error:} For each of the initial conditions and specified time points, we compute the root-mean-squared difference over the grid between the pointwise medians of the 256 reference and 256 model samples and average it over the 32 reference initial conditions. The baseline is computed similarly between the reference samples and the second realization of reference samples.

    \item \textbf{sliced-Wasserstein distance: } The sliced-Wasserstein distance is described \Cref{sec:numerics}. Here it is computed between the reference and model samples using $1{,}000$ random directions, compared with $128$ in the Swiss-roll and pendulum experiments. This is then averaged over the 32 initial conditions.   The baseline reports the sliced-Wasserstein distance computed between the reference samples and second independent draw of the reference samples again averaged over the 32 initial conditions. 

    \item \textbf{$\mathcal C (u)$ defect:} For each initial state $u_0$ and time $t$, we compute the root-mean-squared error between $\mathcal C (u_t)$ of the 256 model samples and $\mathcal C (u_0)$,and average over the 32 initial states.  The baseline computes the same value for initial states $u_0$ and reference samples drawn from the second realization.
        
    \item \textbf{Spread ratio: } For each initial condition and time, we compute the pointwise variance of each grid cell across the model's 256 samples. This is then averaged over the cells and square rooted to provide the spread. This procedure is repeated for the reference samples. The ratio of the model spread to the reference spread is averaged over the 32 initial conditions and reported.  The baseline is reported as the ratio of the spread of the second independent reference ensemble to that of the reference ensemble.  This metric provides a measure of the calibration of the variance of the samples.  Although the metric does not take into account the location of the variance, a value close to 1 indicates that, on average, the model samples have similar spread to the reference samples.  Values below 1 and over 1 indicate overconfidence and underconfidence, respectively.

\end{itemize}

\end{document}